\documentclass[times,sort&compress,3p]{elsarticle}
\journal{Journal of Multivariate Analysis}
\usepackage[labelfont=bf]{caption}

\usepackage{amsmath,amsfonts,amssymb,amsthm,booktabs,color,epsfig,graphicx,hyperref,url}

\theoremstyle{plain}
\newtheorem{theorem}{Theorem}

\newtheorem{proposition}{Proposition}
\newtheorem{lemma}{Lemma}
\newtheorem{corollary}{Corollary}

\theoremstyle{definition}
\newtheorem{definition}{Definition}
\newtheorem{remark}{Remark}

\hypersetup{
  pdftitle={rank-corr-SE},
  pdfauthor={Ye Lu}
}

\newcommand{\RR}{\mathbb R}

\newcommand{\PP}{\mathrm{Pr}}
\newcommand{\EE}{\mathbb E}

\newcommand{\sgn}{\operatorname{sgn}}
\newcommand{\Cov}{\operatorname{Cov}}

\newcommand*\diff{\mathop{}\!\mathrm{d}} 
\newcommand{\given}{\, | \,}

\begin{document}

\begin{frontmatter}

\title{
  Kendall's tau and Spearman's rho for normal location-scale and
  skew-normal scale mixture copulas
}

\author[1]{Ye Lu}

\address[1]{School of Economics, University of Sydney, Australia}


\begin{abstract}
  We derive explicit formulas for Kendall's tau and Spearman's rho for two broad
  classes of asymmetric copulas: normal location-scale mixture copulas and
  skew-normal scale mixture copulas. These classes encompass widely used
  specifications, including the normal scale mixture, skew-normal, and various
  skew-$t$ copulas, as special cases. The derived formulas establish functional
  mappings from copula parameters to rank correlation coefficients,
  and we investigate and compare how asymmetry parameters influence 
  rank correlation properties and drive departures from the elliptically symmetric
  case within these two classes.
  A notable finding is that the introduction of asymmetry in normal location-scale mixture
  copulas restricts the attainable range of rank correlations from the standard
  $[-1,1]$ interval, which is observed under elliptical symmetry, to a strict subset
  of $[-1,1]$. In contrast, the entire interval $[-1,1]$ remains attainable for
  skew-normal scale mixture copulas.
\end{abstract}


\end{frontmatter}

\section{Introduction\label{sec:1}}

Rank correlations are measures of association between two variables that play a
crucial role in many applications. Compared with Pearson's linear correlation,
rank correlations are invariant under nonlinear, rank-preserving
transformations of the measurement scale and are robust to outliers and violations of
normality. Among the various rank correlation measures, Kendall's tau, denoted by
$\tau$, and Spearman's rho, denoted by $\rho_S$, are the most commonly used.

A key concept underlying rank correlations is the \emph{copula}. A
$d$-dimensional copula is a distribution function on $[0,1]^d$ with univariate
marginals that are uniform on $[0,1]$.
By Sklar's theorem, any random vector with continuous marginal distribution
functions (hereafter referred to as \emph{continuous marginals})
admits a unique copula that fully characterizes its dependence structure. For
background on copula theory, see, for example, \cite{Nelsen2006} and 
\cite[][Chapter~5]{McNeilFreyEmbrechts2005}.
An important consequence of Sklar's theorem is that, for a pair of random variables
with continuous marginals, their rank correlations 
depend solely on the associated copula rather than on their marginal distributions.
Specifically, if $C$ denotes the copula of such a pair, Kendall's tau and
Spearman's rho can be expressed as (see, e.g., Prop. 5.29 in
\cite{McNeilFreyEmbrechts2005}):
\begin{equation*}
\tau = 4\int_0^1\int_0^1C(u_1,u_2)\diff C(u_1,u_2) -1
\quad\mbox{and}\quad
\rho_S = 12\int_0^1\int_0^1[C(u_1,u_2)-u_1u_2]\diff u_1\diff u_2.
\end{equation*}
These representations show that rank correlations may be viewed as moments of a
bivariate copula and can therefore be used to fit copula models to data under the
assumption of continuous marginals.

This approach is particularly well developed for \emph{elliptical copulas}, which are
defined as the unique copula implied by elliptical distributions with continuous
marginals via Sklar's theorem. Such copulas are characterized by a correlation matrix
together with the generating variable of the underlying elliptical distribution
(see, e.g., \cite[][Section~2]{Fang2002meta} and \cite[][Section~2]{KluppelbergKuhn2009}).
For a bivariate elliptical copula associated with a strictly positive and absolutely
continuous generating variable, Kendall's tau and the pseudo-correlation
coefficient $\rho$ satisfy:
\begin{equation}\label{eq:tau}
\tau = \frac2\pi\arcsin\rho.
\end{equation}
This relationship is well-established in the literature 
(e.g., \cite[][Theorem 2]{KluppelbergKuhn2009}, \cite[][Theorem 3.1]{Fang2002meta}
and \cite[][Theorem 2]{Lindskog2003kendall})
and is commonly used to estimate $\rho$ based on an estimate of $\tau$ by inversion. 
For the bivariate Gaussian copula, there is also a well-known relationship between
Spearman's rho and $\rho$:
\begin{equation}\label{eq:S-Gaussian}
  \rho_S = \frac6\pi\arcsin\frac\rho2.
\end{equation}
More recently, \cite{HeinenValdesogo2020} show that Spearman's rho for normal scale 
mixtures can be expressed as a correlation mixture of Spearman's rho in
the Gaussian case. 

In contrast to the extensive applications of Kendall's tau and Spearman's rho
in elliptically symmetric models, comparatively little is known about these 
rank correlations when asymmetry is introduced.
To the best of the author's knowledge, the only notable contribution in this
direction is \cite{HeinenValdesogo2022}, who derive formulas for these rank
correlations for the bivariate skew-normal distribution.
More generally, asymmetry can be introduced into elliptical distributions or their
subclasses through various mechanisms while preserving certain aspects of elliptical
structure. For a comprehensive overview, see the survey of various skew-elliptical
distributions in \cite[][Chapter 3]{Genton2004} and the discussion on skewed normal
mixture models in \cite[][Section 5.3.3]{McNeilFreyEmbrechts2005}.

Among the many proposed constructions of asymmetric models, two 
classes have attracted particular attention: \emph{normal location-scale mixtures} and
\emph{skew-normal scale mixtures}.
Their popularity stems from their generality, wide applicability, and tractable
stochastic representations based on multivariate normal variables.
Notably, both classes nest the \emph{normal scale mixture}---an
important subclass of elliptical distributions---as a benchmark special case.

Normal location-scale mixtures are obtained by mixing multivariate
normal distributions over both location and scale. In this class, the location-mixture
coefficients serve as skewness parameters; the model reduces to the elliptically symmetric
benchmark when these parameters are zero.
A notable subclass is the generalized hyperbolic (GH) distributions; see, for example,
\cite{BlaesildJensen1981}. A further special case, known as the GH skew-$t$
distribution, was highlighted by \cite{DemartaMcNeil2005} and has since found
extensive applications in economics and finance \citep{AasHaff2006,
Christoffersen2012, CrealTsay2015, LucasSchwaabZhang2017,OhPatton2023}. 

Skew-normal scale mixtures arise from mixing multivariate skew-normal
distributions over scale. These models reduce to the normal
scale mixture benchmark when the skewness parameters in the underlying skew-normal
distribution are set to zero. 
Significant special cases within this class include the skew-normal distribution
itself and the so-called AC skew-$t$ distribution, named after
\cite{AzzaliniCapitanio2003}, whose applications have been advanced by
\cite{KolloPettere2010}, \cite{Azzalini2014}, \cite{Yoshiba2018}, and
\cite{DengSmithManeesoonthorn2024}, among others. 

In this paper, we derive explicit formulas for Kendall's tau and Spearman's rho for
the these two asymmetric model classes described above. 
Since both classes
possess continuous marginals, the corresponding rank correlations depend solely on the
implied copula classes via Sklar's theorem. 
We demonstrate that, for both models, these rank correlations admit convenient
representations as expectations of mixtures of zero-mean normal cumulative distribution
functions (cdf's), yielding clear functional mappings from copula
parameters to each rank correlation measure.

Applying the derived formulas, we examine how skewness parameters influence both
rank correlations and drive departures from the elliptically symmetric benchmark case.
For both model classes, Kendall's tau and Spearman's
rho are symmetric with respect to the components of the skewness vector and invariant
under its sign change.
Moreover, under a single-skewness specification, both rank correlations are odd
functions of the pseudo-correlation coefficient $\rho$, and increasing asymmetry
reduces their magnitude.
This odd-function property implies that the sign of $\rho$
continues to determine the sign of $\tau$ and $\rho_S$, preserving the behavior
observed in the elliptically symmetric case.

From a modeling perspective, this
feature is noteworthy: the skewness parameter scales the magnitude of the rank
correlations without altering the symmetry between positive and negative dependence
regimes. 
This is advantageous for applications requiring a consistent interpretation of
$\rho$ across dependence regimes while allowing the strength of rank
dependence to vary. Conversely, the single-skewness specification
may be inappropriate for application demanding genuine asymmetry, where the
dependence structure itself differs between positive and negative regimes.

We also identify significant differences in the rank correlation properties of
the two model classes. For example, under the equi-skew setting, where skewness
parameters take identical values, the impact of asymmetry differs in direction
between the two classes. Specifically, increasing the equi-skewness parameter
\emph{increases} the rank correlations for normal location-scale mixture copulas but
\emph{decreases} them for skew-normal scale mixture copulas.

The most notable difference between these two classes concerns the attainable ranges of 
Kendall's tau and Spearman's rho under asymmetry. For normal location-scale mixture
copulas,
these ranges become \emph{strict} subsets of $[-1,1]$.
Specifically, we show that $\tau=\rho_S=-1$ if and only $\rho=-1$ and the skewness
parameters are negatives of each other; and $\tau=\rho_S=1$ if and only $\rho=1$ and
the skewness parameters coincide. 
This behavior sharply contrasts with the elliptical case,
where the full interval $[-1,1]$ is always attainable and both rank correlations
reach $-1$ or $1$ whenever $\rho$ equals $-1$ or $1$.
Conversely, these properties for elliptically symmetric case are well preserved when
asymmetry is introduced through skew-normal scale mixture copulas, which continue to span
the entire $[-1,1]$ interval. 

The rest of the paper is organized as follows. After introducing notation, we
formally define the two classes of asymmetric copulas implied by the normal
location-scale and skew-normal scale mixture distributions in
Section~\ref{sec:MN-MSN}.
Sections~\ref{sec:rank-corr-MN} and \ref{sec:rank-corr-MSN} then present the main
results on the expressions and properties of Kendall's tau and Spearman's rho for
these two copula classes, together with a thorough comparison.
Section~\ref{sec:invertibility} addresses the challenges of rank-based estimation for
these two copula classes shows how the derived formulas can be used to study the
invertibility of the mapping from copula parameters to rank correlations under
specific parameter configurations.
Finally, Section~\ref{sec:conclusion} concludes the paper, and
the mathematical proofs are provided in Section~\ref{sec:proofs}.

Before concluding this section, we clarify the notation used throughout the paper.
We use $\mathrm P$ to denote a Pearson correlation matrix, and
$\varrho$ the \emph{correlation operator}. When the correlation operator operates on
a $d\times d$ covariance matrix $\Sigma=[\sigma_{ij}]$, written $\varrho(\Sigma)$, it
returns the corresponding correlation matrix. When $\varrho$ operates on
$\rho\in(-1,1)$, written $\varrho(\rho)$, it returns a $2\times 2$ correlation matrix
with off-diagonal entries $\rho$. 
Moreover, we use $\mathrm N$ to denote the normal distribution; $\mathrm{Ga}(a,b)$
and $\mathrm{IG}(a,b)$ the gamma and inverse-gamma distributions with shape parameter
$a$ and rate parameter $b$.
We denote by $\phi$ and $\Phi$ the pdf and cdf of the univariate standard normal
distribution, respectively. The pdf of the univariate standard skew-normal
distribution with skewness parameter $\alpha\in\RR$ is
$\phi^{\mathrm{s}}(x;\alpha)\coloneqq2\phi(x)\Phi(\alpha x)$. For a $d\times d$
covariance matrix $\Sigma$, we denote by $\phi_d(x;\Sigma)$ and $\Phi_d(x;\Sigma)$,
for $x\in\RR^d$, the pdf and cdf of $\mathrm N(0,\Sigma)$. When $d=2$, we use the
shorthand notation $\phi_2(x;\rho)$ and $\Phi_2(x;\rho)$ for
$\phi_2(x;\varrho(\rho))$ and $\Phi_2(x;\varrho(\rho))$, respectively. 
Lastly, the sign $=_d$ means both sides of the equality have the same distribution,
and $X\perp Y$ means the two random variables $X$ and $Y$ are independent.

\section{Two classes of asymmetric copulas}
\label{sec:MN-MSN}

We define two classes of asymmetric copulas implied by the normal location-scale
mixture and skew-normal scale mixture distributions in Sections~\ref{sec:lsmn} and
\ref{sec:msn}, respectively.

\subsection{Normal location-scale mixture copulas} \label{sec:lsmn}

Given $\mu,\beta\in\RR^d$, a $d\times d$ positive definite matrix $\Sigma$, and a
univariate distribution $F$ on $(0,\infty)$, a $d$-dimensional random vector $X$ is
said to be a \emph{normal location-scale mixture}, denoted
$X\sim\mathrm{MN}(\mu,\Sigma,\beta,F)$, if $X$ can be expressed as 
\[
X=\mu+W\beta+\sqrt{W}Z,
\]
where $W$ and $Z$ are independent, $W\sim F$, and $Z\sim\mathrm N(0,\Sigma)$.
A normal location-scale mixture is also referred to as a normal \emph{mean-variance
mixture} with a linear mean function; see, for example, Section 3.2.2 in
\cite{McNeilFreyEmbrechts2005}.
The distribution of $X$ reduces to the normal scale mixture---hence becomes
elliptically symmetric---if and only if $\beta=0$. Consequently, the vector $\beta$
acts as the \emph{skewness/asymmetry} parameter. 

A notable subclass of this class of distributions is the \emph{generalized
hyperbolic} (GH) distribution introduced by \cite{Barndorff1977} for the univariate
case and extended by \cite{BlaesildJensen1981} for the multivariate case. This
subclass is derived by selecting the mixing distribution $F$ to be a
\emph{generalized inverse Gaussian} (GIG) distribution \citep{BarndorffHalgreen1977},
denoted by $\mathrm N^-(\lambda,\chi,\psi)$ with parameters $\lambda\in\RR,
\chi,\psi\ge0$. The pdf of the GIG distribution is given by
\begin{equation}\label{eq:GIG-density}
f(x;\lambda,\chi,\psi)
=\frac{(\psi/\chi)^{\lambda/2}}{2 K_\lambda(\sqrt{\psi\chi})}x^{\lambda-1}\exp\left(-\frac{\psi x+\chi/x}2\right) \qquad (x>0),
\end{equation}
where $K_\lambda$ is a modified Bessel function of the second kind with index
$\lambda$, defined as
$K_\lambda(x) = (\pi/2) [I_{-\lambda}(x)-I_\lambda(x)]/\sin\lambda\pi$ for $x>0$,
where $I_\lambda(x)=\sum_{m=0}^\infty[m!\,\Gamma(m+\lambda+1)]^{-1}(x/2)^{2m+\lambda}$
is the modified Bessel function of the first kind.
Historically, $K_\lambda$ has also been referred to as 
the modified Bessel function of the third kind 
(see, e.g., Section A.25 in \cite{McNeilFreyEmbrechts2005} and
10.2.15 in \cite{AbramowitzStegun1965}),
the modified Hankel function, and the MacDonald function (see, e.g., \cite{Kropac1982}).
Note that the GIG density contains the gamma and inverse gamma densities as special
limiting cases, corresponding to $\chi=0$ and $\psi=0$, respectively.
In these cases, \eqref{eq:GIG-density} must be interpreted as a limit.

The GH family is highly flexible, encompassing many known distributions as special
cases. For example, when $\lambda=(d+1)/2$, it corresponds to a $d$-dimensional
hyperbolic distribution. When $\lambda>0, \psi>0$ and $\chi=0$, the distribution is
known as the variance-gamma, generalized Laplace, or Bessel function distribution.
Setting $\lambda=-1/2$ yields the normal inverse gamma (NIG) distribution.
Particular attention is given to the case where $\lambda=-\nu/2, \chi=\nu$ for some
$\nu>0$ and $\psi=0$. Under this parameter setting, the mixing distribution is
$\mathrm N^-(-\nu/2,\nu,0)=_d\mathrm{IG}(\nu/2,\nu/2)$, and the GH family
reduces to a class of asymmetric or skew $t$ distributions, known as the \emph{GH
skew-$t$} distribution. 

Since standardizing the marginals involves strictly increasing transformations that
do not alter the copula, the copula of $\mathrm{MN}(\mu,\Sigma,\beta,F)$ is identical
to that of $\mathrm{MN}(0,\varrho(\Sigma),\mathrm{diag}(\Sigma)^{-1/2}\beta,F)$. We
define the \emph{normal location-scale mixture copula} as follows.

\begin{definition} \label{def:MN-copula}
Let $d\geq2$ be a positive integer. Given $\beta\in\RR^d$, a $d\times d$ correlation
matrix $\mathrm P$ and a univariate distribution $F$ on $(0,\infty)$, the normal
location-scale mixture copula $C_{\mathrm{mn}}(\mathrm P,\beta,F)$ is the copula of
the normal location-scale mixture $\mathrm{MN}(0,\mathrm P,\beta,F)$. In the case
$d=2$, we denote it by $C_{\mathrm{mn}}(\rho,\beta,F)$.
\end{definition}

When the skewness vector $\beta=0$, the copula
$C_{\mathrm{mn}}(\mathrm P,0,F)$ reduces to the normal scale mixture copula, an
important class of elliptical copulas, as considered in Example~1 of
\cite{KluppelbergKuhn2009}.
If the mixing distribution $F$ in Definition~\ref{def:MN-copula} is
$\mathrm{IG}(\nu/2,\nu/2)$ for some $\nu>0$, the resulting copula corresponds to the
so-called GH skew-$t$ copula with degrees of freedom $\nu$.
The bivariate GH skew-$t$ copula is denoted by $C_{\mathrm{GHt}}(\rho,\beta,\nu)$.

\subsection{Skew-normal scale mixture copulas} \label{sec:msn}

A $d$-dimensional \emph{normalized} skew-normal distribution $\mathrm{SN}(0,\mathrm
P,\alpha)$ with a correlation matrix $\mathrm P$ and skewness vector $\alpha\in\RR^d$
is defined by its pdf $2\,\phi_d(x;\mathrm P)\,\Phi(\alpha^\top x)$, for $x\in\RR^d$.
One of the attractive features of the skew-normal family is that it admits a variety
of stochastic representations; see Section~5.1.3 in \cite{Azzalini2014}. For example,
a $d$-dimensional random vector $X\sim \mathrm{SN}(0,\mathrm P,\alpha)$ admits the
representation
\begin{equation}\label{eq:SN-conditioning-0}
X = (Z\given Z_0>0),
\end{equation}
where $Z$ is a $d$-dimensional normal random vector and $Z_0$ is a univariate normal
random variable such that
\begin{equation}\label{eq:SN-conditioning}
  \begin{bmatrix} Z\\Z_0 \end{bmatrix}
  \sim
  \mathrm N_{d+1}(0, \mathrm R),
  \qquad
  \mathrm R = \begin{bmatrix} \mathrm P & \delta \\ \delta^\top & 1 \end{bmatrix},
  \qquad
  \delta =  \frac{\mathrm P\alpha}{\sqrt{1+\alpha^\top\mathrm P\alpha}}.
\end{equation}

It is important to note that although the skew-normal distribution is closed under
marginalization, the marginal distribution of the $i$th component of
$\mathrm{SN}(0,\mathrm P,\alpha)$ is \emph{not} simply $\mathrm{SN}(0,1,\alpha_i)$,
where $\alpha_i$ denotes the $i$th element of $\alpha$. An alternative
parameterization with simpler parameter transformation under marginalization is
motivated by the definition of $\delta$ in \eqref{eq:SN-conditioning} and is
summarized in the following remark.

\begin{remark} \label{remark:delta-prameterization}
  The $d$-dimensional normalized skew-normal distribution can be alternatively
  parameterized by $(\mathrm P,\delta)$, where
  $\delta=(\delta_1,\ldots,\delta_d)^\top\in\RR^d$ serves as the skewness parameter and
  satisfies $\delta^\top\mathrm P^{-1}\delta<1$ and $|\delta_i|<1$ for all
  $i\in\{1,\ldots,d\}$. 
  There is a one-to-one correspondence between $(\mathrm P,\alpha)$ and $(\mathrm
  P,\delta)$:
  \begin{equation}\label{eq:del-alp}
    \delta = \frac{\mathrm P\alpha}{\sqrt{1+\alpha^\top\mathrm P\alpha}},
    \qquad
    \alpha = \frac{\mathrm P^{-1}\delta}{\sqrt{1-\delta^\top\mathrm P^{-1}\delta}}.
  \end{equation}
  Clearly, $\alpha=0$ if and only if $\delta=0$, and changing the sign of one skewness
  vector induces the same sign change in the other. Under the $(\mathrm P,\delta)$
  parameterization, the $i$th marginal distribution of $\mathrm{SN}(0,\mathrm
  P,\delta)$ is simply $\mathrm{SN}(0,1,\delta_i)$.
  However, applying the second mapping in \eqref{eq:del-alp} to the univariate case
  shows that the marginal skewness parameter under the $(\mathrm P,\alpha)$
  parameterization is $\alpha_i^\dagger=\delta_i(1-\delta_i^2)^{-1/2}$, which is not
  equal to the $i$th element of $\alpha$.
\end{remark}

Given $\mu,\alpha\in\RR^d$, a $d\times d$ covariance matrix $\Sigma$, and a
univariate distribution $F$ on $(0,\infty)$, a random vector $X$ is said to be a
\emph{skew-normal scale mixture}, denoted $X\sim\mathrm{MSN}(\mu,\Sigma,\alpha,F)$, if
$X$ can be expressed as
\[
X = \mu+\sqrt WZ,
\]
where $W$ and $Z$ are independent random variables, with $W\sim F$ and
$Z=_d\operatorname{diag}(\Sigma)^{1/2}\mathrm{SN}(0,\varrho(\Sigma),\alpha)$.
Clearly, the distribution of $X$ reduces to the normal scale mixture, and hence
becomes elliptically symmetric, if and only if the skewness vector $\alpha=0$.

Within this class, particular attention is given to an important subclass known as
the \emph{AC skew-$t$} distribution, named after \cite{AzzaliniCapitanio2003} and was
independently proposed by \cite{Gupta2003}.
The AC skew-$t$ distribution is obtained by choosing the mixing distribution $F$ to
be 
$\mathrm{IG}(\nu/2,\nu/2)$ with $\nu>0$.
For a comprehensive analysis of the properties of this class of multivariate skew-$t$
distributions, as well as a historical overview, see \citep[][Section 6.2]{Azzalini2014}.

Since standardizing the margins does not alter the copula, the copula of
$X\sim\mathrm{MSN}(\mu,\Sigma,\alpha,F)$, is the same as the copula of
$\operatorname{diag}(\Sigma)^{-1/2}(X-\mu)\sim\mathrm{MSN}(0,\varrho(\Sigma),\alpha,F)$.
We define the \emph{skew-normal scale mixture copula} as follows.

\begin{definition}\label{def:MSN-copula}
Let $d\geq2$ be a positive integer. Given $\alpha\in\RR^d$, a $d\times d$ correlation
matrix $\mathrm P$ and a univariate distribution $F$ on $(0,\infty)$, the skew-normal
scale mixture copula $C_{\mathrm{msn}}(\mathrm P,\alpha,F)$ is the copula of the
scale mixture of skew-normals $\mathrm{MSN}(0,\mathrm P,\alpha,F)$.
In the case $d=2$, we denote it by $C_{\mathrm{msn}}(\rho,\alpha,F)$. 
\end{definition}

As a degenerate special case, the skew-normal copula is nested in this class of
copulas, and the bivariate skew-normal copula is denoted by
$C_{\mathrm{sn}}(\rho,\alpha)$. If the mixing distribution $F$ is
$\mathrm{IG}(\nu/2,\nu/2)$ for some $\nu>0$, the resulting copula is the AC skew-$t$
copula with degrees of freedom $\nu$. Its bivariate form is denoted by
$C_{\mathrm{ACt}}(\rho,\alpha,\nu)$. 

Due to the correspondence between $\alpha$ and $\delta$ in \eqref{eq:del-alp},
$C_{\mathrm{msn}}(\mathrm P,\alpha,F)$ can be equivalently represented by
$C_{\mathrm{msn}}(\mathrm P,\delta,F)$. The alternative parameterization 
by $\delta$ is instrumental for analyzing the rank correlations for bivariate
skew-normal scale mixture copulas, as shown in Section~\ref{sec:rank-corr-MSN}.
The following remark specializes Remark~\ref{remark:delta-prameterization} to the
bivariate setting.

\begin{remark} \label{remark:delta-prameterization-copula}
  In the bivariate case, the skewness vectors $\alpha=(\alpha_1,\alpha_2)^\top$ and
  $\delta=(\delta_1,\delta_2)^\top$ are related as 
  \begin{equation}\label{eq:del-alp-i}
    \delta_i
    =\frac{\alpha_i+\rho\alpha_{-i}}{\sqrt{1+\alpha_1^2+\alpha_2^2+2\rho\alpha_1\alpha_2}},
    \quad\quad
    \alpha_i
    =\frac{(\delta_i-\rho\delta_{-i})\,/\sqrt{1-\rho^2}}{\sqrt{1-\rho^2-\delta_1^2-\delta_2^2+2\rho\delta_1\delta_2}}, 
  \end{equation}
  for $i\in\{1,2\}$ and $\rho\in(-1,1)$. Here, $-i$ denotes the index other than $i$.
  From \eqref{eq:del-alp-i}, it follows that $\alpha_1=\alpha_2$ if and only if
  $\delta_1=\delta_2$. However, if one of $\alpha_1$ and $\alpha_2$ is zero, it need
  not be the case that one of $\delta_1$ and $\delta_2$ is zero, and vice versa.
  Moreover, $\delta_1$ and $\delta_2$ are constrained by
  $\delta_1^2+\delta_2^2-2\rho\delta_1\delta_2<1-\rho^2$.
\end{remark}

\section{Kendall's tau and Spearman's rho for normal location-scale mixture copulas}
\label{sec:rank-corr-MN}

In this section, we present the formulas and properties of Kendall's tau and
Spearman's rho for bivariate normal location-scale mixture copulas. Then, we
specialize the results to the GH skew-$t$ copula, an important and widely applicable
special case. In the last subsection, we discuss the equi-skew and single-skew
settings.

\subsection{Formulas for Kendall's tau and Spearman's rho}

The following theorem gives the formulas for Kendall's tau and Spearman's rho for the bivariate normal location-scale mixture copula. In the theorem, $\Phi_2(x_1,x_2;\rho)$ denotes the cdf of the bivariate normal distribution $\mathrm N(0,\varrho(\rho))$.

\begin{theorem}\label{thm:MN-tau-S}
  Suppose $X_1$ and $X_2$ are two random variables with continuous cdf's and copula
  $C_{\mathrm{mn}}(\rho,\beta,F)$, where $\rho\in[-1,1]$,
  $\beta=(\beta_1,\beta_2)\in\RR^2$, and $F$ is a univariate distribution on
  $(0,\infty)$. Kendall's tau of $X_1$ and $X_2$ is given by 
  \begin{align}\label{eq:MN-tau}
    \tau(\rho,\beta,F) = 4\,\EE\,\Phi_2(\beta_1V,\,\beta_2V;\rho) - 1,
  \end{align}
  where $V=(W^\star-W)/\sqrt{W^\star+W}$ with $W, W^\star\sim$ i.i.d.\,$F$, and
  Spearman's rho of $X_1$ and $X_2$ is given by
  \begin{align}\label{eq:MN-S}
    \rho_S(\rho,\beta,F) = 12\,\EE\,\Phi_2(\beta_1V_1,\,\beta_2V_2;\rho V_3) - 3, 
  \end{align}
  where 
  $V_i = (W_i-W_3)/\sqrt{W_i+W_3}$,\, for $i=1,2$, and
  $V_3 = W_3/\sqrt{(W_1+W_3)(W_2+W_3)}$, with $W_1,W_2,W_3\sim$ i.i.d.\,$F$.
\end{theorem}

Both Kendall's tau and Spearman's rho, as described in \eqref{eq:MN-tau} and
\eqref{eq:MN-S}, are expectations of mixtures of the bivariate normal cdf, where the
distributions of the mixing variables $V$ and $(V_1,V_2,V_3)$ are determined by the
distribution $F$. Some properties of the distributions of $V$ and $(V_1,V_2,V_3)$ are
summarized in the following remark, and will be used in establishing the properties
of these rank correlations in Section~\ref{sec:lsmn-properties}.

\begin{remark}\label{remark:MN-mix-dist}
  By constructions of $V$ and $(V_1,V_2,V_3)$ in Theorem~\ref{thm:MN-tau-S}, we have
  (i) $V=_d-V$, (ii) $(V_1,V_2,V_3)=_d(V_2,V_1,V_3)$, and (iii)
  $V_i=_d-V_i$ for $i\in\{1,2\}$.
\end{remark}

To highlight the role played by the skewness vector $\beta$, we may relate the
formulas for $\tau$ and $\rho_S$ to their counterparts under symmetry. In particular,
when $\beta=0$, the formula \eqref{eq:MN-tau} simplifies to
\begin{equation}\label{eq:tau-rho-0-F}
  \tau(\rho,0,F)=4\,\Phi_2(0,0;\rho)-1=\frac2\pi\arcsin\rho.
\end{equation}
Equation~\eqref{eq:tau-rho-0-F} is well expected and matches formula \eqref{eq:tau}. 
For Spearman's rho, we have 
\begin{equation}\label{eq:S-rho-0-F}
\rho_S(\rho,0,F)=12\,\EE\,\Phi_2(0,0;\rho V_3)-3 =\frac6\pi\,\EE\arcsin\rho V_3,
\end{equation}
which, unlike Kendall's tau, depends on the mixing distribution $F$ through
$V_3$. The formula in \eqref{eq:S-rho-0-F} was previously derived by
\cite{HeinenValdesogo2020}. We note that the bivariate Gaussian copula can be considered
as a scale normal mixture copula with $F$ degenerated at $1$;
in that case, $V_3=1/2$ almost surely, and $\rho_S=(6/\pi)\arcsin(\rho/2)$ as shown
in \eqref{eq:S-Gaussian}.

\subsection{Properties of Kendall's tau and Spearman's rho} \label{sec:lsmn-properties}

Since Kendall's tau and Spearman's rho given in Theorem~\ref{thm:MN-tau-S} are both
expressed in terms of $\Phi_2(x_1,x_2;\rho)$, their properties can be analyzed using
the characteristics of this bivariate normal cdf. In particular, for
$(x_1,x_2)\in\RR^2$ and $\rho\in(-1,1)$, two useful representations of
$\Phi(x_1,x_2;\rho)$ are given as follows
(see, e.g., Equation 4.6 in \cite{Owen1956} for the representation
\eqref{eq:Phi2-rep}, and Equation B.18 in \cite{Azzalini2014} for the
representation \eqref{eq:Phi2-rep-2}):
\begin{align}
\Phi_2(x_1,x_2;\rho)
&=
\Phi(x_1)\Phi(x_2)+\int_0^\rho\phi_2(x_1,x_2;r) \diff r, \label{eq:Phi2-rep}\\
\Phi_2(x_1,x_2;\rho)
&= 
\frac12\,[\Phi(x_1)+\Phi(x_2)] - a(x_1,x_2) 
-T\left(x_1,\,\frac{x_2-\rho x_1}{x_1\sqrt{1-\rho^2}}\right)
-T\left(x_2,\,\frac{x_1-\rho x_2}{x_2\sqrt{1-\rho^2}}\right), \label{eq:Phi2-rep-2}
\end{align}
where 
$a(x,y) = 1\{\sgn(x)+\sgn(y)<0\}/2$,
and $T$ is the Owen's T function \citep{Owen1956} defined by 
$T(h,a)=(2\pi)^{-1}\int_0^a(1+x^2)^{-1}\exp\{-h^2(1+x^2)/2\}\diff x$ for $h,a\in\RR$.
For two independent standard normal random variables $Z_1$ and $Z_2$, we have
$T(h,a)=\PP\{Z_1>h,0<Z_2<aZ_1\}$ for $h\geq0$ and $a\geq0$. Moreover,
$T(-h,a)=T(h,a)$ and $T(h,-a)=-T(h,a)$. 

The decomposition in~\eqref{eq:Phi2-rep-2} is particularly useful for analyzing the properties of Spearman's rho, whose expression in \eqref{eq:MN-S} involves more intricate mixing structure. We also note the following partial derivatives: 
\begin{equation} \label{eq:dPhi2}
  \frac{\partial}{\partial x_1}\,\Phi_2(x_1,x_2;\rho)  
  =
  \phi(x_1)\Phi\left(\frac{x_2-\rho x_1}{\sqrt{1-\rho^2}}\right),
  \qquad
  \frac{\partial}{\partial\rho}\,\Phi_2(x_1,x_2;\rho)  
  =
  \phi_2(x_1,x_2;\rho).
\end{equation}
Drawing on Theorem~\ref{thm:MN-tau-S}, Remark~\ref{remark:MN-mix-dist}, and the
properties of $\Phi_2(\,\cdot\,;\rho)$, we obtain several properties of the rank
correlations of the normal location-scale mixture copula and summarize them in the
following proposition.

\begin{proposition} \label{prop:MN-tau-S-properties}
  Kendall's tau and Spearman's rho for the bivariate normal location-scale mixture
  copula $C_{\mathrm{mn}}(\rho,\beta,F)$ possess the following properties:
  \smallskip

  \noindent \emph{(i)}
  $\tau(\rho,(\beta_1,\beta_2),F)=\tau(\rho,(\beta_2,\beta_1),F)$,
  $\rho_S(\rho,(\beta_1,\beta_2),F)=\rho_S(\rho,(\beta_2,\beta_1),F)$.
  \smallskip
  
  \noindent \emph{(ii)}
  $\tau(\rho,\beta,F)=\tau(\rho,-\beta,F)$
  and
  $\rho_S(\rho,\beta,F)=\rho_S(\rho,-\beta,F)$.
  \smallskip
  
  \noindent \emph{(iii)}
  $\partial\tau/\partial\rho > 0$
  and
  $\partial\rho_S/\partial\rho>0$.
  \smallskip 
  
  \noindent \emph{(iv)}
  If $\rho=0$, then $\sgn\tau = \sgn\rho_S = \sgn \beta_1\beta_2$.
  \smallskip 
  
  \noindent \emph{(v)}
  $\tau=\rho_S = -1$ if and only if $\rho=-1$ and $\beta_1=-\beta_2$.
  \smallskip 
  
  \noindent \emph{(vi)}
  $\tau=\rho_S = 1$ if and only if $\rho=1$ and $\beta_1=\beta_2$.
\end{proposition}

Properties~(i) and (ii) show that the rank correlations are symmetric
in the two components and invariant to a sign change of the skewness vector.
Note that for two random variables $X_1$ and $X_2$, the rank correlation between
$X_1$ and $X_2$ is always equal to that between $X_2$ and $X_1$ and that between
$-X_1$ and $-X_2$.
Here, if $(X_1,X_2)\sim\mathrm{MN}(0,\varrho(\rho),(\beta_1,\beta_2),F)$, then 
$(X_2,X_1)\sim\mathrm{MN}(0,\varrho(\rho),(\beta_2,\beta_1),F)$
and $(-X_1,-X_2)\sim\mathrm{MN}(0,\varrho(\rho),-(\beta_1,\beta_2),F)$,
so the equalities in (i) and (ii) follow directly from these properties.

Property~(iii) states that both rank correlations are strictly increasing in 
$\rho$, which aligns with the pattern observed for elliptical copulas and confirms
that, even in the presence of asymmetry, increasing $\rho$ strengthens concordance
between the two marginals. 
Property (iv) implies that, when $\rho=0$, the signs of $\tau$ and $\rho_S$ are
determined solely by the signs of the skewness parameters.
If $\beta_1$ and $\beta_2$ are both nonzero and have the same sign, the common
location-mixture term $W$ shifts the margins in the same direction, inducing positive
concordance and hence $\tau>0$ and $\rho_S>0$.
If they have opposite signs, $W$ shifts the margins in opposite directions, resulting
in negative concordance. Only when at least one skewness parameter vanishes do we
obtain $\tau=\rho_S=0$ at $\rho=0$. 
Thus, unlike in the elliptically symmetric case, zero pseudo-correlation does
\emph{not} in general imply zero rank correlation.
Moreover, Fig.~\ref{fig:GHt_tau_S} shows that with nonzero skewness, zero rank
correlation need not correspond to zero pseudo-correlation either.

Properties~(v) and (vi) are particularly revealing: unlike in the elliptically
symmetric case, the rank correlations for the normal location-scale mixture copula
$C_{\mathrm{mn}}(\rho,\beta,F)$ do \emph{not} automatically attain the upper bound $1$ or
lower bound $-1$ when $\rho=1$ or $-1$. Attainment of the upper or lower bound
now depends on the skewness parameters:
when $\rho=1$, we must have $\beta_1=\beta_2$ in order to obtain $\tau=\rho_S=1$;
when $\rho=-1$, we must have $\beta_1+\beta_2=0$ to obtain $\tau=\rho_S=-1$.%

\begin{remark} \label{remark:MN-attainment}
Since the only choice of $(\beta_1,\beta_2)$ that satisfies both conditions
$\beta_1=\beta_2$ and $\beta_1=-\beta_2$ in
Proposition~\ref{prop:MN-tau-S-properties}~(v) and (vi)
is $\beta_1=\beta_2=0$, the attainable range of rank correlations within this copula
class is always a proper subset of $[-1,1]$ whenever at least one of the
skewness parameters is nonzero.
The maximal and minimal values of the rank correlations depend on the skewness
parameters and the mixing distribution, and are given by
$\tau_{\mathrm{min}}(\rho,\beta,F) = \tau(-1,\beta,F)$ and
$\tau_{\mathrm{max}}(\rho,\beta,F) = \tau( 1,\beta,F)$, 
with analogous expressions for Spearman's rho.
\end{remark}

It is worth recalling that for two random variables with continuous cdf's, their 
Kendall's tau and Spearman's rho equal $1$ (resp. $-1$) if and only if their copula is
the comonotonicity copula, i.e., the Fr\'echet upper bound
(resp. countermonotonicity copula, i.e., the Fr\'echet lower bound);
see, e.g., Theorem 3 of \cite{EmbrechtsMcNeilStraumann2002}. 
The conditions in Proposition~\ref{prop:MN-tau-S-properties}~(v) and (vi)
for attaining these boundary cases can be understood directly from the
stochastic representation of the normal location-scale mixture distribution presented
at the beginning of Section~\ref{sec:lsmn}, as explained in the following remark.

\begin{remark} \label{remark:MN-attainment-SR}
  Let $(X_1,X_2)\sim\mathrm{MN}(0,\varrho(\rho), (\beta_1,\beta_2),F)$ 
  so that 
  $X_1 = X_2 + W(\beta_1-\beta_2) + \sqrt W (Z_1-Z_2)$
  with $W\sim F$ and $(Z_1,Z_2) \sim \mathrm N(0,\varrho(\rho))$.
  The copula $C_{\mathrm{mn}}(\rho,\beta,F)$, as the copula of $X_1$ and $X_2$,
  is comonotonic (resp. countermonotonic) if and only if $X_1$ is an almost surely
  increasing (resp. decreasing) transformation of $X_2$;
  see, e.g., Corollary 5.17 and Remark 5.20 of \cite{McNeilFreyEmbrechts2005}.
  When $\rho=1$, $Z_1=Z_2$ almost surely,
  and hence $X_1 = X_2 + W(\beta_1-\beta_2)$ almost surely.
  If $F$ and hence $W$ is nondegenerate, $X_1$ can be an increasing function of $X_2$ almost
  surely only when $\beta_1=\beta_2$.
  Likewise, when $\rho=-1$, $Z_1=-Z_2$ almost surely, and hence
  $X_1 = - X_2 + W(\beta_1+\beta_2)$ almost surely, and $X_1$ can be a decreasing
  function of $X_2$ almost surely only when $\beta_1+\beta_2=0$.
  Thus, the Fr\'echet bounds are attained only in these two special cases.
\end{remark}

\subsection{Kendall's tau and Spearman's rho for GH skew-\texorpdfstring{$t$}{} copula}
\label{sec:rank-corr-GH-t}

\begin{figure}[t!]\centering
  \includegraphics[width=0.95\textwidth]{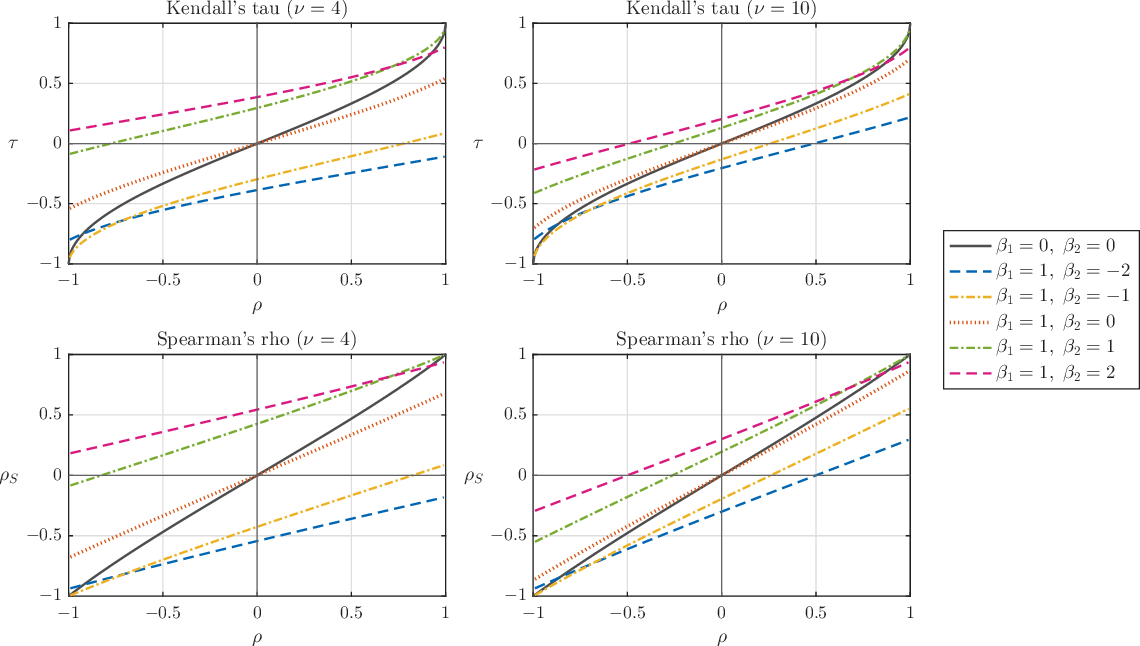}
  \caption{Kendall's tau and Spearman's rho for GH skew $t$ copula under various
  settings of skewness and degrees of freedom parameters.}
  \label{fig:GHt_tau_S}
\end{figure}

As noted earlier, the widely used GH skew-$t$ copula emerges as a special case of the
normal location-scale mixture copula when the mixing distribution $F$ is chosen as
$\mathrm{IG}(\nu/2,\nu/2)$. The formulas for $\tau$ and $\rho_S$ for the bivariate GH
skew-$t$ copula $C_{\mathrm{GHt}}(\rho,\beta,\nu)$ can then be directly obtained from
Theorem~\ref{thm:MN-tau-S}.
We evaluate both rank correlations for the GH skew-$t$ copula under various parameter
settings, illustrate the functional relationships between $\rho$ and each rank
correlation in Fig.~\ref{fig:GHt_tau_S}. For a comparison, the benchmark $t$ copula
case ($\beta_1=\beta_2=0$) is depicted as a dark solid line.

As illustrated in Fig.~\ref{fig:GHt_tau_S}, the asymmetry introduced by nonzero
skewness parameters $\beta_1$ and $\beta_2$ can significantly alter the behavior of
both $\tau$ and $\rho_S$ as functions of $\rho$.
Notably, while both rank correlations remain strictly increasing in $\rho$,
the full interval $[-1,1]$ is no longer attainable once asymmetry is introduced, 
as noted in Remark~\ref{remark:MN-attainment}.
Moreover, the property that rank correlations take the values $-1,0$ or $1$ whenever 
$\rho$ equals $-1,0$ or $1$, respectively, generally ceases to hold in the presence 
of asymmetry. Consistent with Proposition~\ref{prop:MN-tau-S-properties}~(iv),
when $\rho=0$ the sign of the rank correlations is determined by the sign of
$\beta_1\beta_2$.
Furthermore, in all four panels of Fig.~\ref{fig:GHt_tau_S}, except for the benchmark elliptical case
$\beta_1=\beta_2=0$, the rank correlations attain the upper bound $1$ only when
$\beta_1=\beta_2=1$, and they attain the lower bound $-1$ only when $\beta_1=1$ and
$\beta_2=-1$. 
This is consistent with Proposition~\ref{prop:MN-tau-S-properties}~(v) and (vi).

\subsection{Equi-skew and single-skew cases}\label{sec:mn-equi-single-skew}

Building on the general result of Theorem~\ref{thm:MN-tau-S}, we now examine two
special cases: the \emph{equi-skew} setting $(\beta_1=\beta_2)$ and the
\emph{single-skew} setting (one of $\beta_1$ and $\beta_2$ is zero). 

To appreciate the relevance of the equi-skew case, it is helpful to recall that 
elliptical copulas exhibit two key forms of symmetry: \emph{exchangeability} and
\emph{radial symmetry}. In the equi-skew setting where $\beta_1=\beta_2=b\in\RR$,
the copula $C_{\mathrm{nm}}(\rho,(b,b),F)$ preserves the exchangeability property 
of elliptical copulas, but it no longer retains radial symmetry.
Therefore, the equi-skew specification is useful in situations where radial asymmetry
is desirable while exchangeability remains important. In addition, imposing 
equi-skewness reduces the number of parameters in this class of copulas,
a feature that is especially valuable in high-dimensional applications, where parsimony
is crucial for both estimation and interpretability.
For example, both \cite{Christoffersen2012} and \cite{OhPatton2023} use the
GH skew-$t$ copula under equi-skewness or grouped equi-skewness assumptions in their 
high-dimensional applications.

\begin{figure}[t!]\centering
  \includegraphics[width=0.95\textwidth]{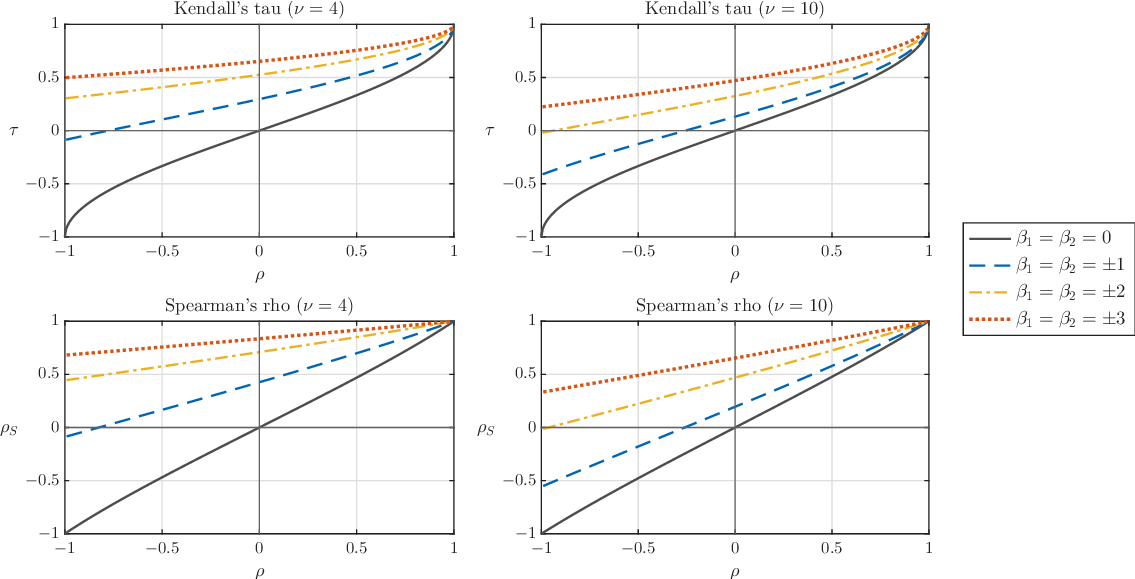}
  \caption{Kendall's tau and Spearman's rho for GH skew $t$-copula under
  equi-skewness.}
  \label{fig:GHt_tau_S_equi_skew}
\end{figure}

Corollary~\ref{cor:equi-skewness-MN} to Theorem~\ref{thm:MN-tau-S} presents
the formulas and properties of Kendall's tau and Spearman's rho for the copula
$C_{\mathrm{mn}}(\rho,(b,b),F)$ with equi-skewness. Since both rank correlations are
invariant under a sign change of $b$
(Proposition~\ref{prop:MN-tau-S-properties}~(ii)), the results in
Corollary~\ref{cor:equi-skewness-MN} focuses
on $b>0$.

\begin{corollary}\label{cor:equi-skewness-MN} (Equi-skewness)
  Suppose $X_1$ and $X_2$ are two random variables with continuous cdf's and copula
  $C_{\mathrm{mn}}(\rho,(b,b),F)$ with $b>0$. Then, the following holds.
  \smallskip
  
  \noindent \emph{(i)}
  Kendall's tau of $X_1$ and $X_2$ is given by 
  $\tau = 4\,\EE\,\Phi_2(bV,bV;\rho) - 1$,
  where $V$ is defined in Theorem~\ref{thm:MN-tau-S}~(i).
  If additionally $\EE\, V^2<\infty$, then
  $\partial\tau/\partial b > 0$ and $\partial^2\tau/\partial b\partial\rho < 0$.
  \smallskip
  
  \noindent \emph{(ii)}
  Spearman's rho of $X_1$ and $X_2$ is given by
  $\rho_S = 12\,\EE\,\Phi_2(bV_1,bV_2;\rho V_3) - 3$,
  where $V_1,V_2,V_3$ are defined in Theorem~\ref{thm:MN-tau-S}~(ii). 
\end{corollary}

Both rank correlations for GH skew-$t$ copulas under equi-skewness are presented in 
Fig.~\ref{fig:GHt_tau_S_equi_skew}. We observe that increasing the level of 
asymmetry rises both $\tau$ and $\rho_S$ for all values of $\rho$.
Moreover, the magnitude of this asymmetry-induced increase diminishes as $\rho$ grows.
In Corollary~\ref{cor:equi-skewness-MN}~(i), we establish that this monotonic pattern
holds in general within this class of copulas under mild assumptions on the mixing
distribution.
Fig.~\ref{fig:GHt_tau_S_equi_skew} also illustrates property (vi)
in Proposition~\ref{prop:MN-tau-S-properties}: under equi-skewness, both rank
correlations attain the upper bound $1$ when $\rho=1$.
In contrast, neither $\tau$ or $\rho_S$ approaches $-1$ as $\rho$ tends to $-1$,
which is fully consistent with property (v) of Proposition~\ref{prop:MN-tau-S-properties}.

Next, we consider the single-skew setting, where one of the skewness parameters is
zero. The following proposition summarizes the properties of Kendall's tau and
Spearman's rho under this condition. In light of
Proposition~\ref{prop:MN-tau-S-properties}~(i) and (ii), we consider $\beta_1=b>0$
and $\beta_2=0$ without loss of generality.  

\begin{proposition}\label{prop:single-skew-MN} (Single-skewness)
  Suppose $X_1$ and $X_2$ are two random variables with continuous cdf's and copula
  $C_{\mathrm{mn}}(\rho,(b,0),F)$ with $b>0$. Then, the following holds.
  \smallskip
  
  \noindent \emph{(i)}
  Kendall's tau and Spearman's rho of $X_1$ and $X_2$ are odd functions of $\rho$.
  \smallskip
  
  \noindent \emph{(ii)}
  Kendall's tau and Spearman's rho of $X_1$ and $X_2$ are strictly decreasing (or
  strictly increasing) in $b$ when $\rho>0$ (or when $\rho<0$).
\end{proposition}

As shown in Proposition~\ref{prop:single-skew-MN}~(i), in the single-skew case both
$\tau$ and $\rho_S$ are odd functions of $\rho$, a property that matches the
behavior of $\tau$ and $\rho_S$ for elliptical copulas. However, as illustrated
in Fig.~\ref{fig:GHt_tau_S_single_skew} for the GH skew-$t$ copula, the introduction of single skewness prevents
$\tau$ and $\rho_S$ from reaching the boundary values $\pm 1$ as $\rho$ tends to
$\pm1$. This is precisely in line with properties (v) and (vi) of
Proposition~\ref{prop:MN-tau-S-properties}. Finally,
Proposition~\ref{prop:single-skew-MN}~(ii) shows that both rank correlations strictly
decrease (increase) with the level of the single skewness when $\rho$ is positive
(negative), a pattern clearly visible in Fig.~\ref{fig:GHt_tau_S_single_skew}.

\begin{figure}[t!]\centering
  \includegraphics[width=0.95\textwidth]{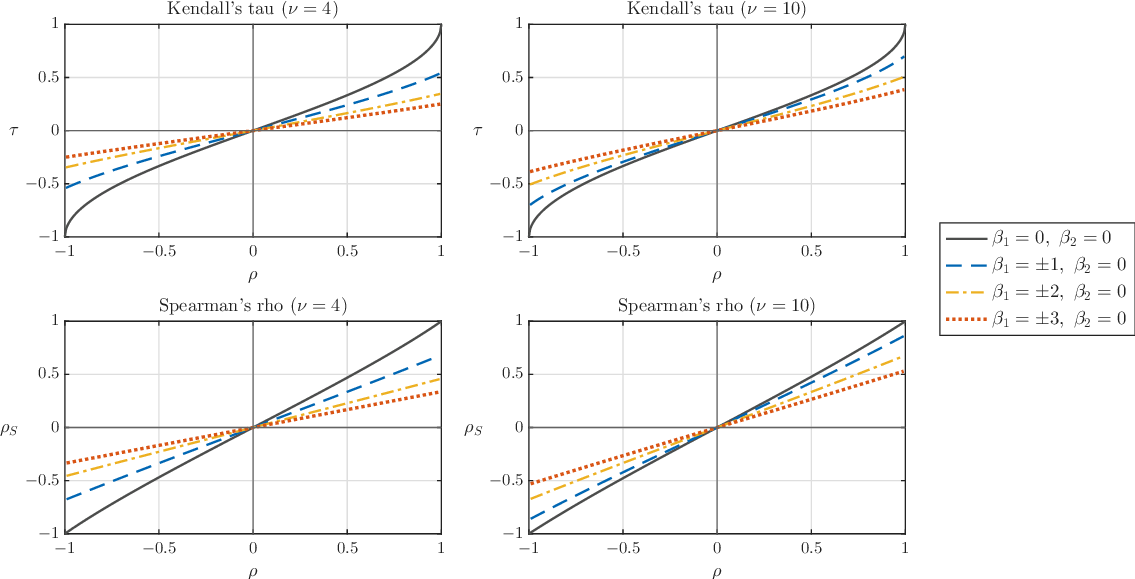}
  \caption{Kendall's tau and Spearman's rho for GH skew-$t$ copula under
  single-skewness.}
  \label{fig:GHt_tau_S_single_skew}
\end{figure}

\section{Kendall's tau and Spearman's rho for skew-normal scale mixture copulas}
\label{sec:rank-corr-MSN}

We now turn our attention to the skew-normal scale mixture copulas and establish the
formulas and properties of Kendall's tau and Spearman's rho for this class. We then
specialize the results to the skew-normal and AC skew-$t$ copulas. We also discuss
the equi-skew and single-skew settings in the final subsection.

\subsection{Formulas for Kendall's tau and Spearman's rho}

We first introduce the following $4\times 4$ and $5\times 5$ correlation-matrix-valued functions, which will play a key role in Theorem~\ref{thm:MSN-tau-S}. For $\rho\in[-1,1]$ and $u=(u_1,u_2)$ with $|u_i|<1$ for $i=1,2$, we define
\begin{equation}\label{eq:P-tau}
  \mathrm P_\tau(\rho,u,v)
  =\left[\begin{array}{cc|cc}
  1&&&\\\rho&1&&\\\cline{1-4}
  u_1v_1&u_2v_1&1&\\
  u_1v_2&u_2v_2&0&1
  \end{array}\right],
  \qquad v=(v_1,v_2),
\end{equation}
and
\begin{equation}\label{eq:P-S}
  \mathrm P_S(\rho,u,v)
  =
  \left[
  \begin{array}{cc|ccc}
  1&&&&\\
  \rho v_3&1&&&\\\cline{1-5}
  u_1v_1&0&1&&\\
  0&u_2v_2&0&1&\\
  u_1v^-_1&u_2v^-_2&0&0&1
  \end{array} 
  \right],
  \qquad v=(v_1,v_2,v_1^-,v_2^-,v_3).
\end{equation}
Here, the upper-triangular entries are omitted in the correlation matrix.
In the following theorem, $\Phi_d(0;\Sigma)$ denotes the cdf of the $d$-dimensional
normal distribution $\mathrm N(0,\Sigma)$ evaluated at zero.

\begin{theorem}\label{thm:MSN-tau-S} 
  Suppose $X_1$ and $X_2$ are two random variables with continuous cdf's and copula
  $C_{\mathrm{msn}}(\rho,\alpha,F)$ where $\rho\in[-1,1]$,
  $\alpha=(\alpha_1,\alpha_2)\in\RR^2$, and $F$ is a univariate distribution on
  $(0,\infty)$. Let $\mathrm P_\tau$ and $\mathrm P_S$ be the two
  correlation-matrix-valued functions defined in \eqref{eq:P-tau} and \eqref{eq:P-S},
  respectively, and let $\delta=(\delta_1,\delta_2)$ be the alternative skewness
  vector that corresponds to $\alpha$ via \eqref{eq:del-alp-i}.
  Then, we have the following.
  \smallskip
  
  \noindent \emph{(i)}
  Kendall's tau of $X_1$ and $X_2$ is given by
  \begin{equation} \label{eq:MSN-tau}
    \tau(\rho,\alpha,F)
    = 
    16\,\EE\,\Phi_4\big(0;\mathrm P_\tau(\rho,\delta,V)\big)-1,
  \end{equation}
  where $V=(V_1,V_2)$ in which $V_1=\sqrt{W_2/(W_1+W_2)}$ and
  $V_2=-\sqrt{W_1/(W_1+W_2)}$, with $W_1,W_2\sim$ i.i.d.\,$F$.
  \smallskip
  
  \noindent \emph{(ii)}
  Spearman's rho of $X_1$ and $X_2$ is given by 
  \begin{equation} \label{eq:MSN-rhoS}
    \rho_S(\rho,\alpha, F)
    =
    96\,\EE\,\Phi_5\big(0;\mathrm P_S(\rho,\delta,V)\big) - 3,
  \end{equation}
  where 
  $V = (V_1,V_2,V^-_1,V^-_2,V_3)$ in which 
  $V_i = \sqrt{W_i/(W_i+W_3)}$, 
  $V^-_i = -\sqrt{W_3/(W_i+W_3)}$, for $i\in\{1,2\}$, and 
  $V_3 = W_3/\sqrt{(W_1+W_3)(W_2+W_3)}$, with $W_1,W_2,W_3\sim$ i.i.d.\,$F$.
\end{theorem}

The following remark highlights some properties of the distributions of the mixing
variables in Theorem~\ref{thm:MSN-tau-S}.

\begin{remark}\label{remark:MSN-mix-dist-1}
  The distribution of $(V_1,V_2)$ in Theorem~\ref{thm:MSN-tau-S}~(i) satisfies
  $(V_1,V_2)=_d -(V_2,V_1)$. Moreover, the distribution of
  $(V_1,V_2,V_1^-,V_2^-,V_3)$ in Theorem~\ref{thm:MSN-tau-S}~(ii) satisfies
  $(V_i,V_i^-)=_d -(V_i^-,V_i)$, for $i\in\{1,2\}$.
\end{remark}

Unlike Theorem~\ref{thm:MN-tau-S}, which expresses $\tau$ and $\rho_S$ for normal
location-scale mixture copulas using only the bivariate normal cdf, the formulas
\eqref{eq:MSN-tau} and \eqref{eq:MSN-rhoS} in Theorem~\ref{thm:MSN-tau-S} involves
four- and five-dimensional normal cdf's. To provide a clearer comparison,
Corollary~\ref{cor:MSN-tau-S-Phi2} offers alternative formulas for $\tau$ and
$\rho_S$ that rely on the bivariate normal cdf. In particular, we define
\begin{equation}\label{eq:alpha-dagger}
  \alpha_i^\dagger
  \coloneqq
  \frac{\delta_i}{\sqrt{1-\delta_i^2}} 
  =
  \frac{\alpha_i+\rho\alpha_{-i}}{\sqrt{1+(1-\rho^2)\alpha_{-i}^2}},
  \qquad i\in\{1,2\},
\end{equation}
and
\begin{equation}\label{eq:rho-dagger}
  \rho^\dagger
  \coloneqq 
  \frac{\rho-\delta_1\delta_2}{\sqrt{(1-\delta_1^2)(1-\delta_2^2)}},
\end{equation}
where $(\delta_1,\delta_2)$ are related to $\rho$ and $(\alpha_1,\alpha_2)$ via
\eqref{eq:del-alp-i}, and the subscript $-i$ in \eqref{eq:alpha-dagger} denotes the
index other than $i$, for $i\in\{1,2\}$.
As noted in Remark~\ref{remark:delta-prameterization},
$\alpha_1^\dagger$ and $\alpha_2^\dagger$ defined in \eqref{eq:alpha-dagger} are the
marginal $\alpha$-skewness parameters of the bivariate skew-normal distribution
$\mathrm{SN}(0,\varrho(\rho),\alpha)$. 

\begin{corollary} \label{cor:MSN-tau-S-Phi2}
  Suppose $X_1$ and $X_2$ are two random variables with continuous cdf's 
  and copula $C_{\mathrm{msn}}(\rho,\alpha,F)$ where
  $\rho\in[-1,1]$, $\alpha=(\alpha_1,\alpha_2)\in\RR^2$, and $F$ is a univariate
  distribution on $(0,\infty)$. Let $\alpha^\dagger_1,\alpha^\dagger_2$ and
  $\rho^\dagger$ be defined by \eqref{eq:alpha-dagger} and \eqref{eq:rho-dagger},
  respectively. Then, we have the following.
  \smallskip
  
  \noindent \emph{(i)}
  Kendall's tau of $X_1$ and $X_2$ in \eqref{eq:MSN-tau} can be alternatively
  expressed as
  \begin{equation}\label{eq:MSN-tau-Phi2}
  \tau(\rho,\alpha,F) = 4\,\EE\,\Phi_2\big(
  \alpha_1^\dagger Z,\, \alpha_2^\dagger Z \,;\, \rho^\dagger
  \big) - 1,
  \end{equation}
  where $Z\coloneqq V_1Y_1+V_2Y_2$, $V=(V_1,V_2)$ is as defined in Theorem~\ref{thm:MSN-tau-S}~(i), and $Y_1,Y_2$ are independent half-normal random variables conditional on $V$.
  \smallskip
  
  \noindent \emph{(ii)}
  Spearman's rho of $X_1$ and $X_2$ in \eqref{eq:MSN-rhoS} can be alternatively
  expressed as
  \begin{equation}\label{eq:MSN-S-Phi2}
  \rho_S(\rho,\alpha,F) = 12\,\EE\,\Phi_2\big(
  \alpha_1^\dagger Z_1,\, \alpha_2^\dagger Z_2;\, \rho^\dagger V_3
  \big) - 3,
  \end{equation}
  where $Z_1\coloneqq V_1Y_1+V_1^-Y_3$, $Z_2\coloneqq V_2Y_2+V_2^-Y_3$,
  $V=(V_1,V_2,V_1^-,V_2^-,V_3)$ is as defined in Theorem~\ref{thm:MSN-tau-S}~(ii), and
  $Y_1,Y_2,Y_3$ are mutually independent half-normal random variables conditional on
  $V$.
\end{corollary}

In Corollary~\ref{cor:MSN-tau-S-Phi2}, the dimension of the normal cdf in the rank
correlation formulas is reduced from four or five (as in Theorem~\ref{thm:MSN-tau-S})
to only two, at a cost of a more intricate construction of mixing variables.
Importantly, the alternative expressions of $\tau$ and $\rho_S$ in
Corollary~\ref{cor:MSN-tau-S-Phi2} closely parallel those in
Theorem~\ref{thm:MN-tau-S} for the other copula class. In the
following remark, we clarify the similarities and differences between the rank
correlation formulas in Theorem~\ref{thm:MN-tau-S} and
Corollary~\ref{cor:MSN-tau-S-Phi2}, and relate them to the common and unique
properties of the rank correlations across different models, as
further discussed in Section~\ref{sec:msn-properties}.

\begin{remark}\label{remark:MSN-mix-dist-2}
(i) Although the distributions of the mixing variables $Z$ and $(Z_1,Z_2)$ in
Corollary~\ref{cor:MSN-tau-S-Phi2} are different from those of $V$ and
$(V_1,V_2)$ in Theorem~\ref{thm:MN-tau-S}, their marginal distributions are all
symmetric about zero. As a result, Kendall's tau and Spearman's rho for both classes
of copulas are invariant under the sign change of the skewness
vector, as shown in Proposition~\ref{prop:MN-tau-S-properties}~(ii) and will be shown
in Proposition~\ref{prop:MSN-tau-S-properties}~(ii).\smallskip

\noindent (ii)
A subtle but important difference in the rank correlation formulas in
Corollary~\ref{cor:MSN-tau-S-Phi2}, compared to those in Theorem~\ref{thm:MN-tau-S},
lies in the nature of the coefficients.
In Corollary~\ref{cor:MSN-tau-S-Phi2}, the
coefficients $\alpha_1^\dagger,\alpha_2^\dagger$ and $\rho^\dagger$ that multiply the
mixing variables are transformations of the skewness parameters $\alpha$
and correlation parameter $\rho$, rather than these parameters themselves. This
distinction is crucial to understanding why, unlike the case for normal
location-scale mixture copulas, the entire interval $[-1,1]$ is attainable for $\tau$
and $\rho_S$ in the class of skew-normal scale mixture copulas
(Proposition~\ref{prop:MSN-tau-S-properties}~(iii)).
\end{remark}
 
\subsection{Properties of Kendall's tau and Spearman's rho} \label{sec:msn-properties}

Using the representations of $\tau$ and $\rho_S$ in
Corollary~\ref{cor:MSN-tau-S-Phi2}, we establish the following properties for
bivariate skew-normal scale mixture copulas.

\begin{proposition} \label{prop:MSN-tau-S-properties}
  Kendall's tau and Spearman's rho for the bivariate skew-normal scale mixture copula
  $C_{\mathrm{msn}}(\rho,\alpha,F)$ possess the following properties:
  \smallskip
  
  \noindent \emph{(i)} 
  $\tau(\rho,(\alpha_1,\alpha_2),F)=\tau(\rho,(\alpha_2,\alpha_1),F)$  
  and  
  $\rho_S(\rho,(\alpha_1,\alpha_2),F)=\rho_S(\rho,(\alpha_2,\alpha_1),F)$.
  \smallskip
  
  \noindent \emph{(ii)} 
  $\tau(\rho,\alpha,F)=\tau(\rho,-\alpha,F)$
  and
  $\rho_S(\rho,\alpha,F)=\rho_S(\rho,-\alpha,F)$.
  \smallskip 
  
  \noindent \emph{(iii)}
  $\tau=\rho_S = -1$ if and only if $\rho=-1$.
  \smallskip 
  
  \noindent \emph{(iv)}
  $\tau=\rho_S = 1$ if and only if $\rho=1$.
\end{proposition}

Properties (i) and (ii) in Proposition~\ref{prop:MSN-tau-S-properties} parallel those
for the normal location-scale mixture copulas stated in 
Proposition~\ref{prop:MN-tau-S-properties}~(i) and (ii).
Interestingly, Proposition~\ref{prop:MSN-tau-S-properties}~(iii)--(iv) show that,
despite the presence of asymmetry, both rank correlations for this class of copulas
still attain their lower and upper bound ($-1$ or $1$) when $\rho=-1$ or $1$,
exactly as in the elliptical case. Moreover, since $\tau$ and $\rho_S$ are continuous
functions of $\rho$, the entire range $[-1,1]$ is attainable.
These findings stand in sharp contrast to the behavior of the normal location-scale
mixture copulas, for which the full range is not attainable; see 
Proposition~\ref{prop:MN-tau-S-properties}~(v)--(vi) and
Remark~\ref{remark:MN-attainment}.

Using the stochastic representations of the skew-normal distribution and its scale
mixtures introduced in Section~\ref{sec:msn}, it is not difficult to see why the 
skew-normal scale mixture copulas always attain the Fr\'echet lower
or upper bounds when the pseudo-correlation $\rho$ reaches its end values,
mirroring the behavior of the elliptical copula. The argument is given in the
following remark, which may be compared with Remark~\ref{remark:MN-attainment-SR}.

\begin{remark}\label{remark:MSN-attain-SR}
  Let $(X_1,X_2)\sim\mathrm{MSN}(0,\varrho(\rho), \alpha, F)$ 
  so that we have $X_1 = X_2 + \sqrt W (Z_1-Z_2)$
  with $W\sim F$ and $(Z_1,Z_2)\sim \mathrm{SN}(0,\varrho(\rho),\alpha)$.
  By the stochastic representation of skew-normal distributions in 
  \eqref{eq:SN-conditioning-0}--\eqref{eq:SN-conditioning}, we may write
  $(Z_1,Z_2) = ((Y_1,Y_2) \given Y_0 > 0)$, where 
  \[
  \begin{bmatrix} Y_1\\Y_2\\Y_0 \end{bmatrix}
  \sim
  \mathrm N_3(0, \mathrm R),
  \qquad 
  \mathrm R = 
  \begin{bmatrix} \mathrm \varrho(\rho) & \delta \\ \delta^\top & 1  \end{bmatrix},
  \qquad
  \delta 
  =
  \frac{\varrho(\rho)\,\alpha}{\sqrt{1+\alpha^\top \varrho(\rho) \alpha}}.
  \]
  When $\rho=1$, we have $Y_1=Y_2$ almost surely, and hence $Z_1=Z_2$ almost surely
  by construction. Consequently, $X_1=X_2$ almost surely as well, implying that the 
  copula $C_{\mathrm{msn}}(-1,\alpha,F)$ is comonotonic, i.e., equal to the Fr\'echet 
  upper bound, regardless of the choices of $\alpha$ or $F$.
  Similarly, when $\rho=-1$, we have $X_1=-X_2$ almost surely, implying that the copula 
  $C_{\mathrm{msn}}(-1,\alpha,F)$ is countermonotonic, i.e., equal to the Fr\'echet 
  lower bound, again independently of $\alpha$ or $F$.
\end{remark}

Although Proposition~\ref{prop:MN-tau-S-properties}~(iii) shows that, for normal
location-scale mixture copulas, both rank correlations increase strictly with $\rho$
when $\beta$ and $F$ are fixed, establishing an analogous result for
skew-normal scale mixture copulas appears theoretically challenging.
Lemmas 1 and 3 of \cite{HeinenValdesogo2022} assert that
$\partial\tau/\partial\rho>0$ and $\partial\rho_S/\partial\rho>0$ for bivariate
skew-normal copulas; however, their partial derivatives are meaningful only if
$\delta_1$ and $\delta_2$ are held fixed as $\rho$ varies.
Since the admissible values of $(\delta_1,\delta_2)$ themselves depend on $\rho$ (see
Remark~\ref{remark:delta-prameterization-copula}), this condition cannot generally be
satisfied.

Nevertheless, Fig.~\ref{fig:ACt_tau_S} suggests that, for AC skew-$t$ copulas,
both Kendall’s tau and Spearman’s rho increase strictly with $\rho$
for any fixed skewness vector $\alpha$.
Moreover, under equi-skewness and single-skewness assumptions, we prove that Kendall’s 
tau is strictly increasing in $\rho$ for skew-normal scale mixture
copulas (Propositions~\ref{prop:equi-skewness-MSN} (i) and
\ref{prop:single-skew-MSN} (ii)).

\subsection{Kendall's tau and Spearman's rho for skew-normal and AC
skew-\texorpdfstring{$t$}{} copulas} 
\label{sec:rank-corr-AC-t}

\begin{figure}[t!]\centering
  \includegraphics[width=0.95\textwidth]{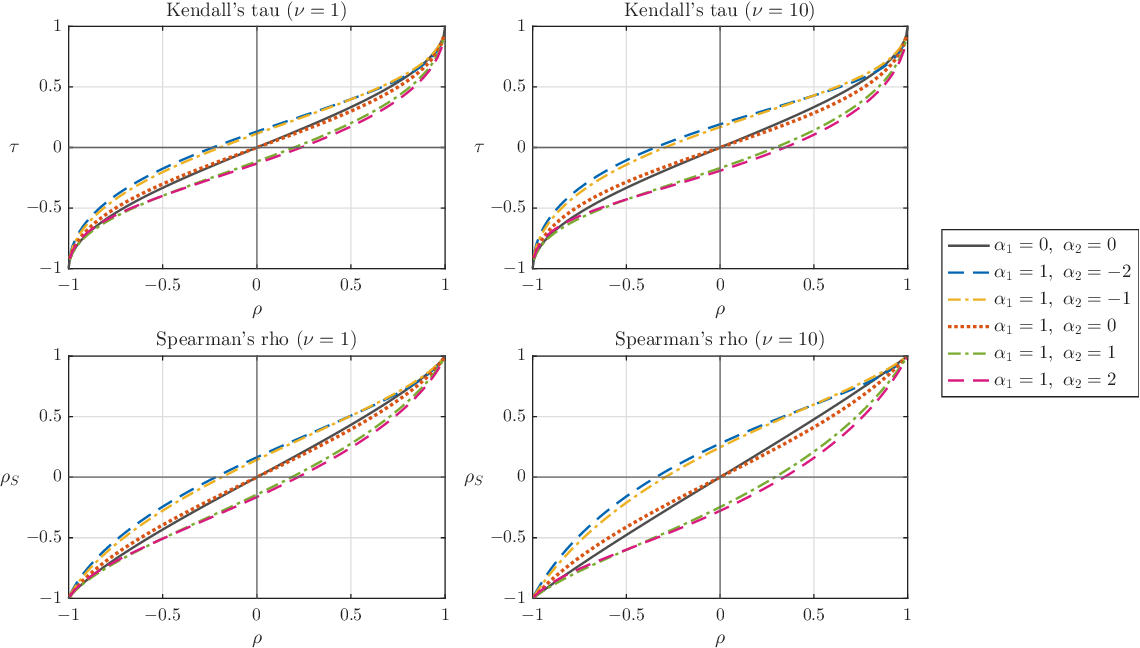}
  \caption{Kendall's tau and Spearman's rho of AC skew $t$ copula under various
  settings of the skewness and degrees of freedom parameters.}
  \label{fig:ACt_tau_S}
\end{figure}

As discussed in Section~\ref{sec:msn}, two important special cases of the skew-normal
scale mixture copula arise when the mixing distribution is either
$\mathrm{IG}(\nu/2,\nu/2)$ or degenerated at $1$. Specifying $F$ as
$\mathrm{IG}(\nu/2,\nu/2)$ in Theorem~\ref{thm:MSN-tau-S} or
Corollary~\ref{cor:MSN-tau-S-Phi2}, we can obtain the formulas for $\tau$ and
$\rho_S$ for the bivariate AC skew-$t$ copula.

On the other hand, substituting the degenerate value $1$ for all the mixing variables
$W$'s in Theorem~\ref{thm:MSN-tau-S} and Corollary~\ref{cor:MSN-tau-S-Phi2}
yields Corollary~\ref{cor:SN-tau-S} in the
supplementary material, which provides formulas for $\tau$ and $\rho_S$ for the
bivariate skew-normal copula. 

\begin{corollary}\label{cor:SN-tau-S} 
  For the bivariate skew-normal copula $C_{\mathrm{sn}}(\rho,\alpha)$, we have the
  following alternative formulas for Kendall's tau and Spearman's rho.
  \smallskip
  
  \noindent\emph{(i)} 
  $\tau=16\,\Phi_4\big(0;\mathrm P_{\tau}(\rho,\delta,(c,-c))\big)-1$ and
  $\rho_S=96\,\Phi_5\big(0;\mathrm P_S(\rho,\delta,(c,c,-c,-c,c^2))\big)-3$,
  where $c=1/\sqrt 2$, and $\delta$ is the alternative skewness vector that
  corresponds to $\alpha$ via \eqref{eq:del-alp-i}.
  \smallskip
  
  \noindent\emph{(ii)} 
  $\tau = 4\,\EE\,\Phi_2\big(\alpha_1^\dagger Z, \alpha_2^\dagger Z;\rho^\dagger\big)-1$ and 
  $\rho_S = 12\,\EE\,\Phi_2\big(\alpha_1^\dagger Z_1, \alpha_2^\dagger Z_2;\rho^\dagger/2\big)-3$,
  where $\alpha^\dagger_1,\alpha^\dagger_2$ and $\rho^\dagger$ defined in \eqref{eq:alpha-dagger} and \eqref{eq:rho-dagger}, $Z=(Y_1-Y_2)/\sqrt 2$, $Z_1=(Y_1-Y_3)/\sqrt 2$, and $Z_2=(Y_2-Y_3)/\sqrt 2$, with $(Y_1,Y_2,Y_3)$ being mutually independent half-normal random variables.
\end{corollary}

The expressions for $\tau$ and $\rho_S$ in part (i) of Corollary~\ref{cor:SN-tau-S}, written in
terms of four- and five-dimensional normal orthant probabilities, are consistent with 
Propositions 1 and 3 of \cite{HeinenValdesogo2022}, who study rank correlations for
bivariate skew-normal distributions. By contrast, the formulas in part (ii) of the
corollary, which involve only bivariate normal cdf's, do not appear in their work.

Fig.~\ref{fig:ACt_tau_S} plots $\tau$ and $\rho_S$ as functions of $\rho$
for the bivariate AC skew-$t$ copula with $\nu=1$ and $\nu=10$. 
It is worth noting that the skew-normal copula arises as the limit $\nu\to\infty$,
and its rank-correlation curves are visually close to those for $\nu=10$.
In contrast to the GH skew-$t$ case in Fig.~\ref{fig:GHt_tau_S}, the
most notable difference in Fig.~\ref{fig:ACt_tau_S} is that both ends of the curves
are tied at $(-1,-1)$ and $(1,1)$, as indicated by
Proposition~\ref{prop:MSN-tau-S-properties}~(iii)--(iv). Moreover, both rank correlations
increase strictly with $\rho$ and, unlike for GH skew-$t$ copulas, always span the
full interval $[-1,1]$. Finally, holding one skewness parameter (here $\alpha_1$),
increasing the other (here $\alpha_2$) \emph{decreases} both rank correlations, which is 
the opposite behavior to that observed for GH skew-$t$ copulas.

\subsection{Equi-skew and single-skew cases} \label{sec:msn-equi-single-skew}

As in Section~\ref{sec:mn-equi-single-skew}, we also consider the equi-skew and 
single-skew settings. First, if $\alpha_1=\alpha_2=a$ for some $a\in\RR$,
it follows from \eqref{eq:del-alp-i} that 
\begin{equation}\label{eq:delta-circ}
\delta_1=\delta_2
=\bar\delta\coloneqq\frac{a(1+\rho)}{\sqrt{1+2a^2(1+\rho)}}.
\end{equation}
Moreover, it follows from \eqref{eq:alpha-dagger} that
\begin{equation}\label{eq:alpha-circ}
\alpha^\dagger_1 = \alpha^\dagger_2 
= a^\dagger \coloneqq
\frac{a(1+\rho)}{\sqrt{1+a^2(1-\rho^2)}}.
\end{equation}
It is straightforward to verify that 
(i) $\mathrm{sgn}(\bar\delta)=\mathrm{sgn}(a^\dagger)=\mathrm{sgn}(a)$,
(ii) $\partial\bar\delta/\partial a>0$ and $\partial a^\dagger/\partial a>0$, and
(iii) $\mathrm{sgn}(\partial\bar\delta/\partial\rho)=\mathrm{sgn}(\partial
a^\dagger/\partial \rho)=\mathrm{sgn}(a)$, for $a\in\RR$ and $\rho\in(-1,1)$.
Furthermore, when $\alpha_1=\alpha_2=a$, it follows from \eqref{eq:rho-dagger} and
\eqref{eq:delta-circ} that
\begin{equation}\label{eq:rho-dagger-circ}
  \rho^\dagger\coloneqq \frac{1+\rho}{1+a^2(1-\rho^2)}-1.
\end{equation}

\begin{figure}[t!] \centering
  \includegraphics[width=0.95\textwidth]{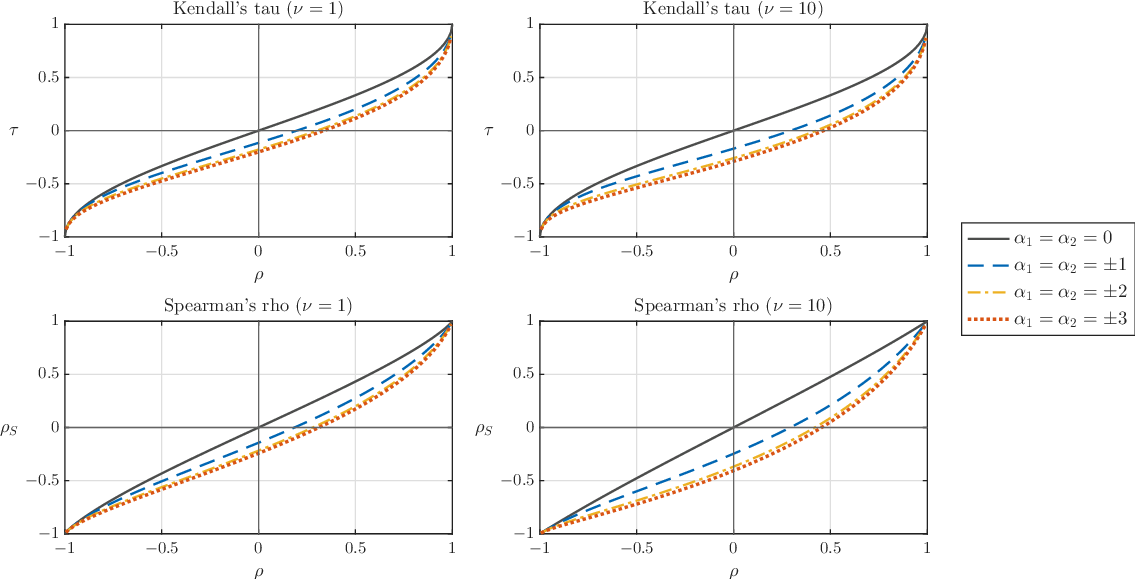}
  \caption{Kendall's tau and Spearman's rho of AC skew $t$-copula under
  equi-skewness.}
  \label{fig:ACt_tau_S_equi_skew}
\end{figure}

Since both rank correlations are invariant under a sign change of the skewness vector
(Proposition~\ref{prop:MSN-tau-S-properties}~(ii)), we focus on $a>0$ in the
following proposition.

\begin{proposition}\label{prop:equi-skewness-MSN} (Equi-skewness)
  Suppose $X_1$ and $X_2$ are two random variables with continuous cdf's
  and copula $C_{\mathrm{msn}}(\rho,(a,a),F)$ with $a>0$.
  Let $a^\dagger$ and $\rho^\dagger$ be defined by \eqref{eq:alpha-circ} and
  \eqref{eq:rho-dagger-circ}. Then, we have the following.
  \smallskip
  
  \noindent \emph{(i)}
  Kendall's tau of $X_1$ and $X_2$ is given by $\tau=4\,\EE\,\Phi_2(a^\dagger Z,
  a^\dagger Z;\rho^\dagger)-1$, where $Z$ is defined in
  Corollary~\ref{cor:MSN-tau-S-Phi2}~(i). It is strictly increasing in $\rho$ and
  strictly decreasing in $a$.
  \smallskip
  
  \noindent \emph{(ii)}
  Spearman's rho of $X_1$ and $X_2$ is given by $\rho_S=12\,\EE\,\Phi_2(a^\dagger
  Z_1, a^\dagger Z_2;\rho^\dagger V_3) - 3$, where $(Z_1,Z_2)$ and $V_3$ are defined,
  respectively, in Corollary~\ref{cor:MSN-tau-S-Phi2}~(ii) and
  Theorem~\ref{thm:MSN-tau-S}~(ii).
\end{proposition}

Fig.~\ref{fig:ACt_tau_S_equi_skew} illustrates Kendall's tau and Spearman's rho as
functions of $\rho$ for AC skew-$t$ copulas under the equi-skew setting. All the
functions are strictly increasing in $\rho$, and their ends meet at $(-1,-1)$ and
$(1,1)$. Unlike the GH skew-$t$ copula case shown in
Fig.~\ref{fig:GHt_tau_S_equi_skew}, increasing the level of asymmetry (i.e. $a$ or
$\bar\delta$ in the equi-skew case) in AC skew-$t$ copulas reduces both rank
correlations for any fixed $\rho$.

For the single-skew case, we recall that
Remark~\ref{remark:delta-prameterization-copula} noted that setting one of the skewness
parameters in $(\alpha_1,\alpha_2)$ to zero does not necessarily imply that one of
$(\delta_1,\delta_2)$ is zero, and vice versa. Here, we consider the single-skew case
as requiring one of $\alpha_1$ and $\alpha_2$ to be zero, since $\alpha_1$
and $\alpha_2$ are unconstrained. Without loss of generality, let
$\alpha_1=a$, for some $a\in\RR$, and $\alpha_2=0$. 
In this case, from \eqref{eq:del-alp-i} we obtain that $\delta_1 = \delta_\circ$ and
$\delta_2 = \delta_\circ\rho$, where $\delta_\circ\coloneqq a/\sqrt{1+a^2}$. Note
that $\mathrm{sgn}(\delta_\circ)=\mathrm{sgn}(a)$  and $\partial\delta_\circ/\partial
a>0$. Moreover, it follows from \eqref{eq:rho-dagger} and \eqref{eq:alpha-dagger}
that in this case, 
\begin{equation}\label{eq:ra-dagger-single}
  \rho^\dagger
  =
  \frac{\rho}{\sqrt{1+a^2(1-\rho^2)}},
  \qquad
  \alpha_1^\dagger =a,
  \qquad 
  \alpha_2^\dagger = \frac{a\rho}{\sqrt{1+a^2(1-\rho^2)}}=a\rho^\dagger.
\end{equation}

Again, in light of Proposition~\ref{prop:MSN-tau-S-properties}~(ii) we focus on the
$a>0$ case in the following proposition. The corresponding results for $a<0$ can be
deduced accordingly.

\begin{figure}[t!]\centering
  \includegraphics[width=0.95\textwidth]{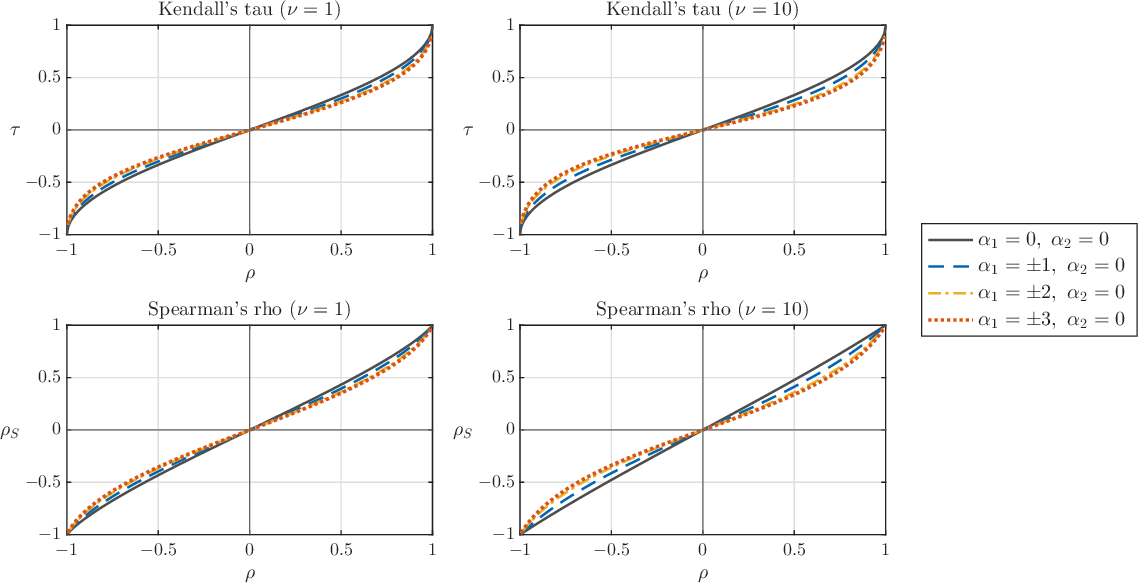}
  \caption{Kendall's tau and Spearman's rho of AC skew $t$-copula under
  single-skewness.}
  \label{fig:ACt_tau_S_single_skew}
\end{figure}

\begin{proposition}\label{prop:single-skew-MSN} (Single-skewness)
  Suppose $X_1$ and $X_2$ are two random variables with continuous cdf's
  and copula $C_{\mathrm{msn}}(\rho,(a,0),F)$ with $a>0$. Let $\rho^\dagger$ be defined
  by \eqref{eq:ra-dagger-single}. Then, we have the following.
  \smallskip
  
  \noindent \emph{(i)}
  Kendall's tau and Spearman's rho of $X_1$ and $X_2$ are odd functions of $\rho$.
  \smallskip
  
  \noindent \emph{(ii)}
  Kendall's tau of $X_1$ and $X_2$ is given by $\tau=4\,\EE\,\Phi_2(a Z,
  a\rho^\dagger Z;\rho^\dagger)-1$,   where $Z$ is defined in 
  Corollary~\ref{cor:MSN-tau-S-Phi2}~(i).
  It is strictly increasing in $\rho$, and strictly decreasing (or increasing) in $a$ when
  $\rho>0$ (or when $\rho<0$).
  \smallskip
  
  \noindent \emph{(iii)}
  Spearman's rho of $X_1$ and $X_2$ is given by $\rho_S=12\,\EE\,\Phi_2(aZ_1,
  a\rho^\dagger Z_2;\rho^\dagger V_3)-3$, where $(Z_1,Z_2)$ and $V_3$ are defined,
  respectively, in Corollary~\ref{cor:MSN-tau-S-Phi2}~(ii) and
  Theorem~\ref{thm:MSN-tau-S}~(ii).
\end{proposition} 

As for normal location-scale mixture copulas (Proposition~\ref{prop:single-skew-MN}),
under single-skewness both rank correlations for skew-normal scale mixture copulas
are odd functions of $\rho$. This is illustrated in Fig.~\ref{fig:ACt_tau_S_single_skew}
which plots $\tau$ and $\rho_S$ against $\rho$ for AC skew-$t$ copulas.
The figure further shows that increasing the level of single-skewness reduces the magnitudes
of both rank correlations, mirroring the behavior observed for GH skew-$t$ copulas.

\section{Invertibility} \label{sec:invertibility}

As mentioned in the Introduction, the formulas for rank correlations
for elliptical copulas are particularly useful for rank-based estimation of copula 
parameters.
In particular, the relationship~\eqref{eq:tau} shows that the map from
$\rho$ to $\tau$ is strictly increasing and be inverted to yield 
$\rho=\sin(\pi\tau/2)$.
Consequently, the pseudo-correlation $\rho$ can be estimated directly from the
empirical Kendall's tau, which is robust to various characteristics of the marginal
distributions.
To estimate an additional parameter in non-Gaussian elliptical copulas using
rank-based methods, the information contained in $\rho_S$ can be further exploited.
For example, 
\cite{HeinenValdesogo2020} investigated rank-based estimation of the parameters
$(\rho,\nu)$ in a $t$-copula by inverting the map from $(\rho,\nu)$ to
$(\tau,\rho_S)$. 

For the two classes of asymmetric copulas analyzed in this paper, however, the
one-to-one relationship between $\rho$ and $\tau$ no longer holds. In fact, both
$\tau$ and $\rho_S$ depend additionally on the skewness and mixing-distribution
parameters.
For instance, in a bivariate GH skew-$t$ copula, $\tau$ and $\rho_S$ are functions of
four parameters, $(\rho, \beta_1,\beta_2,\nu)$; similarly, they are functions of
$(\rho, \alpha_1,\alpha_2,\nu)$ in a bivariate AC skew-$t$ copula. 
It is clearly impossible to invert the map from all these parameters to $(\tau,
\rho_S)$.
So, if one aims to estimate all parameters using a rank-based method,
additional copula moments---such as quantile/tail dependence functions---would be
required.

If, on the other hand, the additional parameters governing asymmetry and the mixing
distribution are known or can be determined from other sources of information, it may
be possible to recover $\rho$ from a rank correlation.
For example, if there is sufficient evidence to assume that the data indeed follow
some normal location-scale or skew-normal scale mixture distribution
(a stronger assumption than specifying only the copula), the
skewness and mixing-distribution parameters can often be interpreted in term of certain
moments (such as skewness or kurtosis) of the marginal distributions. In such cases, 
information from the marginals may be used first to estimate these parameters.

Indeed, Proposition~1 shows that $\tau$ and $\rho_S$ for normal location-scale
mixture copulas are \emph{strictly increasing} functions of $\rho$.
As discussed following Proposition 2, we were unable to formally prove the strict
monotonicity of the rank correlations in $\rho$ for skew-normal scale mixture
copulas, although simulation results for a variety of mixing distributions strongly
suggest that this property holds. In practice, however, the values of the skewness
and mixing-distribution parameters are typically unknown or hard to be precisely
estimated, and the fact that the attainable ranges of $\tau$ and $\rho_S$ are
generally strictly smaller than $[-1,1]$ further complicates rank-based estimation. 

Despite these challenges, it is still of interest to examine the invertibility of
the mapping from copula parameters to $(\tau,\rho_S)$ in certain special cases.
Specifically, we consider settings in which the mixing distribution is known and the
skewness parameters are either equal or one is zero.
In both cases, the copula is characterized by only two unknown parameters, $\rho$ and
a single skewness parameter. These cases are examined in the following two subsections.

\subsection{Invertibility under equi-skewness and known mixing distribution}

For a bivariate normal location-scale mixture copula with equal skewness parameters
($\beta_1=\beta_2=b$) and a \emph{known} mixing distribution $F$, 
the model involves two unknown parameters: the pseudo-correlation $\rho$ and the
common skewness $b$.
To assess whether the map from $(\rho, b)$ to $(\tau,\rho_S)$ can be inverted,
we begin with two preliminary observations.
First, by Proposition~\ref{prop:MN-tau-S-properties}~(ii),
the sign of $b$ is not identifiable from the rank correlations; hence, we
restrict attention to $b>0$. 
Second, Proposition~\ref{prop:MN-tau-S-properties}~(vi) shows that 
$\rho=1$ if and only if $\tau=\rho_S=1$, regardless of the value of $b$. 
Thus, $b$ is not identifiable when $\rho=1$.
Given these observations, the remaining question of interest is the identification of
$(\rho,b)$ up to a sign flip of $b$ for $\rho\in[-1,1)$, which amounts to examining
the invertibility of 
\begin{equation}\label{eq:R-mn-equi}
  R_{\mathrm{mn}}(\rho, b) 
  \coloneqq
  \begin{bmatrix} \tau(\rho,(b,b),F) \\ \rho_S(\rho,(b,b),F) \end{bmatrix},
  \qquad 
  \rho\in[-1,1), \ b>0,
\end{equation}
with the expressions for $\tau$ and $\rho_S$ given in Corollary~\ref{cor:equi-skewness-MN}.

Analogously, for a bivariate skew-normal scale mixture copula with equal skewness
parameters ($\alpha_1=\alpha_2=a$) and a known mixing distribution $F$, the model
also depends on two parameters, $\rho$ and $a$, with the sign of $a$ likewise not
identifiable from rank correlations
(Proposition~\ref{prop:MSN-tau-S-properties}~(ii)). 
Moreover, Proposition~\ref{prop:MSN-tau-S-properties}~(iii) and (iv) imply that
$\tau=\rho_S=\pm1$ whenever $\rho=\pm1$, so $a$ is not identifiable at
$\rho=\pm1$. Accordingly, we will study the identification of $(\rho,a)$ up to a sign
flip of $a$ for $\rho\in(-1,1)$ by analyzing the invertibility of 
\begin{equation}\label{eq:R-msn-equi}
  R_{\mathrm{msn}}(\rho, a) 
  \coloneqq
  \begin{bmatrix} \tau(\rho,(a,a),F) \\ \rho_S(\rho,(a,a),F) \end{bmatrix},
  \qquad 
  \rho\in(-1,1), \ a>0,
\end{equation}
with the corresponding expressions for $\tau$ and $\rho_S$ provided in 
Proposition~\ref{prop:equi-skewness-MSN}.

Note that the inverses functions of $R_{\mathrm{mn}}(\rho,b)$ and 
$R_{\mathrm{msn}}(\rho,a)$ exist if and only if their corresponding Jacobians are
nonsingular. The following proposition provides the explicit expressions of these
Jacobians.

\begin{proposition}\label{prop:J-equi} 
The Jacobian of the function $R_{\mathrm{mn}}$ defined in \eqref{eq:R-mn-equi} is
given by
\begin{equation*}\label{eq:J-mn-equi}
J_{\mathrm{mn}}(\rho, b) = 4
\begin{bmatrix}
  \EE\,\phi_2(bV,bV;\rho) &
  \EE\, V\phi^{\mathrm s}(bV;\alpha_\rho) \\
  3\,\EE\, V_3 \phi_2(bV_1,bV_2;\rho V_3) & 
  3\,\EE\, [g(V_1,V_2,V_3; \rho, b) + g(V_2,V_1,V_3; \rho, b)]
\end{bmatrix},
\quad \rho\in(-1,1),\ b>0,
\end{equation*}
where 
$\alpha_\rho \coloneqq \sqrt{(1-\rho)/(1+\rho)})$,
$g(x,y,z; \rho, b) \coloneqq x \phi(bx)
\Phi\big(b(y-\rho xz)/\sqrt{1-\rho^2z^2}\big)$, 
and $V,V_1,V_2,V_3$ are defined in Theorem~\ref{thm:MN-tau-S}.
The Jacobian of the function $R_{\mathrm{msn}}$ defined in \eqref{eq:R-msn-equi} is
given by
\[
  J_{\mathrm{msn}}(\rho,a) = 4
  \begin{bmatrix}
  \EE\,\phi_2(a^\dagger Z, a^\dagger Z;\rho^\dagger)
  & \EE\, Z \phi^{\mathrm s}(a^\dagger Z, \alpha_{\rho^\dagger})\\
  3\,\EE\, V_3\,\phi_2(a^\dagger Z_1, a^\dagger Z_2;\rho^\dagger V_3)
  & 3\,\EE\,[g(Z_1,Z_2,V_3;\rho^\dagger,\alpha^\dagger) 
           + g(Z_2,Z_1,V_3;\rho^\dagger,\alpha^\dagger) ]
  \end{bmatrix}
  B,
\]
for $\rho\in(-1,1)$ and $a>0$, where $a^\dagger$ and $\rho^\dagger$ are defined in
\eqref{eq:alpha-circ} and \eqref{eq:rho-dagger-circ}, $\alpha_{\rho^\dagger}\coloneqq
\sqrt{(1-\rho^\dagger)/(1+\rho^\dagger)})$, $(Z,Z_1,Z_2)$ and $V_3$ are defined in
Corollary~\ref{cor:MSN-tau-S-Phi2} and Theorem~\ref{thm:MSN-tau-S}~(ii), $g$ is the
same function as defined
above, and
\[
  B\coloneqq
  \begin{bmatrix}
    s^{-1} + 2a^2\rho(1+\rho)s^{-2} & -2a(1+\rho)(1-\rho^2) s^{-2}\\
    a[1+a^2(1+\rho)] s^{-3/2}       & (1+\rho) s^{-3/2}
  \end{bmatrix},
  \qquad
  s\coloneqq 1+a^2(1-\rho^2).
\]
\end{proposition}

\begin{figure}[t!]\centering
  \includegraphics[width=0.95\textwidth]{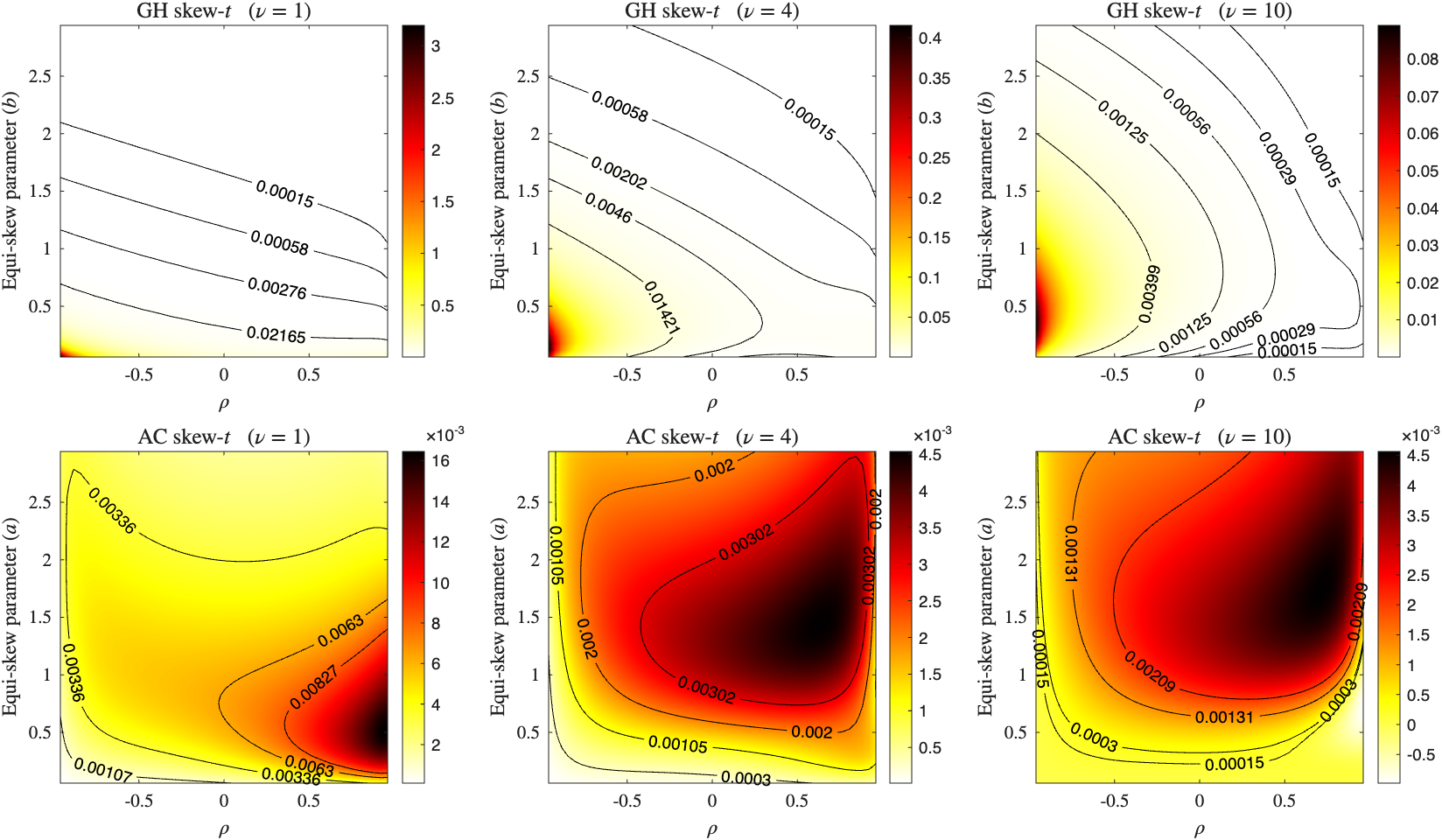}
  \caption{Contour plots of $\det J_{\mathrm{GHt}}(\rho,b)$
  for GH skew-$t$ copulas with equi-skewness $b\in(0,3)$ and $\rho\in(-1,1)$ 
  in the upper panels, and $\det J_{\mathrm{ACt}}(\rho,a)$ for AC skew-$t$ copulas
  with equi-skewness $a\in(0,3)$ and $\rho\in(-1,1)$ in the lower panels.}
  \label{fig:Jacobian_equi_skew}
\end{figure}

From these expressions for the Jacobians $J_{\mathrm{mn}}(\rho,b)$ and
$J_{\mathrm{mn}}(\rho,a)$, it is generally difficult to assess their invertibility (or
nonsingularity) analytically. Nevertheless, these formulas can be used to diagnose
invertibility of a Jacobian at certain parameter values by numerically evaluating 
its determinants. 

Fig.~\ref{fig:Jacobian_equi_skew} shows contour plots of the determinants of the
Jacobians for the GH skew-$t$ copulas $\det J_{\mathrm{GHt}}(\rho,b)$
and for the AC skew-$t$ copulas $\det J_{\mathrm{ACt}}(\rho,b)$, under
equi-skewness and known mixing distribution $\mathrm{IG}(\nu/2,\nu/2)$
for $\nu=1,4$ and $10$. The expectations in Proposition~\ref{prop:J-equi} are
computed via quasi-Monte Carlo integration.

The figure shows that $\det J_{\mathrm{GHt}}(\rho,b)>0$ across the entire parameter
region for all three values of $\nu$. By contrast, for the AC skew-$t$ copula with 
$\nu=10$, $J_{\mathrm{ACt}}(\rho,a)$ becomes singular when the equi-skewness
parameter $a$ is near zero or when $\rho$ approaches the boundary values $\pm1$. 
The contour plots also reveal regions where the Jacobian is nearly singular,
indicating \emph{weak} invertibility of the rank-correlation mapping
even though the numerical value of the determinant remains nonzero.
For example, taking $1.5\times 10^{-4}$ as a threshold, increasing $\nu$ tends to enlarge 
the region with $|\det J_{\mathrm{GHt}}(\rho,b)|>1.5\times 10^{-4}$,
while the corresponding region with $|\det J_{\mathrm{ACt}}(\rho,a)|>1.5\times
10^{-4}$ tends to shrink. 
Moreover, weak invertibility for GH skew-$t$ copulas typically arises at large values
of both $\rho$ and $b$, whereas for AC skew-$t$ copulas it occurs typically when
$\rho$ is near $-1$ and $a$ is close to zero.

\subsection{Invertibility under single-skewness and known mixing distribution}

Next, we turn to the single-skew case. For both copula classes under consideration,
when only one skewness parameter is nonzero and the mixing distribution $F$ is known,
the model involves two unknown parameters: the pseudo-correlation $\rho$ and the 
single skewness parameter, denoted by $b$ or $a$ depending on the copula class.

To assess whether $(\rho, b)$ or $(\rho, a)$ can be identified from the rank correlations
$(\tau,\rho_S)$ in this setting, we make several preliminary observations.
First, as in the equi-skewness case, the sign of the single skewness parameter $b$ or
$a$ is not identifiable either; hence, we restrict attention to $b>0$ or $a>0$.
Second, the single skewness parameter is not identifiable when $\rho=0$, since both 
rank correlations are odd functions of $\rho$ and therefore is zero at $\rho=0$,
regardless of the skewness parameter; see Proposition~\ref{prop:single-skew-MN}~(i)
and Proposition~\ref{prop:single-skew-MSN}~(i).
Third, in the skew-normal scale mixture case, the skewness parameter is not
identifiable when $\rho=\pm1$, because $\tau=\rho_S=\pm1$ at these end points.
Finally, since both $\tau$ and $\rho_S$ are odd functions of $\rho$, it suffices to 
only consider $\rho>0$, as the case $\rho<0$ follows by symmetry.

Given these observations, what remains of interest for the normal location-scale
mixture copulas is the identification of $(\rho,b)$ up to a sign flip of $b$ for
$\rho\in(0,1]$, which can be analyzed by examining the invertibility of
\begin{equation}\label{eq:R-mn-single}
  R^\circ_{\mathrm{mn}}(\rho, b) \coloneqq
  \begin{bmatrix}
    \tau(\rho,(0,b),F) \\ \rho_S(\rho,(0,b),F) 
  \end{bmatrix},
  \qquad 
  \rho\in(0,1], \ b>0,
\end{equation}
with the expressions for $\tau$ and $\rho_S$ given in
Proposition~\ref{prop:single-skew-MN}. 
Similarly, for the skew-normal scale mixture copulas, the remaining question of
interest concerns the identification of $(\rho,a)$ up to a sign flip of $a$ for
$\rho\in(0,1)$, which can be analyzed by examining the invertibility of
\begin{equation}\label{eq:R-msn-single}
  R^\circ_{\mathrm{msn}}(\rho, a) \coloneqq
  \begin{bmatrix}
    \tau(\rho,(0,a),F) \\ \rho_S(\rho,(0,a),F) 
  \end{bmatrix},
  \qquad 
  \rho\in(0,1), \ a>0,
\end{equation}
with the expressions for $\tau$ and $\rho_S$ given in
Proposition~\ref{prop:single-skew-MSN}.

The following proposition provides the explicit expressions of the Jacobians of
$R^\circ_{\mathrm{mn}}(\rho,b)$ and $R^\circ_{\mathrm{msn}}(\rho,a)$.

\begin{proposition}\label{prop:J-single} 
The Jacobian of the function $R^\circ_{\mathrm{mn}}$ defined in
\eqref{eq:R-mn-single} is given by
\[
J^\circ_{\mathrm{mn}}(\rho, b) = 
\begin{bmatrix}
  4\, \EE\,\phi_2(0,bV;\rho) & 
  2\, \EE\,V\phi^{\mathrm s}(bV; \alpha^\circ(\rho)) \\
  12\,\EE\,V_3 \phi_2(0,bV_2;\rho V_3) & 
  6\, \EE\,V_1 \phi^{\mathrm s}(bV_1; \alpha^\circ(\rho V_3))
\end{bmatrix},
\quad \rho\in(0,1),\ b>0,
\]
where $\alpha^\circ(r)=-r(1-r^2)^{-1/2}$
and $V,V_1,V_2,V_3$ are defined in Theorem~\ref{thm:MN-tau-S}.
The Jacobian of the function $R^\circ_{\mathrm{msn}}$ defined in
\eqref{eq:R-msn-single} is given by
\[
  J^\circ_{\mathrm{msn}}(\rho,a) = 4
  \begin{bmatrix}
  \gamma_1 \EE\, h(Z;a,\rho^\dagger)
  & 
  \gamma_2 \EE\, h(Z;a,\rho^\dagger) 
  + \EE\, c(Z;a,\rho^\dagger)
  \\
  3 \gamma_1 \EE\, \tilde h(Z_1,Z_2,V_3;a,\rho^\dagger)
  &
  3 \left(\gamma_2 \EE\,\tilde h(Z_1,Z_2,V_3;a,\rho^\dagger)
  + \EE\,\tilde c(Z_1,Z_2,V_3;a,\rho^\dagger) \right)
  \end{bmatrix}
\]
for $\rho\in(0,1)$ and $a>0$,
where $\rho^\dagger$ is defined in \eqref{eq:ra-dagger-single}, 
$\gamma_1\coloneqq (1+a^2)[1+a^2(1-\rho^2)]^{-3/2}$,
$\gamma_2\coloneqq -a\rho(1-\rho^2)[1+a^2(1-\rho^2)]^{-3/2}$,
$(Z,Z_1,Z_2)$ and $V_3$ are defined in Corollary~\ref{cor:MSN-tau-S-Phi2}
and Theorem~\ref{thm:MSN-tau-S}~(ii), 
and
\begin{align*}
  h(Z;a,\rho^\dagger) 
  &=
  a f(Z;a,\rho^\dagger) + \phi_2(aZ, a\rho^\dagger Z;\rho^\dagger),
  \qquad
  c(Z;a,\rho^\dagger)
  =
  \rho^\dagger f(Z; a,\rho^\dagger) + \frac12Z\phi(aZ),
  \\
  \tilde h(Z_1,Z_2,V_3;a,\rho^\dagger) 
  &=
  a \tilde f(Z_1,Z_2,V_3;a,\rho^\dagger) + V_3 \phi_2(aZ_1, a\rho^\dagger Z_2;\rho^\dagger
  V_3),
  \\
  \tilde c(Z_1,Z_2,V_3;a,\rho^\dagger) 
  &=
  \rho^\dagger \tilde f(Z_1,Z_2,V_3;a,\rho^\dagger)
  + Z_1\phi(aZ_1)\, 
  \Phi\left( \frac{a\rho^\dagger(Z_2 - Z_1V_3)}{\sqrt{1-(\rho^\dagger
  V_3)^2}}\right),
\end{align*}
with 
$
f(Z;a,\rho^\dagger) 
= Z \phi(a\rho^\dagger Z) \Phi\big( a\sqrt{1-\rho^{\dagger2}} Z \big)
$
and
$
\tilde f(Z_1,Z_2,V_3;a,\rho^\dagger) 
= Z_2 \phi(a\rho^\dagger Z_2)
\Phi
\big(
  a(Z_1 - \rho^{\dagger2} Z_2 V_3)\big[1-(\rho^\dagger V_3)^2\big]^{-1/2}
\big)
$.
\end{proposition}

\begin{figure}[t!]\centering
  \includegraphics[width=0.95\textwidth]{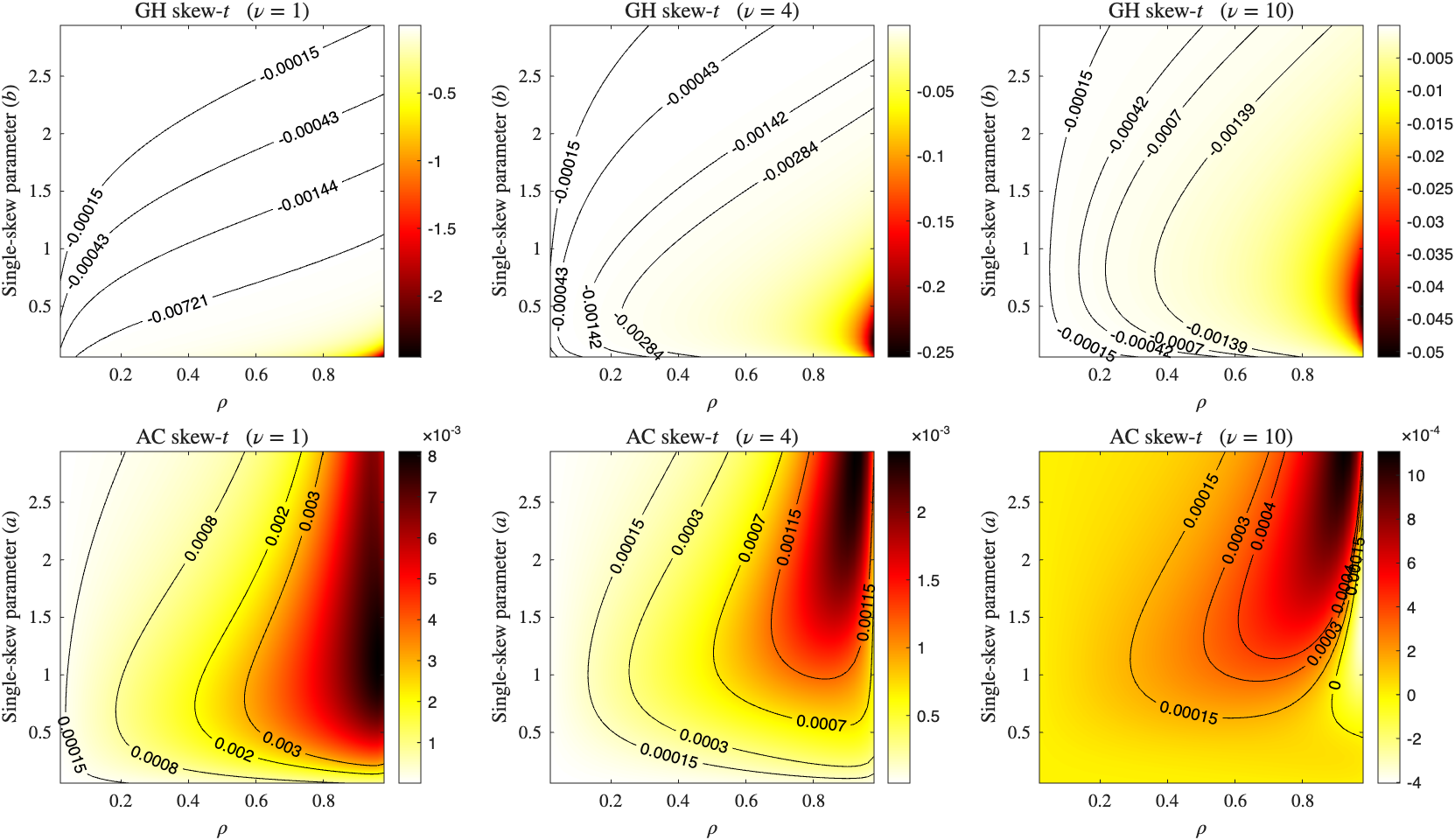}
  \caption{Contour plots of $\det J^\circ_{\mathrm{GHt}}(\rho,b)$ for GH skew-$t$
  copulas with single-skewness $b\in(0,3)$ and $\rho\in(0,1)$ in the upper panels,
  and $\det J^\circ_{\mathrm{ACt}}(\rho,a)$ for AC skew-$t$ copulas
  with single-skewness $a\in(0,3)$ and $\rho\in(0,1)$ in the lower panels.}
  \label{fig:Jacobian_single_skew}
\end{figure}

Similar to the previous case, the invertibility of 
$J^\circ_{\mathrm{mn}}(\rho,b)$ and $J^\circ_{\mathrm{mn}}(\rho,a)$ cannot, in
general, be easily assessed from the formulas in the above proposition.
However, numerical integration can be used to compute these Jacobians and hence
their determinants at specific parameter values.

Fig.~\ref{fig:Jacobian_single_skew} presents contour plots of the determinants of the
Jacobians for the GH skew-$t$ copulas $\det J^\circ_{\mathrm{GHt}}(\rho,b)$ and for
the AC skew-$t$ copulas $\det J^\circ_{\mathrm{ACt}}(\rho,b)$, under single-skewness
and known mixing distribution $\mathrm{IG}(\nu/2,\nu/2)$ for $\nu=1,4$ and $10$.
The plots in Fig.~\ref{fig:Jacobian_single_skew} show that
$|\det J^\circ_{\mathrm{GHt}}(\rho,b)|>0$
across the entire parameter region for all three values of $\nu$, whereas 
$|\det J^\circ_{\mathrm{ACt}}(\rho,b)|>0$ 
holds throughout the region only for $\nu=1$ and $4$. 
When $\nu=10$, $\det J^\circ_{\mathrm{ACt}}(\rho,a)$ can reach zero.
The contour plots also highlight regions where the Jacobian is nearly singular,
indicating weak invertibility of the rank-correlation mapping.
As in the equi-skew case, increasing $\nu$ generally enlarges the region where $|\det
J_{\mathrm{GHt}}(\rho,b)|$ exceeds a small threshold, while the corresponding region
where $|\det J_{\mathrm{ACt}}(\rho,a)|$ exceeds the same threshold tends to shrink
with larger $\nu$.
Moreover, for both GH and AC skew-$t$ copulas, weak invertibility 
tends to occur when $\rho$ is close to zero, which is consistent with the fact that
the skewness parameter is not identified when $\rho=0$.

\section{Concluding remarks} \label{sec:conclusion}

We derived explicit formulas for Kendall's tau and Spearman's rho for two
broad classes of asymmetric copulas, with the widely used GH skew-$t$, AC skew-$t$
and skew-normal copulas arising as special cases.
These tractable results rely on the special properties of the
multivariate normal distribution, the building block of both model classes.
In particular, the multivariate normal is the only elliptical family for which uncorrelated
components are independent. This property is crucial for the stochastic
representation used to derive our formulas for skew-normal scale mixture copulas,
as highlighted in the proofs of Lemmas~\ref{lemma:orthant-prob-tau} and
\ref{lemma:orthant-prob-S} (which underpin Theorem~\ref{thm:MSN-tau-S}).
Although skew-normal scale mixtures already form a broad subclass of skew-elliptical
models (see \cite{BrancoDey2001}), extending these results further is challenging because this
property is not available in general.

Our analysis reveals several important findings regarding how asymmetry affects 
rank correlations in different models.
In particular, we show that, unlike for skew-normal scale mixture copulas,
the presence of asymmetry in normal location-scale mixture copulas restricts
the attainable range of Kendall's tau and Spearman's rho to a \emph{strict} subset of
$[-1,1]$. 
Specifically, pushing the pseudo-correlation $\rho$ toward $\pm1$ does not
necessarily yield rank correlations close to $\pm1$ unless the skewness parameters
align (e.g. $\beta_1=\beta_2$ for $\rho=1$ and $\beta_1=-\beta_2$ for $\rho=-1$). 
This finding has important implications for both interpretation and application of this
class of asymmetric copulas. 
It indicates that extreme linear correlation and extreme rank concordance need not
coincide under asymmetry. In practice, this restriction may complicate rank-based
estimation: inversion procedures with predetermined asymmetry parameters may fail
when the empirical $\tau$ or $\rho_S$ lies outside the model's attainable set, and
model selection should therefore account for these range limitations. 
Empirical studies employing this class of copulas, including the GH skew-$t$ copula
as a special case, are advised to report the estimated parameters together with
the implied attainable interval for $\tau$ and $\rho_S$, to clarify whether observed
dependence levels are compatible with the chosen model specification.

\section{Mathematical proofs} \label{sec:proofs}

\subsection{Useful lemmas} \label{sec:lemma}

We first present some useful lemmas that will be used in the proofs in
Sections~\ref{sec:proof-sec4} and \ref{sec:proof-sec5}.

\begin{lemma}\label{lemma:tau-rhoS-X}
  Let $(X_1, X_2)$ be a bivariate random vector with continuous marginal distribution
  functions. Then,
  \smallskip
  
  \noindent \emph{(i)}
  Kendall's tau of $X_1$ and $X_2$ is given by 
  $\tau = 4\,\PP\{X_1<X^\star_1,\,X_2<X^\star_2\}-1$,
  where $(X^\star_1,X^\star_2)$ is an independent copy of $(X_1,X_2)$;
  \smallskip
  
  \noindent \emph{(ii)}
  Spearman's rho of $X_1$ and $X_2$ is given by
  $\rho_S=12\,\PP\{X_1<X^\circ_1,\,X_2<X^\circ_2\}-3$,
  where $(X^\circ_1,X^\circ_2)$ is a pair of independent random variables independent
  of $(X_1,X_2)$ and having the same marginal distribution functions.
\end{lemma}
  
\begin{proof}[\textbf{\upshape Proof of Lemma~\ref{lemma:tau-rhoS-X}:}]
  See the proof of Proposition 5.29 in \cite{McNeilFreyEmbrechts2005}.
\end{proof}
  
\begin{lemma}\label{lemma:bi-normal-diag}
  For $x\in\RR$ and $\rho\in(-1,1]$, we have $\Phi_2(x,x;\rho)=\Phi^{\mathrm
  s}(x;\alpha_\rho)$, where $\Phi^{\mathrm s}(\,\cdot\,;\alpha_\rho)$ is the univariate
  standard skew-normal cdf with $\alpha_\rho=\sqrt{(1-\rho)/(1+\rho)}>0$.
\end{lemma}

\begin{proof}[\textbf{\upshape Proof of Lemma~\ref{lemma:bi-normal-diag}:}]
  Suppose the random vector $(X_1,X_2)$ has cdf $\Phi_2(\cdot;\rho)$ with
  $\rho\in(-1,1)$. Then, it follows that 
  $\Phi_2(x,x;\rho) = 2\PP\{X_1\le X_2,X_2\le x\}
  = 2\int_{-\infty}^x\PP\{X_1\le y \given X_2=y\}\phi(y) \diff y$.
  Moreover, since $(X_1 \given X_2=y)\sim\mathrm N(\rho y,1-\rho^2)$, we have
  $
  \PP\{X_1\le y \given X_2=y\}
  =\PP
  \big\{(X_1-\rho y)/\sqrt{1-\rho^2}\le (1-\rho)y/\sqrt{1-\rho^2} \,\big|\, X_2=y \big\}
  =\Phi(\alpha_\rho y).
  $
  It then follows that
  $
  \Phi_2(x,x;\rho)
  =\int_{-\infty}^x2\phi(y)\Phi(\alpha_\rho y) \diff y
  =\int_{-\infty}^x\phi^{\mathrm s}(y;\alpha_\rho) \diff y = \Phi^{\mathrm s}(x;\alpha_\rho) 
  $.
  Lastly, if $\rho=1$, then $\alpha_\rho=0$.
  Finally, it can be easily shown that $\Phi_2(x,y;1)=\Phi(\min\{x,y\})$, from which
  it follows that $\Phi_2(x,x;1)=\Phi(x)=\Phi^{\mathrm{s}}(x;0)$ for any $x\in\RR$.
\end{proof}

\begin{lemma}\label{lemma:positive}
  Suppose $X$ is a real valued symmetrically distributed random variable with a
  well-defined pdf and finite first moment, and $\phi^{\mathrm s}(\,\cdot\,;\alpha)$ is
  the pdf of the univariate standard skew-normal distribution with $\alpha>0$. For
  any $b>0$,
  \smallskip
  
  \noindent \emph{(i)}
  $\EE\, X\phi^{\mathrm s}(bX;\alpha) > 0$;
  \smallskip
  
  \noindent \emph{(ii)}
  $\EE\, X\phi^{\mathrm s}(bX;\alpha)$ is strictly increasing in $\alpha$, if,
  additionally, $0<\EE X^2<\infty$. 
\end{lemma}

\begin{proof}[\textbf{\upshape Proof of Lemma~\ref{lemma:positive}:}]
  Let $f$ denote the pdf of $X$, which is a positive even function on $\RR$. 
  Then, by the change of variables,
  \begin{equation} \label{eq:EXphi}
    \EE\, X\phi^{\mathrm s}(bX;\alpha)
    =\int_{-\infty}^\infty x\phi^{\mathrm s}(bx;\alpha)f(x)\diff x
    =b^{-2}\int_{-\infty}^\infty yf(y/b)\phi^{\mathrm s}(y;\alpha)\diff y
    =b^{-2}\,\EE\, g(Y),
  \end{equation}
  where $Y$ is a univariate standard skew-normal random variable with skewness $\alpha$, and $g(x)\coloneqq xf(x/b)$ for $x\in\RR$. Since $\alpha>0$, we have $\EE g(Y)>0$, as $Y$ is strictly positively skewed and $g$ is an odd function on $\RR$. This proves part~(i) of the lemma.
  To show part~(ii) of the lemma, note that if $0<\EE\,X^2<\infty$, then taking partial
  derivative with respect to $\alpha$ on both sides of the first equality in \eqref{eq:EXphi} yields
  (noticing that $\phi^{\mathrm s}(bx;\alpha)=2\phi(bx)\Phi(\alpha b x)$)
  \[
    \frac{\partial}{\partial\alpha}\,\EE X\phi^{\mathrm s}(bX;\alpha)
    =
    2b\int_{-\infty}^\infty x^2\phi(bx)\phi(\alpha bx)f(x)\diff x
    =
    2b\,\EE X^2\phi(bX)\phi(\alpha bX)>0,
  \]
  as desired to be shown.
\end{proof}

\begin{lemma}\label{lemma:orthant-prob-tau}
  Let $(X_1,X_2)\sim\mathrm{SN}(0,\varrho(\rho),\alpha)$ where $\rho\in(-1,1)$ and
  $\alpha\in\RR^2$, and $(X^\star_1,X^\star_2)$ be an independent copy of
  $(X_1,X_2)$. Let $a, b$ be two strictly positive constants. We have
  \[
  \PP\{aX^\star_1>bX_1,\,aX^\star_2>bX_2\} 
  = 4\,\Phi_4\big(0;\mathrm P_\tau(\rho,\delta,v)\big),
  \]
  where $\mathrm P_\tau$ is defined by \eqref{eq:P-tau},
  $\delta=(\delta_1,\delta_2)^\top=\varrho(\rho)\alpha/\sqrt{1+\alpha^\top\varrho(\rho)\alpha}$,
  and $v=(v_1,v_2)$ with $v_1=a/\sqrt{a^2+b^2}$ and $v_2=-b/\sqrt{a^2+b^2}$. 
\end{lemma}

\begin{proof}[\textbf{\upshape Proof of Lemma~\ref{lemma:orthant-prob-tau}:}]
  Let $Z=(Z_1,\ldots,Z_6)^\top$ be a $6$-dimensional normal random vector with mean zero and covariance matrix (upper-triangular entries omitted)
  \begin{align*}
  \mathrm P_Z=\left[
  \begin{array}{ccc|ccc}
  1&&&&&\\
  \rho&1&&&&\\
  \delta_1&\delta_2&1&&&\\\cline{1-6}
  0&0&0&1&& \\   
  0&0&0&\rho&1&\\   
  0&0&0&\delta_1&\delta_2&1
  \end{array}\right]
  =\begin{bmatrix}\mathrm P_3&\mathrm 0\\\mathrm0&\mathrm P_3\end{bmatrix},\quad \mbox{where}\quad
  \mathrm P_3=
  \begin{bmatrix} \varrho(\rho)&\delta\\\delta^\top&1 \end{bmatrix}.
  \end{align*}
  By the unique property of the multivariate normal distribution where uncorrelated
  components are independent, we have $(Z_1,Z_2,Z_3)\perp(Z_4,Z_5,Z_6)$.
  Furthermore, using the stochastic representation of the
  skew-normal distribution via conditioning \citep[][Section~5.1.3]{Azzalini2014}, we
  can represent the random vectors $(X_1,X_2)$ and $(X_1^\star,X_2^\star)$ satisfying
  the assumptions in the lemma as:
  $(X_1,X_2)=_d(Z_4,Z_5)\given\{Z_6>0\}$ and
  $(X^\star_1,X^\star_2)=_d(Z_1,Z_2)\given\{Z_3>0\}$. Next, for $a,b>0$, we
  define
  \begin{align}
    \begin{bmatrix}Y_1\\Y_2\end{bmatrix}
    =_d
    \left.\left(a\begin{bmatrix}Z_1\\Z_2\end{bmatrix}-b\begin{bmatrix}Z_4\\Z_5\end{bmatrix}\right)
    \,\right|\,\{Z_3>0,Z_6>0\}.\label{eq:Y}
  \end{align}
  Then, to prove the lemma it suffices to show that $\PP\{Y_1>0,\,Y_2>0\}=4\,\Phi_4(0;\mathrm P_\tau(\rho,\delta,v))$. To see this, we define a $4$-dimensional normal random vector $W=(W_1,\ldots,W_4)^\top$ by
  \[
    W=\begin{bmatrix} \mathrm A&\mathrm B\\\mathrm C&\mathrm D \end{bmatrix} Z
    \quad\mbox{where}\quad
    \mathrm A=\begin{bmatrix}a&0&0\\0&a&0\end{bmatrix},\quad
    \mathrm B=\begin{bmatrix}-b&0&0\\0&-b&0\end{bmatrix},\quad
    \mathrm C=\begin{bmatrix}0&0&1\\0&0&0\end{bmatrix}, \quad \mbox{and}\quad 
    \mathrm D=\begin{bmatrix}0&0&0\\0&0&1\end{bmatrix}.
  \]
  so that
  \begin{align}\label{eq:Ws}
  W_1=aZ_1-bZ_4,\quad W_2=aZ_2-bZ_5,\quad W_3=Z_3,\quad W_4=Z_6.
  \end{align}
  By construction, $W\sim\mathrm N(0,\Sigma_W)$, where
  \[
  \Sigma_W
  =\left[
  \begin{array}{cc|cc}
  a^2+b^2&(a^2+b^2)\rho&a\delta_1&-b\delta_1\\ (a^2+b^2)\rho&a^2+b^2&a\delta_2&-b\delta_2\\\cline{1-4}
  a\delta_1&a\delta_2&1&0\\-b\delta_1&-b\delta_2&0&1
  \end{array}
  \right].
  \]
  By \eqref{eq:Y} and \eqref{eq:Ws}, we have
  $(Y_1,Y_2)=_d(W_1,W_2)\given\{W_3>0,W_4>0\}$, and consequently,
  \[
  \PP\{Y_1>0,\,Y_2>0\}
  =\PP\{W_1>0,W_2>0\given W_3>0,W_4>0\}
  =\frac{\PP\{W_1>0,W_2>0,W_3>0,W_4>0\}}{\PP\{W_3>0,W_4>0\}}
  =4\,\Phi_4(0;\varrho(\Sigma_W)).
  \]
  The proof is then complete upon noticing that $\varrho(\Sigma_W)=\mathrm
  P_\tau(\rho,\delta,v)$.
\end{proof}

\begin{lemma}\label{lemma:orthant-prob-S}
  Let $(X_1,X_2)\sim\mathrm{SN}(0,\varrho(\rho),\alpha)$ where $\rho\in(-1,1)$ and
  $\alpha\in\RR^2$. Let $(X^\circ_1,X^\circ_2)$ be a bivariate skew-normal vector
  satisfying $(X^\circ_1,X^\circ_2)\perp (X_1,X_2)$, $X^\circ_1\perp X^\circ_2$ and
  $X^\circ_i=_dX_i$ for i=1,2. Moreover, let $a_1,a_2,b$ be three strictly
  positive constants and $\mathrm P_S$ be defined by \eqref{eq:P-S}. We have
  \[
  \PP\{a_1X^\circ_1>bX_1,\, a_2X^\circ_2>bX_2\}
  =8\,\Phi_5\big(0;\mathrm P_S(\rho,\delta,v)\big),
  \]
  where
  $\delta=(\delta_1,\delta_2)^\top=\varrho(\rho)\alpha/\sqrt{1+\alpha^\top\varrho(\rho)\alpha}$,
  $v=(v_1,v_2,v_1^-,v_2^-,v_3)$ with
  $v_1=a_1/\sigma_1$, $v_2=a_2/\sigma_2$, $v^-_1=-b/\sigma_1$, $v^-_2=-b/\sigma_2$
  and $v_3=b^2/(\sigma_1\sigma_2)$, where $\sigma_1=\sqrt{a_1^2+b^2}$ and
  $\sigma_2=\sqrt{a_2^2+b^2}$.
\end{lemma}

\begin{proof}[\textbf{\upshape Proof of Lemma~\ref{lemma:orthant-prob-S}:}]
  Let $Z=(Z_1,\ldots,Z_7)^\top$ be a normal random vector with mean zero and
  covariance matrix 
  \[
    \mathrm P_Z
    =
    \begin{bmatrix}\mathrm B&\mathrm C\\\mathrm C^\top&\mathrm I_3\end{bmatrix},
    \qquad
    \mathrm B
    =
    \begin{bmatrix}\mathrm I_2&\mathrm 0\\\mathrm 0&\varrho(\rho)\end{bmatrix},
    \qquad
    \mathrm C^\top
    =
    \begin{bmatrix}\mbox{diag}(\delta)&\mathrm 0\\\mathrm 0&\delta^\top\end{bmatrix}.
  \]
  By the unique property of the multivariate normal distribution where uncorrelated
  components are independent, we have $(Z_3,Z_4,Z_7)\perp(Z_1,Z_2,Z_5,Z_7)$ and 
  $(Z_1,Z_5)\perp(Z_2,Z_6)$.
  Furthermore, using the stochastic representation of the skew-normal distribution via
  conditioning \citep[][Section 5.1.3]{Azzalini2014}, we can represent the
  random vectors $(X_1,X_2)$ and $(X_1^\circ,X_2^\circ)$ satisfying the assumptions in
  the lemma as 
  $(X_1,X_2)=_d(Z_3,Z_4)\given\{Z_7>0\}$, $X_1^\circ=_d(Z_1\given Z_5>0)$ and
  $X_2^\circ=_d(Z_2\given Z_6>0)$. Next, for $a_1,a_2,b>0$, define
  \begin{equation}\label{eq:Y-S}
    \begin{bmatrix}Y_1\\Y_2\end{bmatrix}
    =_d
    \left.\left(\begin{bmatrix}a_1Z_1\\a_2Z_2\end{bmatrix}-b\begin{bmatrix}Z_3\\Z_4\end{bmatrix}\right)\,\right|\,\{Z_5>0,Z_6>0,Z_7>0\}.
  \end{equation}
  To prove the lemma it suffices to show $\PP\{Y_1>0,Y_2>0\}=8\Phi_5(0;\mathrm
  P_S(\rho,\delta,v))$. To see this, define a normal random vector
  $W=(W_1,\ldots,W_5)^\top$ by
  \[
    W =
    \begin{bmatrix} \mathrm A&\mathrm 0\\\mathrm 0&\mathrm I_3 \end{bmatrix}Z,
    \qquad
    \mathrm A =
    \begin{bmatrix}a_1&0&-b&0\\0&a_2&0&-b\end{bmatrix} = 
    \begin{bmatrix}\mbox{diag}(a_1,a_2)&-b\mathrm I_2\end{bmatrix},
  \]
  so that
  \begin{equation}\label{eq:Ws-S}
    W_1 = a_1Z_1-bZ_3,\quad
    W_2 = a_2Z_2-bZ_4,\quad
    W_3 = Z_5,\quad
    W_4 = Z_6,\quad
    W_5 = Z_7.
  \end{equation}
  By construction, $W\sim\mathrm N(0,\Sigma_W)$, where
  \[
    \Sigma_W
    =
    \begin{bmatrix}\mathrm A\mathrm B\mathrm A^\top&\mathrm A\mathrm C\\\mathrm
    C^\top \mathrm A^\top&\mathrm I_3 \end{bmatrix}
    =
    \left[
    \begin{array}{cc|ccc}
    a_1^2+b^2&b^2\rho&a_1\delta_1&0&-b\delta_1\\
    b^2\rho&a_2^2+b^2&0&a_2\delta_2&-b\delta_2\\ \cline{1-5}
    a_1\delta_1&0&1&0&0\\
    0&a_2\delta_2&0&1&0\\
    -b\delta_1&-b\delta_2&0&0&1
    \end{array}\right].
  \]
  By \eqref{eq:Y-S} and \eqref{eq:Ws-S}, 
  $(Y_1,Y_2)=_d(W_1,W_2)\given\{W_3>0,W_4>0,W_5>0\}$, and then
  \[
    \PP\{Y_1>0,Y_2>0\}
    = \PP\{W_1>0,W_2>0\given W_3>0,W_4>0,W_5>0\}
    = \frac{\PP\{W_1>0,W_2>0,W_3>0,W_4>0,W_5>0\}}{\PP\{W_3>0,W_4>0,W_5>0\}},
  \]
  which is equal to $8\,\Phi_5(0;\varrho(\Sigma_W))$.
  The proof is then complete upon noticing that $\varrho(\Sigma_W)=\mathrm P_S(\rho,\delta,v)$.
\end{proof}

\begin{lemma}\label{lemma:Phi_derivative}
  Let $\mathrm P$ be a $d\times d$ correlation matrix whose upper-left $2\times 2$
  sub-matrix is given by $\varrho(\rho)$ with $\rho\in(-1,1)$. Then, we have
  \[
    \frac{\partial\Phi_d(0;\mathrm P)}{\partial\rho}=
    \begin{cases}
      \big(2\pi\sqrt{1-\rho^2}\big)^{-1} >0, & d=2, \\
      \big(2\pi\sqrt{1-\rho^2}\big)^{-1}\,
      \Phi_{d-2}\big(0;\mathrm P/\varrho(\rho)\big) >0, & d>2,
    \end{cases}
  \]
  where $\mathrm P/\varrho(\rho)$ is the Schur complement of $\varrho(\rho)$ in
  $\mathrm P$.
\end{lemma}

\begin{proof}[\textbf{\upshape Proof of Lemma~\ref{lemma:Phi_derivative}:}]
  In the case of $d=2$, we have $\mathrm P=\varrho(\rho)$ and hence $\Phi_2(0;\mathrm
  P)=\Phi_2(0;\rho)$ and $\phi_2(x;\mathrm P)=\phi_2(x;\rho)$. Then, 
  \[
    \frac{\partial\Phi_2(0;\rho)}{\partial\rho}
    = \int_{(-\infty,0]^2}\frac{\partial^2\phi_2(x;\rho)}{\partial x_1\partial x_2}\diff x
    = \phi_2(0,0;\rho)=\frac{1}{2\pi\sqrt{1-\rho^2}}>0,
  \]
  where the second equality above holds due to
  $\partial\phi(x;\rho)/\partial\rho=\partial^2\phi(x;\rho)/\partial x_1\partial x_2$
  \citep[][Equation~(3)]{Plackett1954}. If $d>2$, note that
  \[
    \frac{\partial\Phi_d(0;\mathrm P)}{\partial\rho}
    =
    \int_{(-\infty,0]^{d-2}}\phi_d(0,0,x_3,\ldots,x_d;\mathrm P)\diff x_3\cdots \diff x_d,
  \]
  and then, by the conditioning property of multivariate normal distributions, 
  \[
    \phi_d(0,0,x_3,\ldots,x_d;\mathrm P)    
    = \phi_2(0,0;\rho)\phi_{d-2}(x_3,\ldots,x_d;\mathrm P/\varrho(\rho))
    = \frac{\phi_{d-2}(x_3,\ldots,x_d;\mathrm P/\varrho(\rho))}{2\pi\sqrt{1-\rho^2}}.
  \]
  This leads to the claim in the lemma for $d>2$ case.
\end{proof}

\begin{lemma}\label{lemma:Phi4-derivative-delta}
Let $\mathrm P_\tau(\rho,\delta,v)$ be the correlation matrix-valued function defined in \eqref{eq:P-tau}, where $\rho\in(-1,1)$, $\delta=(\delta_1,\delta_2)$ and $v=(v_1,v_2)$, and let $V=(V_1,V_2)$ be the random vector defined in Theorem~\ref{thm:MSN-tau-S}~(i). We have 
\[
  \frac{\partial\Phi_4(0;\mathrm P_\tau(\rho,\delta,V))}{\partial\delta_1}
  =_d
  \frac{V_1\arcsin \tilde V(\delta_1,\delta_2)}{2\pi^2\sqrt{1-V_1^2\delta_1^2}},
  \qquad
  \frac{\partial\Phi_4(0;\mathrm P_\tau(\rho,\delta,V))}{\partial\delta_2}
  =_d
  \frac{V_1\arcsin \tilde V(\delta_2,\delta_1)}{2\pi^2\sqrt{1-V_1^2\delta_2^2}},
\]
where
\[
  \tilde V(\delta_1,\delta_2)
  \coloneqq \frac{V_2(\delta_2-\rho \delta_1)}{\sqrt{1-\delta_1^2}\sqrt{1-\rho^2-(\delta_1^2-2\rho\delta_1\delta_2+\delta_2^2)V_1^2}}.
\]
\end{lemma}

\begin{proof}[\textbf{\upshape Proof of Lemma~\ref{lemma:Phi4-derivative-delta}:}]
  Consider the following four matrices:
  \[
    \mathrm P_a=\left[
    \begin{array}{cc|cc} 1&&&\\a&1&&\\\cline{1-4} \rho&b&1&\\c&0&d&1 \end{array}\right],\quad
    \mathrm P_b=\left[
    \begin{array}{cc|cc} 1&&&\\b&1&&\\\cline{1-4} \rho&a&1&\\d&0&c&1 \end{array}\right], \quad
    \mathrm P_c=\left[
    \begin{array}{cc|cc} 1&&&\\c&1&&\\\cline{1-4} a&0&1&\\\rho&d&b&1 \end{array}\right],\quad
    \mathrm P_d=\left[
    \begin{array}{cc|cc} 1&&&\\d&1&&\\\cline{1-4} b&0&1&\\\rho&c&a&1 \end{array}\right],
  \]
  where $a=v_1\delta_1$, $b=v_1\delta_2$, $c=v_2\delta_1$ and $d=v_2\delta_2$.
  It is clear that 
  $\Phi_4(0;\mathrm P_\tau(\rho,\delta,v))=\Phi_4(0;\mathrm P_a)=\Phi_4(0;\mathrm
  P_b)=\Phi_4(0;\mathrm P_c)=\Phi_4(0;\mathrm P_d)$, and, by
  total differentiation and Lemma~\ref{lemma:Phi_derivative}, we can deduce that
  \begin{align}
    \frac{\partial}{\partial\delta_1}\Phi_4(0;\mathrm P_\tau)
    &=\frac{v_1\Phi_2\big(0;\mathrm P_a/\varrho(a)\big)}{2\pi\sqrt{1-a^2}}
    +\frac{v_2\Phi_2\big(0;\mathrm P_c/\varrho(c)\big)}{2\pi\sqrt{1-c^2}},\label{eq:Phi4-d1}\\
    \frac{\partial}{\partial\delta_2}\Phi_4(0;\mathrm P_\tau)
    &=\frac{v_1\Phi_2\big(0;\mathrm P_b/\varrho(b)\big)}{2\pi\sqrt{1-b^2}}
    +\frac{v_2\Phi_2\big(0;\mathrm P_d/\varrho(d)\big)}{2\pi\sqrt{1-d^2}}.\label{eq:Phi4-d2}
  \end{align}
  Some matrix algebra yields
  \begin{equation} \label{eq:Phi2-abcd}
    \Phi_2\big(0;\mathrm P_s/\varrho(s)\big)
    =
    \frac14+\frac{\arcsin u_s(a,b,c,d)}{2\pi},
    \qquad s\in\{a,b,c,d\},
  \end{equation}
  where
  \begin{align*}
    u_a(a,b,c,d)
    &=\frac{d(1-a^2)-c(\rho-ab)}{\sqrt{1-a^2-c^2}\sqrt{1-\rho^2-(a^2-2\rho ab+b^2)}},\ \
    u_b(a,b,c,d)
    =\frac{c(1-b^2)-d(\rho-ab)}{\sqrt{1-b^2-d^2}\sqrt{1-\rho^2-(a^2-2\rho ab+b^2)}},\\
    u_c(a,b,c,d)
    &=\frac{b(1-c^2)-a(\rho-cd)}{\sqrt{1-a^2-c^2}\sqrt{1-\rho^2-(c^2-2\rho cd+d^2)}}, \ \ 
    u_d(a,b,c,d)
    =\frac{a(1-d^2)-b(\rho-cd)}{\sqrt{1-b^2-d^2}\sqrt{1-\rho^2-(c^2-2\rho cd+d^2)}}.
  \end{align*}
  Using the fact that $(V_1,V_2)=_d -(V_2,V_1)$, we can verify: 
  \begin{align}
    \tilde V(\delta_1,\delta_2)=u_a(V_1\delta_1,V_1\delta_2,V_2\delta_1,V_2\delta_2)
    &=_d
    -u_c(V_1\delta_1,V_1\delta_2,V_2\delta_1,V_2\delta_2),\label{eq:ua}\\
    \tilde V(\delta_2,\delta_1)=u_b(V_1\delta_1,V_1\delta_2,V_2\delta_1,V_2\delta_2)
    &=_d
    -u_d(V_1\delta_1,V_1\delta_2,V_2\delta_1,V_2\delta_2).\label{eq:ub}
  \end{align}
  Noticing $V_1=_d-V_2$, the first statement in the lemma holds due to
  \eqref{eq:Phi4-d1}, \eqref{eq:Phi2-abcd}, \eqref{eq:ua}, and the second statement
  holds due to \eqref{eq:Phi4-d2}, \eqref{eq:Phi2-abcd}, \eqref{eq:ub}.
\end{proof}

\subsection{Proofs of the results in Section~\ref{sec:rank-corr-MN}}
\label{sec:proof-sec4}

\smallskip

\begin{proof}[\textbf{\upshape Proof of Theorem \ref{thm:MN-tau-S}:}] 
  Since $X_1$ and $X_2$ have continuous cdf's, their Kendall's tau and Spearman's rho
  depend solely on their copula.
  It therefore suffices to consider, without loss of generality,
  $X=(X_1,X_2)^\top\sim\mathrm{MN}(0,\varrho(\rho),\beta,F)$, which admits the
  representation $X=W\beta+\sqrt{W}Z$ for some $Z\sim\mathrm N(0,\varrho(\rho))$,
  $W\sim F$, and $W\perp Z$.
  
  To derive Kendall's tau of $X_1$ and $X_2$ in part (i), we define
  $X^\star\coloneqq W^\star\beta+\sqrt{W^\star}Z^\star$ where $W^\star$ and $Z^\star$
  are independent copies of $W$ and $Z$ respectively, and $W^\star$ is independent
  of $Z^\star$. By construction, $X^\star$ is an independent copy of $X$, and
  \begin{align} \label{eq:Y1Y2-0}
  \left.\begin{bmatrix}X^\star\\X\end{bmatrix} \,\right|\,
  \{W,W^\star\} \sim \mathrm N\left(\begin{bmatrix}W^\star\beta\\W\beta\end{bmatrix},\ 
  \begin{bmatrix}W^\star\varrho(\rho)&\mathrm 0\\\mathrm 0&W\varrho(\rho)\end{bmatrix}\right).
  \end{align}
  Now, we define $Y=(Y_1,Y_2)^\top\coloneqq X-X^\star$. Note that, by
  Lemma~\ref{lemma:tau-rhoS-X}~(i), Kendall's tau of $X_1$ and $X_2$ is given by
  $\tau = 4\,\PP\{Y_1<0,Y_2<0\}-1$. To obtain the joint probability of $Y_1<0$ and
  $Y_2<0$, we first deduce from \eqref{eq:Y1Y2-0} that 
  \begin{equation} \label{eq:Y1Y2}
    Y\given\{W,W^\star\} \sim 
    \mathrm N\big((W-W^\star)\,\beta,\ (W+W^\star)\,\varrho(\rho)\big),
  \end{equation}
  and note that $\PP\{Y_1<0,Y_2<0 \given W,W^\star\} = \PP\{S_1<V\beta_1,\,
  S_2<V\beta_2 \given W,W^\star\}$, with
  $S_i:=[Y_i-(W-W^\star)\beta_i]\,/\,\sqrt{W+W^\star}$ for $i=1,2$, and $V$ is
  defined in part~(i) of the theorem.
  Since \eqref{eq:Y1Y2} implies $(S_1,S_2)^\top \given \{W,W^\star\}\sim\mathrm
  N\left(0,\varrho(\rho)\right)$, we have 
  $\PP\{Y_1<0,Y_2<0 \given W,W^\star\} 
  =\PP\{S_1<V\beta_1,\, S_2<V\beta_2 \given W,W^\star\}
  =\Phi_2(V\beta_1,V\beta_2;\rho)$.
  By the law of iterated expectations, $\PP\{Y_1<0,Y_2<0\} = \EE\,\Phi_2(V\beta_1,V\beta_2;\rho)$.
  Part~(i) of the theorem then holds immediately due to Lemma~\ref{lemma:tau-rhoS-X}~(i).
  
  To derive Spearman's rho of $X_1$ and $X_2$ in part (ii), we
  define $X^\circ\coloneqq (X_1^\circ,X_2^\circ)^\top$ with
  $X_i^\circ=\beta_iW_i+\sqrt{W_i}Z^\circ_i$ for $i\in\{1,2\}$,
  where $W_1,W_2 \sim$ i.i.d.$\,F$ and are independent of $W$, and $Z^\circ_1,
  Z^\circ_2\sim$ i.i.d.$\,\mathrm N(0,1)$ and are independent of $W_1,W_2,W$ and $Z$.
  By construction, $X^\circ$ is independent of $X$, has independent marginals, and
  $X_i^\circ =_d X_i$ for $i\in\{1,2\}$. Moreover, 
  \[
    \left.\begin{bmatrix}X^\circ\\X \end{bmatrix}\,\right|\,\{W,W_1,W_2\}
    \sim\mathrm N\left(\,
    \begin{bmatrix}W_1\beta_1\\W_2\beta_2\\W\beta\end{bmatrix}, \ \
    \begin{bmatrix}\mbox{diag}(W_1,W_2) & 0\\ 0 & W\varrho(\rho)\end{bmatrix}\right).
  \]
  Next, define $Y^\circ=(Y^\circ_1,Y^\circ_2)^\top:=X-X^\circ$. By
  Lemma~\ref{lemma:tau-rhoS-X}~(ii), Spearman's tau of $X_1$ and $X_2$ is 
  $\rho_S = 12\,\PP\{Y^\circ_1<0,Y^\circ_2<0\}-3$. To obtain the joint probability of
  $Y^\circ_1<0$ and $Y^\circ_2<0$, we first note that
  $Y^\circ\given\{W,W_1,W_2\}\sim\mathrm N(\mu_Y,\Sigma_Y)$ where 
  \[
  \mu_Y    
  =
  \begin{bmatrix}(W-W_1)\beta_1\\(W-W_2)\beta_2\end{bmatrix},\quad
  \Sigma_Y 
  =
  \begin{bmatrix}W_1&0\\0&W_2\end{bmatrix}+W\begin{bmatrix}1&\rho\\\rho&1\end{bmatrix}
  =
  \begin{bmatrix}W_1+W & W\rho\\ W\rho & W_2+W\end{bmatrix}.
  \]
  Define $S_i^\circ=[Y_i^\circ-(W-W_i)\,\beta_i]\,/\,\sqrt{W+W_i}$ for $i\in\{1,2\}$.
  Then,
  $(S_1^\circ, S_2^\circ)^\top \given \{W,W_1,W_2\}\sim\mathrm N(0,\varrho(V_3\rho))$
  where $V_3=W/\sqrt{(W_1+W)(W_2+W)}$. Moreover, given $V_i\coloneqq (W_i-W)/\sqrt{W_i+W}$
  for $i\in\{1,2\}$, we have
  \[
  \PP\{Y^\circ_1 < 0,Y_2^\circ < 0\given W,W_1,W_2\}
  =\PP\{S^\circ_1 < V_1\beta_1,\, S_2^\circ < V_2\beta_2 \given W,W_1,W_2\}
  =\Phi_2(V_1\beta_1,V_2\beta_2;V_3\rho ).
  \]
  By the law of iterated expectations, $\PP\{Y^\circ_1<0,Y^\circ_2<0\} = \EE\,\Phi_2(V_1\beta_1,V_2\beta_2;V_3\rho )$. Part~(ii) of the theorem then holds immediately due to Lemma~\ref{lemma:tau-rhoS-X}~(ii).
\end{proof}

\begin{proof}[\textbf{\upshape Proof of Proposition \ref{prop:MN-tau-S-properties}:}] 
  Part (i) of the proposition follows directly from the expressions of $\tau$ and
  $\rho_S$ given in Theorem~\ref{thm:MN-tau-S} and the exchangeability of the
  bivariate normal cdf: $\Phi_2(x_1,x_2;\rho)=\Phi_2(x_2,x_1;\rho)$.
  
  Next, we prove part~(ii). First, $\tau(\rho,\beta,F)=\tau(\rho,-\beta,F)$ follows
  directly from the expression of $\tau$ in \eqref{eq:MN-tau} and the fact that $V=_d
  -V$. To show that $\rho_S(\rho,\beta,F)=\rho_S(\rho,-\beta,F)$, it suffices to
  prove that 
  \begin{equation}\label{eq:EPhi2-}
  \EE\,\Phi_2(\beta_1V_1,\beta_2V_2;\rho V_3)=\EE\,\Phi_2(-\beta_1V_1,-\beta_2V_2;\rho V_3),
  \end{equation}
  where the mixing variables $(V_1,V_2,V_3)$ are defined in Theorem~\ref{thm:MN-tau-S}. Evaluating the decomposition of $\Phi_2(\,\cdot\,;\rho)$ in \eqref{eq:Phi2-rep-2} at $(-x_1,-x_2)$ and then using the identity $T(h,a)=T(-h,a)$, see e.g. (B.19) in \citep[][p. 235]{Azzalini2014},
  \begin{align}
  \Phi_2(-x_1,-x_2;\rho)
  = \frac12\,[\Phi(-x_1)+\Phi(-x_2)] - a(-x_1,-x_2) 
  -T\left(x_1,\,\frac{x_2-\rho x_1}{x_1\sqrt{1-\rho^2}}\right)
  -T\left(x_2,\,\frac{x_1-\rho x_2}{x_2\sqrt{1-\rho^2}}\right). \label{eq:Phi2-rep-2-}
  \end{align}
  Since the last two terms in \eqref{eq:Phi2-rep-2-} and in \eqref{eq:Phi2-rep-2} are
  identical, \eqref{eq:EPhi2-} holds if
  $\EE\,\Phi(\beta_iV_i)=\EE\,\Phi(-\beta_iV_i)$ for $i\in\{1,2\}$, and
  $\EE\,a(\beta_1V_1,\beta_2V_2)=\EE\,a(-\beta_1V_1,-\beta_2V_2)$. Since these
  conditions are guaranteed by the fact that $V_i=_d -V_i$ for $i\in\{1,2\}$ and
  by the definition of function $a$, we can show that \eqref{eq:EPhi2-} holds, and
  hence $\rho_S(\rho,\beta,F)=\rho_S(\rho,-\beta,F)$.
  
  To prove part~(iii), first note that $\partial \tau/\partial\rho =
  4\EE\,\phi_2(\beta_1V,\beta_2V;\rho)$ and $\partial\rho_S/\partial\rho =
  12\EE\,V_3\phi_2(\beta_1V_1,\beta_2V_2;\rho V_3)$ by \eqref{eq:dPhi2}. Both
  partial derivatives are strictly positive because the normal pdf $\phi_2$ is
  strictly positive and $V_3$ is a strictly positive random variable. This proves
  part~(iii).
  
  For part (iv), note that Theorem~\ref{thm:MN-tau-S} implies 
  $\tau(0,\beta,F) = 4\EE\Phi_2(\beta_1V, \beta_2V; 0) - 1
   = 4\Cov(\Phi(\beta_1 V), \Phi(\beta_2 V))$ and 
  $\rho_S(0,\beta,F) = 12\EE\Phi_2(\beta_1V, \beta_2V; 0) - 3
   = 12\Cov(\Phi(\beta_1 V), \Phi(\beta_2 V))$.
  So, it suffices to show 
  $\sgn \Cov(\Phi(\beta_1 V), \Phi(\beta_2 V))=\sgn\beta_1\beta_2$.
  First, if $\beta_1\beta_2=0$ (at least one of $\beta_1$ and $\beta_2$ is zero),
  then $\Cov(\Phi(\beta_1 V), \Phi(\beta_2 V))=0$.
  Next, let $V^\star$ be an independent copy of $V$. Using the identity 
  $2\Cov(\Phi(\beta_1 V), \Phi(\beta_2 V))
  = \EE (\Phi(\beta_1 V) - \Phi(\beta_1 V^\star))
  (\Phi(\beta_2 V) - \Phi(\beta_2 V^\star)) $ 
  and the fact that $\Phi$ is strictly increasing, we we can easily see that
  $\Cov(\Phi(\beta_1 V), \Phi(\beta_2 V))$ is strictly positive if and only if
  $\beta_1\beta_2>0$ and strictly negative if and only if $\beta_1\beta_2<0$.
  This completes the proof of part (iv).
  
  Lastly, with part~(iii) in mind, to show parts~(v) and (vi),
  it suffices to show that
  \begin{equation} \label{eq:tau-rhoS}
    \tau(-1,\beta,F) = \rho_S(-1,\beta,F) = -1 \mbox{ if and only if } \beta_1=-\beta_2,
    \quad
    \tau( 1,\beta,F) = \rho_S( 1,\beta,F) =  1 \mbox{ if and only if } \beta_1= \beta_2.
  \end{equation}
  Recall that for continuous margins, $\tau=-1$ (resp. $\tau=1$) 
  if and only if the copula is the Fr\'echet lower bound (resp. Fr\'echet upper bound), 
  which holds if and only if $\rho_S=-1$ (resp. $\rho_S=1$).
  Therefore, to show \eqref{eq:tau-rhoS},
  it suffices to only show the conditions for $\tau$.
  First, by Theorem~\ref{thm:MN-tau-S}, 
  $\tau(-1,\beta,F) = 4\EE\,\Phi_2(\beta_1 V,\beta_2V;-1)-1$.
  For any fixed $v\in\RR$ and $Z\sim N(0,1)$, we have
  $\Phi_2(\beta_1v,\beta_2v\,;1)
  = \PP(Z\leq \beta_1v, -Z\leq \beta_2v)
  = \PP(-\beta_2v\leq Z\leq \beta_1v)\cdot 1\{(\beta_1+\beta_2)v\geq 0\}$.
  Then, it is easy to verify that
  \begin{align*}
  \EE\,\Phi_2(\beta_1 V,\beta_2V;-1)
  &=
  \EE
  \left[\Phi(\beta_1 V)-\Phi(-\beta_2V)\right]
  \mathrm 1\{(\beta_1+\beta_2)V\geq0\} 
  \begin{cases}
  = 0, & \mbox{if } \beta_1+\beta_2=0, \\
  >0,  & \mbox{if } \beta_1+\beta_2 \neq 0.
  \end{cases}
  \end{align*}
  Therefore, $\tau(-1,\beta,F) = -1$ if and only if $\beta_1=-\beta_2$, which proves
  the first statement in \eqref{eq:tau-rhoS} for $\tau$.
  To show the second statement in \eqref{eq:tau-rhoS} for $\tau$, note that
  Theorem~\ref{thm:MN-tau-S} implies that
  $\tau(1,\beta,F) = 4\EE\,\Phi_2(\beta_1 V,\beta_2V;1)-1$.
  For any fixed $v\in\RR$ and $Z\sim N(0,1)$, we have 
  $\Phi_2(\beta_1v,\beta_2v\,;1) 
  = \PP(Z\leq \beta_1v, Z\leq \beta_2v) 
  = \Phi(\min\{\beta_1v,\beta_2v\})$.
  Without loss of generality, assume $\beta_1\leq \beta_2$. Then,
  \begin{align*}
  \EE\,\Phi_2(\beta_1 V,\beta_2V;1)
  &=
  \EE\,\Phi(\min\{\beta_1 V,\beta_2V\})\mathrm 1\{V>0\}
  +
  \EE\,\Phi(\min\{\beta_1 V,\beta_2V\})\mathrm 1\{V<0\}\\
  &=
  \EE\,\Phi(\beta_1 V)\mathrm 1\{V>0\}
  +
  \EE\,\Phi(\beta_2V)\mathrm 1\{V<0\}\\
  &=
  \EE
  \left[\Phi(\beta_1 V)+\Phi(-\beta_2 V)\right]
  \mathrm 1\{V>0\}
  \begin{cases}
  =1/2,  & \mbox{if } \beta_1=\beta_2, \\
  \in(0, 1/2), & \mbox{if } \beta_1\neq \beta_2,
  \end{cases}
  \end{align*}
  where we used the fact that $V=_d-V$ in the last two equalities.
  Therefore, $\tau(1,\beta,F) = 1$ if and only if $\beta_1=\beta_2$, and 
  the second statement in \eqref{eq:tau-rhoS} for $\tau$ is proved.
\end{proof}

\begin{proof}[\textbf{\upshape Proof of Corollary \ref{cor:equi-skewness-MN}:}] 
  The expressions for $\tau$ and $\rho_S$ in the corollary follows directly from 
  Theorem~\ref{thm:MN-tau-S} under the equi-skewness condition.
  For the partial derivatives in part~(i), note that
  \[
  \frac{\partial}{\partial b}\,\tau(\rho,b,F)
  = 4\,\EE\,\frac{\partial}{\partial b}\Phi^{\mathrm s}(bV;\alpha_\rho) 
  = 4\,\EE\,V\phi^{\mathrm s}(bV;\alpha_\rho) >0,
  \]
  where $\alpha_\rho\coloneqq\sqrt{(1-\rho)/(1+\rho)}$, and
  the inequality holds due to Lemma~\ref{lemma:positive}~(i).
  Furthermore, Lemma~\ref{lemma:positive}~(ii) implies 
  $\partial^2\tau/\partial b\partial \alpha_\rho>0$.
  Since $\partial\alpha_\rho/\partial\rho=-(1-\rho)^{-1/2}(1+\rho)^{-3/2}<0$,
  we have $\partial^2\tau/\partial b\partial\rho<0$.
  This completes the proof.
\end{proof}

\begin{proof}[\textbf{\upshape Proof of Proposition~\ref{prop:single-skew-MN}:}] 
  We first prove part~(i) of the proposition. By equation \eqref{eq:Phi2-rep} and a change of variables, 
  \begin{align}
  \Phi_2(x_1,x_2;\rho)
  &=\Phi(x_1)\Phi(x_2)+\frac1{2\pi}\int_0^{\arcsin\rho}\exp\left(-\frac{x_1^2+x_2^2-2x_1x_2\sin\theta}{2\cos^2\theta}\right) \diff\theta, \notag
  \end{align}
  from which we can easily deduce that
  \begin{align}
  \Phi_2(0,x;\rho)+\Phi_2(0,x;-\rho)=\Phi(x), \label{eq:Phi2-2}
  \end{align}
  for any $x\in\RR$ and $\rho\in(-1,1)$. Then, by Theorem~\ref{thm:MN-tau-S}~(i) and \eqref{eq:Phi2-2}, 
  \begin{align*}
  \tau(-\rho,(0,b),F)
  =4\,\EE\,\Phi_2(0,bV;-\rho)-1
  &=4\,\EE\,\Phi(bV)-4\,\EE\,\Phi_2(0,bV;\rho)-1
  =4\,\EE\,\Phi(bV)-2-\tau(\rho,(0,b),F).
  \end{align*}
  Since $V=_d -V$ and $\Phi(x)+\Phi(-x)=1$, we have $\EE\,\Phi(bV)=1/2$. It then
  follows immediately that $\tau(-\rho,(0,b),F)=-\tau(\rho,(0,b),F)$. Similarly, we can
  show $\rho_S(-\rho,(0,b),F)=-\rho_S(\rho,(0,b),F)$ by using
  Theorem~\ref{thm:MN-tau-S}~(ii) and \eqref{eq:Phi2-2} and noticing that $V_2=_d
  -V_2$. This completes the proof of part~(ii).
  
  To prove part~(ii), it suffices to show that $\tau(\rho,(0,b),F)$ and
  $\rho_S(\rho,(0,b),F)$ are both strictly increasing in $b$ for $b>0$ and $\rho<0$,
  provided part~(i) of the proposition and the fact that $\tau$ and $\rho_S$ are
  invariant under the sign change of the skewness vector
  (Proposition~\ref{prop:MN-tau-S-properties}~(ii)). Taking derivatives with respect
  to $\beta_1$ on both sides of \eqref{eq:MN-tau} and \eqref{eq:MN-S} and using the
  partial derivative given by \eqref{eq:dPhi2} yields
  \[
    \frac{\partial}{\partial\beta_1}\tau(\rho,\beta,F)
    =
    4\,\EE\,V\phi(\beta_1V)\Phi\left(\frac{(\beta_2-\rho\beta_1)V}{\sqrt{1-\rho^2}}\right),
    \qquad
    \frac{\partial}{\partial\beta_1}\rho_S(\rho,\beta,F)
    =
    12\,\EE\,V_1\phi(\beta_1V_1)\Phi\left(\frac{\beta_2V_2-\rho\beta_1V_1V_3}{\sqrt{1-\rho^2
    V_3^2}}\right).
  \]
  Evaluated at $\beta_1=b$ and $\beta_2=0$, together with the definition of $\phi^{\mathrm s}$, we obtain: 
  \[
    \frac{\partial}{\partial b}\,\tau(\rho,(0,b),F)
    =
    2\,\EE\,V\phi^{\mathrm s}\left(bV;\,\alpha(\rho)\right),
    \qquad
    \frac{\partial}{\partial b}\,\rho_S(\rho,(0,b),F)
    =
    6\,\EE\,V_1\phi^{\mathrm s}\left(bV_1;\,\alpha(\rho V_3)\right),
  \]
  where $\alpha(r)\coloneqq -r/\sqrt{1-r^2}$ for any $r\in(-1,1)$. Since $\rho<0$ and $V_3>0$ almost surely, $\alpha(\rho)>0$ and $\alpha(\rho V_3)>0$ almost surely. Moreover, as both $V$ and $V_1$ are symmetrically distributed, Lemma~\ref{lemma:positive}~(i) applies so that both partial derivatives above are strictly positive. This proves part~(iii) of the proposition.
\end{proof}

\subsection{Proofs of the results in Section~\ref{sec:rank-corr-MSN}}
\label{sec:proof-sec5}

\smallskip

\begin{proof}[\textbf{\upshape Proof of Theorem~\ref{thm:MSN-tau-S}:}] 
  Since $X_1$ and $X_2$ have continuous marginal cdf's, their rank correlations
  depend solely on their copula. It therefore suffices to consider, without loss of
  generality, $X=(X_1,X_2)\sim\mathrm{MSN}(0,\varrho(\rho),\alpha,F)$, which admits
  the representation $X=\sqrt{W}Z$ for some
  $Z=(Z_1,Z_2)\sim\mathrm{SN}(0,\varrho(\rho),\alpha)$, $W\sim F$, and $W$ is
  independent of $Z$.
  
  To derive Kendall's tau of $X_1$ and $X_2$ stated in part (i) of the theorem, we start by defining $X^\star=(X_1^\star,X_2^\star)^\top\coloneqq\sqrt{W^\star}Z^\star$ where $W^\star$ and $Z^\star$ are independent copies of $W$ and $Z$, respectively, and $W^\star$ is independent of $Z^\star$. Then, by construction, $X^\star$ is an independent copy of $X$. Then, by Lemma~\ref{lemma:tau-rhoS-X}~(i), Kendall's tau of $X_1$ and $X_2$ is given by 
  \begin{align}
  \tau 
  = 4\,\PP\{X_1<X^\star_1,\,X_2<X^\star_2\}-1
  = 4\,\EE\,\PP\left\{\sqrt{W}Z_1 < \sqrt{W^\star}Z^\star_1,\,\sqrt
  WZ_2<\sqrt{W^\star}Z^\star_2 \,\big|\, W,W^\star\right\}-1.
  \label{eq:pf-thm2-1}
  \end{align}
  Applying Lemma~\ref{lemma:orthant-prob-tau} to $(Z_1,Z_2)$ and $(Z^\star_1,Z_2^\star)$, we have 
  \begin{align}\label{eq:pf-thm2-2}
  \PP
  \left\{\sqrt{W}Z_1<\sqrt{W^\star}Z^\star_1, \sqrt{W}Z_2<\sqrt{W^\star}Z^\star_2
  \,\big|\, W,W^\star\right\} 
  =
  4\,\Phi_4\big(0;\mathrm P_\tau(\rho,\delta,S)\big),
  \end{align}
  where $S = \big(\sqrt{W^\star/(W^\star+W)},-\sqrt{W/(W^\star+W)}\big)$. Part~(i) of
  the theorem then holds due to \eqref{eq:pf-thm2-1}, \eqref{eq:pf-thm2-2}, and the
  fact that $S =_d V$ for $V$ being defined in part~(i) of the theorem.
  
  To derive Spearman's rho of $X_1$ and $X_2$ given in part (ii) of the theorem, we
  define $(Z_1^\circ,Z_2^\circ)$ to be a bivariate skew-normal vector satisfying
  $(Z_1^\circ,Z_2^\circ)\perp(Z_1,Z_2)$, $Z_1^\circ\perp Z_2^\circ$ and
  $Z_i^\circ=_dZ_i$ for $i\in\{1,2\}$.
  Next, define
  $X_1^\circ\coloneqq\sqrt{W_1}Z_1^\circ$ and
  $X_2^\circ\coloneqq\sqrt{W_2}Z_2^\circ$, where $W_1,W_2\sim$ i.i.d.\,$F$ and are
  independent of $W$, $Z_1^\circ$ and $Z_2^\circ$. By construction,
  $(X_1^\circ,X_2^\circ)\perp(X_1,X_2)$, $X^\circ_1\perp X^\circ_2$, and
  $X_i^\circ=_dX_i$ for $i\in\{1,2\}$.
  By Lemma~\ref{lemma:tau-rhoS-X}~(ii), Spearman's rho of $X_1$ and $X_2$ is 
  \begin{equation} \label{eq:pf-thm2-3}
  \rho_S 
  = 12\,\PP\{X_1<X^\circ_1,\,X_2<X^\circ_2\}-3
  = 12\,\EE\,\PP\left\{\sqrt WZ_1 < \sqrt{W_1}Z^\circ_1, \sqrt WZ_2 <
  \sqrt{W_2}Z^\circ_2 \,\big|\, W,W_1,W_2\right\}-3.
  \end{equation}
  Applying Lemma~\ref{lemma:orthant-prob-S} to $(Z_1,Z_2)$ and $(Z^\circ_1,Z^\circ_2)$, we have
  \begin{align}\label{eq:pf-thm2-4}
  \PP\left\{\sqrt WZ_1 < \sqrt{W_1}Z^\circ_1, \sqrt WZ_2 < \sqrt{W_2}Z^\circ_2 
    \,\big|\, W,W_1,W_2\right\}
  =
  8\Phi_5\big(0;\mathrm P_S(\rho,\delta,T)\big),
  \end{align}
  where $T=_dV$, where $V$ is defined in part~(ii) of the theorem.
  Then, part~(ii) follows from \eqref{eq:pf-thm2-3} and \eqref{eq:pf-thm2-4}.
\end{proof}

\begin{proof}[\textbf{\upshape Proof of Corollary~\ref{cor:MSN-tau-S-Phi2}:}] 
  We prove part~(i) first. Define a four-dimensional random vector $X=(X_1,X_2,X_3,X_4)$ satisfying $X\given V\sim\mathrm N(0,\mathrm P_\tau(\rho,\delta,V))$ where $\mathrm P_\tau$ is defined by \eqref{eq:P-tau}. Then, we have
  \begin{equation}
  \Phi_4\big(0;\mathrm P_\tau(\rho,\delta(\rho,\alpha),V)\big) 
  = \PP\{X_1<0,X_2<0,X_3<0,X_4<0 \given V\}
  = \frac14\,\PP\{X_1<0,X_2<0\given X_3<0,X_4<0, V\}, \label{eq:Phi4-2}
  \end{equation}
  where in the second equality we used $\PP\{X_3<0,X_4<0\given V\}=(1/2)^2=1/4$,
  since $(X_3,X_4)\given V\sim\mathrm N(0,\mathrm I_2)$ and uncorrelated components
  are independent under multivariate normality. Consequently, by the
  original expression \eqref{eq:MSN-tau} for Kendall's tau and \eqref{eq:Phi4-2}, we
  have
  \begin{equation}\label{eq:tau-alternative}
  \tau(\rho,\alpha,F) = 4\,\EE\,\PP\{X_1<0,X_2<0\given X_3<0,X_4<0, V\} - 1.
  \end{equation}
  Using the conditioning property of the multivariate normal distribution and the identity $V_1^2+V_2^2=1$, we can deduce that $(X_1,X_2)\given \{X_3,X_4,V\} \sim \mathrm N(\mu, \Sigma)$, where 
  \[
  \mu =
  \begin{bmatrix} (V_1X_3+V_2X_4)\delta_1\\ (V_1X_3+V_2X_4)\delta_2 \end{bmatrix},
  \qquad
  \Sigma = 
  \begin{bmatrix}
    1-\delta_1^2 & \rho - \delta_1\delta_2 \\ \rho - \delta_1\delta_2 & 1-\delta_2^2
  \end{bmatrix}.
  \]
  Therefore, $\PP\{X_1<0,X_2<0\given X_3,X_4,V\}$ is equal to
  \begin{equation}\label{eq:tau-alternative-2}
  \Phi_2\big(-\alpha_1^\dagger(V_1X_3+V_2X_4)\,,\,-\alpha_2^\dagger(V_1X_3+V_2X_4)\,;\,\rho^\dagger\big),
  \end{equation}
  where $\alpha^\dagger_1,\alpha^\dagger_2$ and $\rho^\dagger$ are defined by \eqref{eq:alpha-dagger} and \eqref{eq:rho-dagger}. Since $(X_3,X_4)\given V$ is a pair of independent standard normal random variates, equation~\eqref{eq:MSN-tau-Phi2} follows from the definition of $(Y_1,Y_2)$, along with \eqref{eq:tau-alternative} and \eqref{eq:tau-alternative-2}.    
  
  Next, we prove part~(ii). Define $X=(X_1,X_2,X_3,X_4,X_5)$ satisfying $X\given V\sim\mathrm N(0,\mathrm P_S(\rho,\delta,V))$ where $\mathrm P_S$ is defined by \eqref{eq:P-S}. Then, we have
  \begin{align}
  \Phi_5\big(0;\mathrm P_S(\rho,\delta(\rho,\alpha),V)\big) 
  &= \PP\{X_1<0,X_2<0,X_3<0,X_4<0,X_5<0 \given V\}\notag\\
  &= \frac18\,\PP\{X_1<0,X_2<0\given X_3<0,X_4<0,X_5<0, V\}, \label{eq:Phi5-2}
  \end{align}
  where in the second equality we used $\PP\{X_3<0,X_4<0,X_5<0\given
  V\}=(1/2)^3=1/8$, since $(X_3,X_4,X_5)\given V\sim\mathrm N(0,\mathrm I_3)$ and
  uncorrelated components are independent under normality.
  Consequently, it follows from \eqref{eq:MSN-rhoS} and \eqref{eq:Phi5-2} that 
  \begin{equation}\label{eq:S-alternative}
  \rho_S(\rho,\alpha,F) = 12\,\EE\,\PP\{X_1<0,X_2<0\given X_3<0,X_4<0,X_5<0, V\} - 3.
  \end{equation}
  By the conditioning property of the multivariate normal and the identities
  $V_i^2+(V_i^-)^2=1$, for $i=1,2$, and $V_1^-V_2^-=V_3$, we obtain 
  $(X_1,X_2)\given \{X_3,X_4,X_5,V\} \sim \mathrm N(\mu, \Sigma)$, where 
  \[
  \mu = 
  \begin{bmatrix}
    (V_1X_3+V_1^-X_5)\delta_1\\ (V_2X_4+V_2^-X_5)\delta_2
  \end{bmatrix},
  \qquad
  \Sigma = 
  \begin{bmatrix}
    1-\delta_1^2 & V_3(\rho - \delta_1\delta_2) \\ 
    V_3(\rho - \delta_1\delta_2) & 1-\delta_2^2
  \end{bmatrix}.
  \]
  Therefore, $\PP\{X_1<0,X_2<0\given X_3,X_4,X_5,V\}$ is equal to
  \begin{equation}\label{eq:S-alternative-2}
    \Phi_2\big(-\alpha_1^\dagger(V_1X_3+V_1^-X_5)\,,\,-\alpha_2^\dagger(V_2X_4+V_2^-X_5)\,;\,\rho^\dagger
    V_3\big).
  \end{equation}
  Since $(X_3,X_4,X_5)\given V$ are mutually independent standard normal random
  variates, \eqref{eq:MSN-S-Phi2} follows from the definition of $(Y_1,Y_2,Y_3)$,
  along with \eqref{eq:S-alternative} and \eqref{eq:S-alternative-2}.
\end{proof}

\begin{proof}[\textbf{\upshape Proof of Proposition~\ref{prop:MSN-tau-S-properties}:}] 
  We will analyze the rank correlations using the expressions given in 
  Corollary~\ref{cor:MSN-tau-S-Phi2}.
  Part~(i) of the proposition follows directly
  from the exchangeability of bivariate normal cdf and the fact that 
  $\alpha^\dagger_1(\rho,(\alpha_1,\alpha_2))=\alpha^\dagger_2(\rho,(\alpha_2,\alpha_1))$
  and
  $\alpha^\dagger_2(\rho,(\alpha_1,\alpha_2))=\alpha^\dagger_1(\rho,(\alpha_2,\alpha_1))$,
  which can be inferred from \eqref{eq:alpha-dagger}.
  For part~(ii), note that $\tau(\rho,\alpha,F)=\tau(\rho,-\alpha,F)$ follows directly
  from the expression of $\tau$ in Corollary~\ref{cor:MSN-tau-S-Phi2}~(i) and $Z=_d-Z$.
  Moreover, by equation~(2) in Corollary~\ref{cor:MSN-tau-S-Phi2}~(ii), to show
  $\rho_S(\rho,\alpha,F)=\rho_S(\rho,-\alpha,F)$, it suffices to show
  \[
  \EE\,\Phi_2\big(\alpha_1^\dagger Z_1, \alpha_2^\dagger Z_2; \rho^\dagger V_3\big)
  =\EE\,\Phi_2(-\alpha_1^\dagger Z_1, -\alpha_2^\dagger Z_2; \rho^\dagger V_3),
  \]
  where the mixing variables $(Z_1,Z_2,V_3)$ are defined in
  Corollary~\ref{cor:MSN-tau-S-Phi2}~(ii). This equality can be established using the
  same proof strategy in the proof of
  Proposition~\ref{prop:MN-tau-S-properties}~(ii), along with the fact that
  $Z_i=_d-Z_i$, for $i=\{1,2\}$.
  
  For parts~(iii) and (iv), first note that $\tau$ or $\rho_S$ cannot attain the
  boundary values $-1$ or $1$ if $\rho\in(-1,1)$. This follows from the fact that, by
  construction, the skew-normal mixture copula is always absolutely continuous with 
  respect to the Lebesgue measure on $(0,1)^2$ for all $\rho\in(-1,1)$, whereas neither
  the Fr\'echet lower nor upper bound copula is absolutely continuous. Hence, the
  copula cannot coincide with either Fr\'echet bound when $\rho\in(-1,1)$, and
  therefore the rank correlations cannot reach $-1$ or $1$ when $\rho\in(-1,1)$.
  Moreover, since $\tau=\pm1$ if and only if $\rho_S=\pm1$, to establish parts~(iii)
  and (iv) it remains to show that 
  \[
    \tau(-1,\alpha,F) = -1,
    \qquad
    \tau(1,\alpha,F) = 1.
  \]
  We first observe that, for any $\alpha\in\RR^2$, 
  $\delta_1(-1,\alpha) = -\delta_2(-1,\alpha)$,
  and
  $\delta_1(1,\alpha)=\delta_2(1,\alpha)$,
  so that, by definition,
  $\alpha_1^\dagger=-\alpha_2^\dagger\coloneqq \alpha^\circ$
  and 
  $\rho^\dagger=-1$ when $\rho=-1$,
  and 
  $\alpha_1^\dagger=\alpha_2^\dagger\coloneqq \alpha_\circ$ and $\rho^\dagger=1$
  when $\rho=1$.
  Substituting $\rho=-1$ into equation~(1) in Corollary~\ref{cor:MSN-tau-S-Phi2}2~(i) yields
  $\tau(-1,\alpha,F) = 4\,\EE\,\Phi_2\big(\alpha^\circ Z,-\alpha^\circ Z;-1\big) -1$.
  Using the identity $\Phi_2(x,-x;-1)=\Phi(x)-\Phi(x)=0$ for all $x\in\RR$,
  we obtain that 
  $\tau(-1,\alpha,F) = 4\times0 - 1= -1$.
  For $\rho=1$, substituting $\rho=1$ into equation~(1) in Corollary~\ref{cor:MSN-tau-S-Phi2}2~(i)
  and applying Lemma~\ref{lemma:bi-normal-diag} for the case  $\rho=1$ gives
  $\tau(1,\alpha,F) 
  = 4\,\EE\,\Phi_2\big(\alpha_\circ Z,\alpha_\circ Z;1\big) -1 
  = 4\,\EE\,\Phi^{\mathrm{s}}(\alpha_\circ Z;0)-1
  = 4\,\EE\,\Phi(\alpha_\circ Z)-1$,
  where $Z=V_1Y_1+V_2Y_2$ as defined in Corollary~\ref{cor:MSN-tau-S-Phi2}~(i).
  Since $\alpha_\circ Z =_d -\alpha_\circ Z$ and 
  $\Phi(\alpha_\circ Z)+\Phi(-\alpha_\circ Z)=1$,
  we have 
  $\EE\,\Phi(\alpha_\circ Z) = 1/2$, and hence $\tau(1,\alpha,F) = 4\times1/2-1=1$. 
  This completes the proof of the proposition.
\end{proof}

\begin{proof}[\textbf{\upshape Proof of Corollary \ref{cor:SN-tau-S}:}] 
  Since the skew-normal copula is the special case of the skew-normal scale mixture
  copula when the mixing distribution $F$ is degenerated at $1$, substituting this
  degenerate distribution in the statements of Theorem~\ref{thm:MSN-tau-S} and
  Corollary~\ref{cor:MSN-tau-S-Phi2} immediately yields this corollary.
\end{proof}

\begin{proof}[\textbf{\upshape Proof of Proposition~\ref{prop:equi-skewness-MSN}:}] 
  The expressions of $\tau$ of $\rho_S$ stated in part~(i) and part~(ii) of the
  proposition follow directly from Corollary~\ref{cor:MSN-tau-S-Phi2} under the equi-skew setting. To
  analyze $\partial\tau/\partial\rho$, we deduce from \eqref{eq:dPhi2} that
  \[
  \frac{\partial}{\partial x}\Phi_2(x,x;\rho) 
  =
  2\phi(x)\Phi\left(\alpha_\rho\, x\right) = \phi^{\mathrm s}(x;\alpha_{\rho^\dagger}),
  \qquad
  \frac{\partial}{\partial \rho}\Phi_2(x,x;\rho) = \phi_2(x,x;\rho),
  \]
  where $\alpha_{\rho^\dagger}\coloneqq\sqrt{(1-\rho^\dagger)/(1+\rho^\dagger)}>0$. Then, we have
  \[
  \frac{\partial}{\partial\rho}\EE\,\Phi_2(a^\dagger Z, a^\dagger Z;\rho^\dagger)
  =\frac{\partial a^\dagger}{\partial \rho}\,\EE\,\phi^{\mathrm{s}}(a^\dagger Z;\alpha_{\rho^\dagger}) Z 
  +\frac{\partial \rho^\dagger}{\partial \rho}\,\EE\,\phi_2(a^\dagger Z,a^\dagger Z;\rho^\dagger).
  \]
  We note that (i) $\partial a^\dagger/\partial\rho>0$ when $a>0$, (ii)
  $\EE\,\phi^{\mathrm{s}}(a^\dagger Z;\alpha_{\rho^\dagger}) Z >0$ when $a>0$, by
  Lemma~\ref{lemma:positive}~(i) and the fact that $Z=_d-Z$, (iii)
  $\partial\rho^\dagger/\partial\rho>0$, and (iv) $\phi_2$ is a strictly positive
  function. Consequently, $\partial \EE\,\Phi_2(a^\dagger Z, a^\dagger
  Z;\rho^\dagger)/\partial \rho>0$, and hence $\partial
  \tau(\rho,a,F)/\partial\rho>0$. 
  Next, we prove $\partial\tau/\partial a<0$. Since $\partial\bar\delta/\partial a>0$ and $\bar\delta$ shares the same sign as $a$, to show $\partial\tau/\partial a<0$ it suffices to show 
  \begin{align}
  \frac{\partial}{\partial\bar\delta}\,\EE\,\Phi_4\big(0;\mathrm P_\tau(\rho,(\bar\delta,\bar\delta),V)\big) < 0 \quad\text{for}\quad \bar\delta>0.
  \label{eq:Phi4-derivative-delta-star}
  \end{align}
  By Lemma~\ref{lemma:Phi4-derivative-delta}, we have 
  $\partial\,\Phi_4\big(0;\mathrm P_\tau(\rho,(\bar\delta,\bar\delta),V)\big)/\partial\bar\delta
  =_d
  \pi^{-2} V_1\arcsin \tilde V(\bar\delta,\bar\delta)(1-V_1^2\bar\delta^2)^{-1/2}$. 
  Since $V_1>0$ and
  $\tilde V(\bar\delta,\bar\delta)
  =
  V_2(1-\rho)\bar\delta\,[(1-\bar\delta^2)(1-\rho^2-2V_1^2\bar\delta^2(1-\rho))]^{-1/2}<0$
  when $\bar\delta>0$, we can easily deduce \eqref{eq:Phi4-derivative-delta-star}
  upon verifying the conditions of interchanging expectation and differentiation.
  This completes the proof.
\end{proof}

\begin{proof}[\textbf{\upshape Proof of Proposition~\ref{prop:single-skew-MSN}:}] 
  We first note that the expressions of $\tau$ and $\rho_S$ stated in parts~(ii) and
  (iii) of the proposition follow directly from Corollary~\ref{cor:MSN-tau-S-Phi2} and the parameters given
  in \eqref{eq:ra-dagger-single} in the single-skew setting.
  Then, we prove part~(i) of the proposition. To show that
  $\tau(\rho,(0,a),F)$ is an odd functions of $\rho$, or
  $\tau(\rho,(0,a),F)+\tau(-\rho,(0,a),F)=0$, it suffices to show that (applying
  Corollary~\ref{cor:MSN-tau-S-Phi2}~(i))  
  \begin{equation}\label{eq:E1/2}
    \EE\left[\Phi_2(aZ,a\rho^\dagger Z;\rho^\dagger)
    +
    \Phi_2(aZ,-a\rho^\dagger Z;-\rho^\dagger)\right] 
    = 
    \frac12,
  \end{equation}
  where $\rho^\dagger$ is defined by \eqref{eq:ra-dagger-single}. Applying the
  representation of $\Phi_2(x_1,x_2;\rho)$ in \eqref{eq:Phi2-rep} and then using the
  identities $\phi_2(x_1,x_2;-\rho)=\phi_2(x_1,-x_2;\rho)$ and $\phi(x)+\phi(-x)=1$,
  we can deduce that the left-hand side of \eqref{eq:E1/2} is equal to
  $\EE\,\Phi(aZ)$. Since $\EE\,\Phi(aZ)=1/2$, as shown above,
  equation~\eqref{eq:E1/2} holds. Analogously, we can show that
  \[
  \EE\left[\Phi_2(aZ_1,a\rho^\dagger Z_2;\rho^\dagger V_3)
  +\Phi_2(aZ_1,-a\rho^\dagger Z_2;-\rho^\dagger V_3)\right] = \frac12,
  \]
  where $Z_1,Z_2$ and $V_3$ are defined in Corollary~\ref{cor:MSN-tau-S-Phi2}~(ii). Consequently,
  $\rho_S(\rho,(0,a),F)$ is an odd functions of $\rho$. 
  
  Provided that Kendall's tau is an odd function of $\rho$, to prove part~(ii) of the
  proposition, we focus on $\rho>0$, and show that $\partial\tau/\partial\rho>0$ and
  $\partial\tau/\partial a<0$, when $\rho>0$ and $a>0$. To analyze
  $\partial\tau/\partial\rho$, we first deduce from \eqref{eq:dPhi2} that
  \[
    \frac{\partial}{\partial x_2}\Phi_2(x_1,x_2;\rho) 
    = \phi(x_2)\Phi\left(\frac{x_1-\rho x_2}{\sqrt{1-\rho^2}}\right) 
    = \frac12\,\phi^{\mathrm s}\left(x_2;\,\frac{x_1/x_2-\rho}{\sqrt{1-\rho^2}}\right),
    \qquad
    \frac{\partial}{\partial \rho}\Phi_2(x_1,x_2;\rho) 
    = \phi_2(x_1,x_2;\rho).
  \]
  Then, we have
  \[
    \frac{\partial}{\partial\rho}\EE\,\Phi_2(aZ,a\rho^\dagger Z;\rho^\dagger)=
    \frac12\frac{\partial\rho^\dagger}{\partial\rho}
    \EE\left[a\,\phi^{\mathrm{s}}\left(a\rho^\dagger Z;\alpha_{\rho^\dagger}\right)Z 
    +2\phi_2(aZ,a\rho^\dagger Z;\rho^\dagger)
    \right],
  \]
  where 
  $\alpha_{\rho^\dagger}\coloneqq (1-\rho^{\dagger2})/\rho^\dagger\sqrt{1-\rho^{\dagger 2}}$.
  Since
  $\alpha_{\rho^\dagger}>0$ and $a\rho^\dagger>0$, when both $\rho$ and $a$ are
  strictly positive, it follows from Lemma~\ref{lemma:positive}~(i) that 
  $\EE \,\phi^{\mathrm{s}}\left(a\rho^\dagger Z;\alpha_{\rho^\dagger}\right)Z>0$.
  Along with
  $\partial\rho^\dagger/\partial\rho>0$ and $\phi_2(\,\cdot\,;\rho^\dagger)>0$,
  we deduce
  $\partial\EE\,\Phi_2(aZ,a\rho^\dagger Z;\rho^\dagger)/\partial\rho>0$ 
  and hence $\partial\tau/\partial\rho>0$, as desired to be shown.
  
  To analyze $\partial\tau/\partial a$, we first note that
  $\partial\delta_\circ/\partial a>0$ and
  $\mathrm{sgn}(\delta_\circ)=\mathrm{sgn}(a)$.
  Therefore, to show $\partial\tau/\partial a<0$ when $\rho>0$ and $a>0$, it suffices
  to prove that
  \begin{equation} \label{eq:Phi4-derivative-delta-circ}
    \frac{\partial}{\partial\delta_\circ}\EE\,\Phi_4\big(0;\mathrm
    P_\tau(\rho,(\delta_\circ,\delta_\circ\rho),V)\big)<0,
    \qquad
    \rho>0,\quad \delta_\circ>0,
  \end{equation}
  using the expression of Kendall's tau in Theorem~\ref{thm:MSN-tau-S}~(i). By
  Lemma~\ref{lemma:Phi4-derivative-delta}, we have
  $
    \partial\,\Phi_4\big(0;\mathrm P_\tau(\rho,(\delta_\circ,\delta_\circ\rho),V)\big)/\partial\delta_\circ
    =_d
    (2\pi^2)^{-1}\rho
    V_1\arcsin \tilde
    V(\delta_\circ\rho,\delta_\circ)[1-(V_1\delta_\circ\rho)^2]^{-1/2}
  $.
  Since $V_1>0$ and 
  $
    \tilde V(\delta_\circ\rho,\delta_\circ)
    =
    V_2\delta_\circ(1-\rho^2)\{(1-\rho^2\delta_\circ^2)
    [(1-\rho^2)(1-V_1^2\delta_\circ^2)]\}^{-1/2}<0
  $
  when $\delta_\circ>0$, we deduce \eqref{eq:Phi4-derivative-delta-circ} after
  verifying the conditions of interchanging expectation and differentiation. This
  completes the proof of part~(ii).
\end{proof}

\subsection{Proofs of the results in Section~\ref{sec:invertibility}}
\label{sec:proof-invertibility}

\smallskip

\begin{proof}[\textbf{\upshape Proof of Proposition~\ref{prop:J-equi}:}] 
  The expression of $J_{\mathrm{mn}}(\rho,b)$ in the first half of the proposition can
  be easily deduced from the formulas of $\tau$ and $\rho_S$ for normal location-scale
  mixture copulas under equi-skewness in Corollary~\ref{cor:equi-skewness-MN}, by
  applying \eqref{eq:dPhi2}. 
  As for $J_{\mathrm{msn}}(\rho,a)$ in the second half of the
  proposition, we begin with the formulas of $\tau$ and $\rho_S$ for skew-normal scale
  mixture copulas under equi-skewness in Proposition~\ref{prop:equi-skewness-MSN}.
  First, define $s\coloneqq 1+a^2(1-\rho^2)$ so that 
  $a^\dagger = a(1+\rho)s^{-1/2}$ and $\rho^\dagger = (1+\rho)s^{-1}-1$.
  Given $\partial s/\partial a = 2a(1-\rho^2)$ and $\partial s/\partial\rho = -2\rho a^2$,
  we can easily verify that 
  \[
  \begin{bmatrix}
  \partial\rho^\dagger/\partial \rho & \partial\rho^\dagger/\partial a \\
  \partial a^\dagger/\partial\rho    & \partial a^\dagger/\partial a
  \end{bmatrix} = B,
  \]
  where $B$ is the matrix stated in the proposition.
  Moreover, note that
  \begin{align*}
    \frac{\partial}{\partial\rho} \Phi_2(a^\dagger Z, a^\dagger Z; \rho^\dagger)
    &= \phi_2(a^\dagger Z, a^\dagger Z;\rho^\dagger)
    \frac{\partial\rho^\dagger}{\partial \rho} 
    +  Z \phi^{\mathrm s}(a^\dagger Z;\alpha_{\rho^\dagger})
    \frac{\partial a^\dagger}{\partial \rho},
    \\
    \frac{\partial}{\partial a} \Phi_2(a^\dagger Z, a^\dagger Z; \rho^\dagger)
    &= \phi_2(a^\dagger Z, a^\dagger Z;\rho^\dagger)\frac{\partial\rho^\dagger}{\partial a} 
    +  Z \phi^{\mathrm s}(a^\dagger Z;\alpha_{\rho^\dagger})
    \frac{\partial a^\dagger}{\partial a},
    \\
    \frac{\partial}{\partial\rho} \Phi_2(a^\dagger Z_1, a^\dagger Z_2; \rho^\dagger
    V_3)
    &= V_3 \phi_2(a^\dagger Z_1, a^\dagger Z_2;\rho^\dagger V_3)
    \frac{\partial\rho^\dagger}{\partial \rho} 
    + 
    \left[
      g(Z_1,Z_2,V_3;\rho^\dagger, a^\dagger) + g(Z_2,Z_1,V_3;\rho^\dagger, a^\dagger)
    \right]
    \frac{\partial a^\dagger}{\partial \rho}, 
    \\
    \frac{\partial}{\partial a} \Phi_2(a^\dagger Z_1, a^\dagger Z_2; \rho^\dagger V_3)
    &= V_3 \phi_2(a^\dagger Z_1, a^\dagger Z_2;\rho^\dagger V_3)
    \frac{\partial\rho^\dagger}{\partial a} 
    + 
    \left[
      g(Z_1,Z_2,V_3;\rho^\dagger, a^\dagger) + g(Z_2,Z_1,V_3;\rho^\dagger, a^\dagger)
    \right]
    \frac{\partial a^\dagger}{\partial a},
  \end{align*}
  where $\alpha_{\rho^\dagger}$ and $g$ are defined in the proposition. 
  The expression of $J_{\mathrm{msn}}(\rho,a)$ then follows immediately.
\end{proof}

\begin{proof}[\textbf{\upshape Proof of Proposition~\ref{prop:J-single}:}] 
  The expression of $J_{\mathrm{mn}}^\circ(\rho,b)$ in the first half of the proposition can
  be easily deduced from the formulas of $\tau$ and $\rho_S$ for normal location-scale
  mixture copulas under single-skewness in Proposition~\ref{prop:single-skew-MN}, by applying
  \eqref{eq:dPhi2}. 
  To show the expression of $J_{\mathrm{msn}}^\circ(\rho,a)$, 
  we will use the formulas of $\tau$ and $\rho_S$ for skew-normal scale
  mixture copulas under single-skewness are given in Proposition~\ref{prop:single-skew-MSN}.
  Define $s\coloneqq 1+a^2(1-\rho^2)$ so that $\rho^\dagger = \rho/\sqrt{s}$.
  Given $\partial s/\partial a = 2a(1-\rho^2)$ and $\partial s/\partial\rho = -2\rho a^2$,
  we have $\partial\rho^\dagger/\partial\rho = \gamma_1$ and  
  $\partial\rho^\dagger/\partial a= \gamma_2$, where $\gamma_1$ and $\gamma_2$ are
  defined in the proposition. Moreover, note that
  \begin{align*}
    \frac{\partial}{\partial\rho} \Phi_2(aZ, a\rho^\dagger Z; \rho^\dagger)
    &= 
    h(Z;a,\rho^\dagger) \frac{\partial\rho^\dagger}{\partial \rho},
    \quad
    \frac{\partial}{\partial a} \Phi_2(aZ, a\rho^\dagger Z; \rho^\dagger)
    =
    c(Z;a,\rho^\dagger)
    +
    h(Z;a,\rho^\dagger) \frac{\partial\rho^\dagger}{\partial a},
    \\
    \frac{\partial}{\partial\rho} 
    \Phi_2(aZ_1, a\rho^\dagger Z_2; \rho^\dagger V_3)
    &=
    \tilde h(Z_1,Z_2,V_3;a,\rho^\dagger) 
    \frac{\partial\rho^\dagger}{\partial \rho}, 
    \\
    \frac{\partial}{\partial a}
    \Phi_2(aZ_1, a\rho^\dagger Z_2; \rho^\dagger V_3)
    &=
    \tilde c(Z_1,Z_2,V_3;a,\rho^\dagger) 
    +
    \tilde h(Z_1,Z_2,V_3;a,\rho^\dagger) \frac{\partial\rho^\dagger}{\partial a},
  \end{align*}
  where $h, c, \tilde h$ and $\tilde c$ are defined in the proposition.
  The expression of $J^\circ_{\mathrm{msn}}(\rho,a)$ then follows immediately.
\end{proof}

\bibliographystyle{myjmva}
\bibliography{rank-corr-SE}

@article{AasHaff2006,
  title={The Generalized Hyperbolic Skew {S}tudent’s $t$-Distribution},
  author={Aas, Kjersti and Haff, Ingrid Hob{\ae}k},
  journal={Journal of Financial Econometrics},
  volume={4},
  number={2},
  pages={275--309},
  year={2006},
  publisher={Oxford University Press}
}

@book{AbramowitzStegun1965,
  title={Handbook of Mathematical Functions: with Formulas, Graphs, and Mathematical Tables},
  author={Abramowitz, Milton and Stegun, Irene A},
  year={1965},
  publisher={Dover: New York}
}

@book{Azzalini2014,
  title={The Skew-Normal and Related Families},
  author={Azzalini, Adelchi},
  year={2014},
  publisher={Cambridge University Press}
}

@article{AzzaliniCapitanio2003,
  title={Distributions Generated by Perturbation of Symmetry with Emphasis on a Multivariate Skew t-Distribution},
  author={Azzalini, Adelchi and Capitanio, Antonella},
  journal={Journal of the Royal Statistical Society Series B: Statistical Methodology},
  volume={65},
  number={2},
  pages={367--389},
  year={2003},
  publisher={Oxford University Press}
}

@article{Barndorff1977,
  title={Exponentially Decreasing Distributions for the logarithm of Particle Size},
  author={Barndorff-Nielsen, Ole},
  journal={Proceedings of the Royal Society of London. A. Mathematical and Physical Sciences},
  volume={353},
  number={1674},
  pages={401--419},
  year={1977},
  publisher={The Royal Society London}
}

@article{BarndorffHalgreen1977,
  title={Infinite Divisibility of the Hyperbolic and Generalized Inverse {G}aussian Distributions},
  author={Barndorff-Nielsen, Ole and Halgreen, Christian},
  journal={Zeitschrift f{\"u}r Wahrscheinlichkeitstheorie und verwandte Gebiete},
  volume={38},
  number={4},
  pages={309--311},
  year={1977},
  publisher={Springer}
}

@incollection{BlaesildJensen1981,
  author      = {Blaesild, Preben and Jensen, J Ledet},
  title       = {Multivariate Distributions of Hyperbolic Type},
  booktitle   = {Statistical Distributions in Scientific Work},
  year        = {1981},
  pages       = {45--66},
  publisher   = {Springer Netherlands}
}

@article{BrancoDey2001,
  title={A General Class of Multivariate Skew-Elliptical Distributions},
  author={Branco, M{\'a}rcia D and Dey, Dipak K},
  journal={Journal of Multivariate Analysis},
  volume={79},
  number={1},
  pages={99--113},
  year={2001},
  publisher={Elsevier}
}

@article{Christoffersen2012,
  title={Is the Potential for International Diversification Disappearing? {A} Dynamic Copula Approach},
  author={Christoffersen, Peter and Errunza, Vihang and Jacobs, Kris and Langlois, Hugues},
  journal={The Review of Financial Studies},
  volume={25},
  number={12},
  pages={3711--3751},
  year={2012},
  publisher={Society for Financial Studies}
}

@article{CrealTsay2015,
  title={High Dimensional Dynamic Stochastic Copula Models},
  author={Creal, Drew D and Tsay, Ruey S},
  journal={Journal of Econometrics},
  volume={189},
  number={2},
  pages={335--345},
  year={2015},
  publisher={Elsevier}
}

@article{DemartaMcNeil2005,
  title={The $t$ Copula and Related Copulas},
  author={Demarta, Stefano and McNeil, Alexander J},
  journal={International Statistical Review},
  volume={73},
  number={1},
  pages={111--129},
  year={2005},
  publisher={Wiley Online Library}
}

@article{DengSmithManeesoonthorn2024,
  title={Large Skew-$t$ Copula Models and Asymmetric Dependence in Intraday Equity Returns},
  author={Deng, Lin and Smith, Michael Stanley and Maneesoonthorn, Worapree},
  journal={Journal of Business \& Economic Statistics},
  pages={1--17},
  year={2024},
  publisher={Taylor \& Francis}
}

@incollection{EmbrechtsMcNeilStraumann2002,
  title={Correlation and Dependence in Risk Management: Properties and Pitfalls},
  author={Embrechts, Paul and McNeil, Alexander and Straumann, Daniel},
  booktitle={Risk Management: Value at Risk and Beyond},
  pages={176--223},
  year={2002},
  publisher={Cambridge University Press}
}

@article{Fang2002meta,
  title={The Meta-Elliptical Distributions with Given Marginals},
  author={Fang, Hong Bin and Fang, Kai Tai},
  journal={Journal of Multivariate Analysis},
  volume={82},
  number={1},
  pages={1--16},
  year={2002},
  publisher={Elsevier}
}

@book{Genton2004,
  title={Skew-Elliptical Distributions and Their Applications: A Journey Beyond Normality},
  author={Genton, Marc G},
  year={2004},
  publisher={Chapman \& Hall/CRC}
}

@article{Gupta2003,
  title={Multivariate Skew $t$-Distribution},
  author={Gupta, A. K.},
  journal={Statistics},
  volume={37},
  number={4},
  pages={359--363},
  year={2003},
  publisher={Taylor \& Francis}
}

@article{HeinenValdesogo2020,
  title={Spearman Rank Correlation of the Bivariate {S}tudent $t$ and Scale Mixtures of Normal Distributions},
  author={Heinen, Andr{\'e}as and Valdesogo, Alfonso},
  journal={Journal of Multivariate Analysis},
  volume={179},
  pages={104650},
  year={2020},
  publisher={Elsevier}
}

@article{HeinenValdesogo2022,
  title={The {K}endall and {S}pearman Rank Correlations of the Bivariate Skew Normal Distribution},
  author={Heinen, Andr{\'e}as and Valdesogo, Alfonso},
  journal={Scandinavian Journal of Statistics},
  volume={49},
  number={4},
  pages={1669--1698},
  year={2022},
  publisher={Wiley Online Library}
}

@article{KluppelbergKuhn2009,
  title={Copula Structure Analysis},
  author={Kl{\"u}ppelberg, Claudia and Kuhn, Gabriel},
  journal={Journal of the Royal Statistical Society Series B: Statistical Methodology},
  volume={71},
  number={3},
  pages={737--753},
  year={2009},
  publisher={Oxford University Press}
}

@incollection{KolloPettere2010,
  title={Parameter Estimation and Application of the Multivariate Skew $t$-copula},
  author={Kollo, Tonu and Pettere, Gaida},
  booktitle={Copula {T}heory and {I}ts {A}pplications},
  pages={289--298},
  year={2010},
  publisher={Springer}
}

@article{Kropac1982,
  title={Some Properties and Applications of Probability Distributions based on
{M}ac{D}onald Function},
  author={Krop{\'a}{\v{c}}, Old{\v{r}}ich},
  journal={Aplikace matematiky},
  volume={27},
  number={4},
  pages={285--302},
  year={1982},
  publisher={Institute of Mathematics, Academy of Sciences of the Czech Republic}
}

@incollection{Lindskog2003kendall,
  title={Kendall’s tau for Elliptical Distributions},
  author={Lindskog, Filip and McNeil, Alexander and Schmock, Uwe},
  booktitle={Credit risk: Measurement, Evaluation and Management},
  pages={149--156},
  year={2003},
  publisher={Springer}
}

@article{LucasSchwaabZhang2017,
  title={Modeling Financial Sector Joint Tail Risk in the {E}uro Area},
  author={Lucas, Andr\'{e} and Schwaab, Bernd and Zhang, Xin},
  journal={Journal of Applied Econometrics},
  volume={32},
  pages={171--191},
  year={2017}
}

@book{McNeilFreyEmbrechts2005,
  title={Quantitative Risk Management: Concepts, Techniques and Tools},
  author={McNeil, Alexander J and Frey, R{\"u}diger and Embrechts, Paul},
  year={2005},
  publisher={Princeton University Press}
}

@book{Nelsen2006,
  title={An Introduction to Copulas},
  author={Nelsen, Roger B},
  year={2006},
  publisher={Springer}
}

@article{OhPatton2023,
  title={Dynamic Factor Copula Models with Estimated Cluster Assignments},
  author={Oh, Dong Hwan and Patton, Andrew J},
  journal={Journal of Econometrics},
  volume={237},
  number={2},
  pages={105374},
  year={2023},
  publisher={Elsevier}
}

@article{Owen1956,
  title={Tables for Computing Bivariate Normal Probabilities},
  author={Owen, Donald B},
  journal={The Annals of Mathematical Statistics},
  volume={27},
  number={4},
  pages={1075--1090},
  year={1956},
  publisher={JSTOR}
}

@article{Plackett1954,
  title={A Reduction Formula for Normal Multivariate Integrals},
  author={Plackett, Robin L},
  journal={Biometrika},
  volume={41},
  number={3/4},
  pages={351--360},
  year={1954},
  publisher={JSTOR}
}

@article{Yoshiba2018,
  title={Maximum Likelihood Estimation of Skew-$t$ Copulas with Its Applications to
  Stock Returns},
  author={Yoshiba, Toshinao},
  journal={Journal of Statistical Computation and Simulation},
  volume={88},
  number={13},
  pages={2489--2506},
  year={2018},
  publisher={Taylor \& Francis}
}

\end{document}